\documentclass[aps,%
twocolumn,%
floatfix,%
pra,%
amsfonts,%
groupedaddress,%
superscriptaddress]{revtex4-2}%
\usepackage[utf8]{inputenc}
\usepackage[T1]{fontenc}
\usepackage{mathrsfs}       
\usepackage{mathtools}
\usepackage{amsmath}
\usepackage{amsfonts}       
\usepackage[titletoc,title]{appendix}
\usepackage{algorithm,algorithmic}

\usepackage{bbm}            
\usepackage{bm}             
\usepackage{amsthm}         
\usepackage{makecell}
\usepackage[dvipsnames]{xcolor}
\definecolor{beamer@blendedblue}{rgb}{0.2,0.2,0.7}
\usepackage{times}
\usepackage{caption}
\usepackage{subfigure}
\usepackage{booktabs} 
\usepackage{multirow} 

\usepackage[colorlinks=true,%
linkcolor=beamer@blendedblue,%
bookmarks=true,%
breaklinks=true,%
filecolor=beamer@blendedblue,%
anchorcolor=yellow,%
citecolor=beamer@blendedblue,%
urlcolor=beamer@blendedblue,%
pdfauthor={Congcong Zheng, Xutao Yu, Zaichen Zhang, Ping Xu, and Kun Wang},%
pdfsubject={Efficient Verification of Stabilizer Code Subspaces with Local Measurements},%
CJKbookmarks=true]{hyperref}

\usepackage{float}    
\usepackage{pifont}

\newtheorem{definition}{Definition}

\newtheorem{lemma}[definition]{Lemma}
\newtheorem{theorem}[definition]{Theorem}

\mathchardef\ordinarycolon\mathcode`\:
\mathcode`\:=\string"8000
\def\vcentcolon{\mathrel{\mathop\ordinarycolon}}
\begingroup \catcode`\:=\active
  \lowercase{\endgroup
  \let :\vcentcolon
  }

\DeclareFontFamily{U}{mathx}{\hyphenchar\font45}
\DeclareFontShape{U}{mathx}{m}{n}{<-> mathx10}{}
\DeclareSymbolFont{mathx}{U}{mathx}{m}{n}
\DeclareMathAccent{\widebar}{0}{mathx}{"73}

\newcommand{\wh}[1]{\widehat{#1}}

\newcommand{\ket}[1]{\vert{#1}\rangle}
\newcommand{\bra}[1]{\langle{#1}\vert}
\newcommand{\ketbra}[1]{\vert{#1}\rangle\!\langle{#1}\vert}
\newcommand{\braket}[2]{\langle #1\vert #2\rangle}
\newcommand\proj[1]{\vert{#1}\rangle\!\langle{#1}\vert}

\newcommand{\density}[1]{\mathscr{D}(#1)}

\DeclareMathOperator{\tr}{Tr}  
\newcommand{\1}{\mathbbm{1}}

\DeclareMathOperator{\rank}{rank}
\newcommand{\spn}[1]{{\rm span}\{#1\}}

\makeatletter
\newsavebox{\@brx}
\newcommand{\llangle}[1][]{\savebox{\@brx}{\(\m@th{#1\langle}\)}%
  \mathopen{\copy\@brx\kern-0.5\wd\@brx\usebox{\@brx}}}
\newcommand{\rrangle}[1][]{\savebox{\@brx}{\(\m@th{#1\rangle}\)}%
  \mathclose{\copy\@brx\kern-0.5\wd\@brx\usebox{\@brx}}}
\makeatother

\newcommand*{\cD}{\mathcal{D}}

\newcommand*{\cG}{\mathcal{G}}
\newcommand*{\cH}{\mathcal{H}}

\newcommand*{\cK}{\mathcal{K}}

\newcommand*{\cM}{\mathcal{M}}

\newcommand*{\cP}{\mathcal{P}}

\newcommand*{\cS}{\mathcal{S}}

\newcommand*{\cV}{\mathcal{V}}

\newcommand{\bZ}{\mathbb{Z}}

\definecolor{wildstrawberry}{rgb}{1.0, 0.26, 0.64}
\definecolor{googleblue}{HTML}{4285F4}
\definecolor{googlered}{HTML}{DB4437}
\definecolor{googleyellow}{HTML}{F4B400}
\definecolor{googlegreen}{HTML}{0F9D58}
\definecolor{klevinblue}{HTML}{002FA7}
\definecolor{tiffanyblue}{HTML}{0ABAB5}

\newcommand{\red}[1]{\textcolor{googlered}{#1}}
\newcommand{\blue}[1]{\textcolor{googleblue}{#1}}
\newcommand{\green}[1]{\textcolor{googlegreen}{#1}}
\newcommand{\yellow}[1]{\textcolor{googleyellow}{#1}}

\newcommand\prlsection[1]{\textit{\textbf{#1}}---}

\begin{document}

\newcommand{\thetitle}{{Efficient Verification of Stabilizer Code Subspaces with Local Measurements}}
\title{\thetitle}
\author{Congcong Zheng}%
\affiliation{State Key Lab of Millimeter Waves, Southeast University, Nanjing 211189, China}%
\affiliation{Frontiers Science Center for Mobile Information Communication and Security, Southeast University, Nanjing 210096, People's Republic of China}%

\author{Xutao Yu}%
\affiliation{State Key Lab of Millimeter Waves, Southeast University, Nanjing 211189, China}%
\affiliation{Frontiers Science Center for Mobile Information Communication and Security, Southeast University, Nanjing 210096, People's Republic of China}%
\affiliation{Purple Mountain Laboratories, Nanjing 211111, People's Republic of China}%

\author{Zaichen Zhang}%
\thanks{Corresponding author: \href{zczhang@seu.edu.cn}{
zczhang@seu.edu.cn}}%
\affiliation{Frontiers Science Center for Mobile Information Communication and Security, Southeast University, Nanjing 210096, People's Republic of China}%
\affiliation{Purple Mountain Laboratories, Nanjing 211111, People's Republic of China}%
\affiliation{National Mobile Communications Research Laboratory, Southeast University, Nanjing 210096, China.}%

\author{Ping Xu}%
\affiliation{Institute for Quantum Information \& State Key Laboratory of High Performance Computing,%
College of Computer Science and Technology, National University of Defense Technology,%
Changsha 410073, China}%
\affiliation{Hefei National Laboratory, Hefei 230088, China}%

\author{Kun Wang}
\thanks{Corresponding author: \href{nju.wangkun@gmail.com}{nju.wangkun@gmail.com}}%
\affiliation{Institute for Quantum Information \& State Key Laboratory of High Performance Computing,%
College of Computer Science and Technology, National University of Defense Technology,%
Changsha 410073, China}%

\begin{abstract}
We address the task of verifying whether a quantum computer, 
designed to be protected by a specific stabilizer code, 
correctly encodes the corresponding logical qubits. 
To achieve this, we develop a general framework for subspace verification and 
explore several stabilizer code subspaces of practical significance.
First, we present two efficient verification strategies for general stabilizer code subspaces, 
utilizing measurements of their stabilizer generators and stabilizer groups, respectively. 
Then, building on the observation that certain tests can be conducted in parallel when the subspace exhibits specific structural properties, 
we propose a coloring strategy tailored to graph code subspaces and 
an XZ strategy tailored to Calderbank-Shor-Steane (CSS) code subspaces. 
Compared to stabilizer-based strategies, these new strategies require significantly fewer measurement settings and consume fewer state copies, approaching near-global optimality. 
Notably, all the strategies employ a limited number of Pauli measurements, are non-adaptive, 
and work on mixed states,
enabling efficient experimental certification of both logical qubits and logical operations in noisy quantum computers. 
This work contributes to the first systematic study of efficient verification of stabilizer code subspaces with local measurements. 
\end{abstract}
\date{\today}
\maketitle

\section{Introduction}

Current quantum systems often fail to work as desired due to the presence of quantum noise, 
making it crucial to accurately describe the actual quantum system and correct errors in quantum information processing.
Quantum tomography is the standard approach for characterizing the entire quantum system~\cite{PhysRevLett.129.133601,PhysRevA.105.032427}, 
but it is highly resource-intensive and thus impractical for large-scale systems. 
In many practical applications, however, full characterization is unnecessary. 
Consequently, various resource-efficient methods have been developed to certify quantum systems~\cite{eisert2020quantum, kliesch2021theory}, 
including fidelity estimation~\cite{flammia2011direct, huang2020predicting, elben2020crossplatforma, zhu2022crossplatform, zheng2024crossplatform} and  
entanglement detection~\cite{zhou2020singlecopies, zhu2019optimal, brandao2004separable, elben2020mixedstate}. 
Among these methods, quantum state verification
~\cite{pallister2018optimala,wang2019optimala,yu2022statisticala,li2020optimal,dangniam2020optimal,chen2023memory,li2021verification}
is designed to verify whether quantum states are prepared as desired 
and has been experimentally validated~\cite{jiang2020towards,Zhang_2020,Xia_2022}.
Specifically, the verification strategies focus on using local operators 
and classical communication (LOCC) to verify entangled states. 

\begin{figure}[!htbp]
    \centering
    \includegraphics[width=0.9\linewidth]{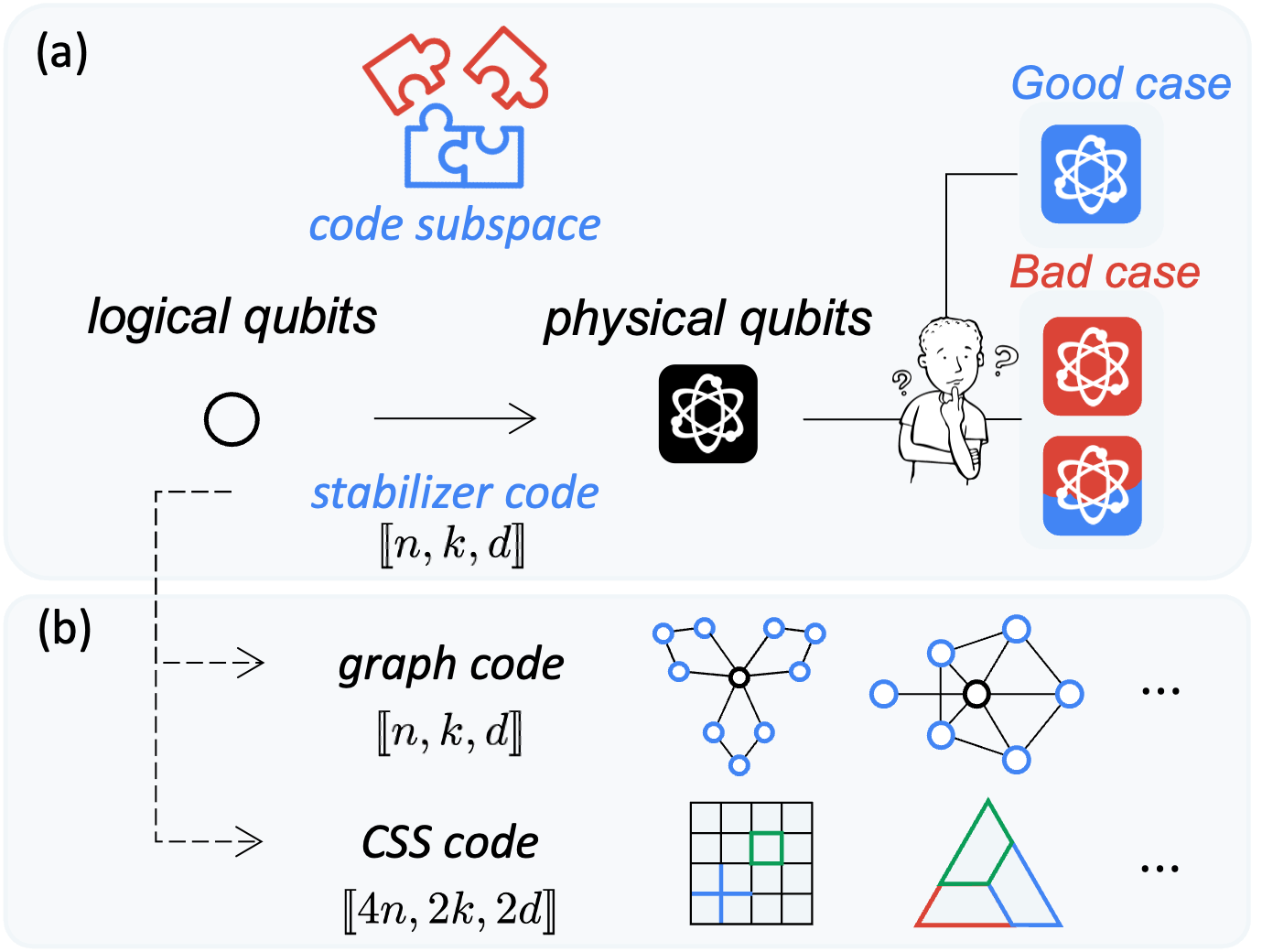}
    \caption{\raggedright 
    (a) Quantum code subspace verification. 
    A quantum code encodes logical information into a large code subspace, depicted as blue puzzles, while the red puzzles represent the complementary subspace. 
    Given states, either after encoding or after performing logical operators, 
    the goal of subspace verification is to distinguish between two cases: 
    (i) \textbf{Good case}: The states are in the target code subspace; 
    (ii) \textbf{Bad case}: The states lie outside the target code subspace, indicating an unreliable encoding process or the occurrence of a correctable error during computation. 
    (b) The stabilizer codes hierarchy under consideration. 
    An $[\![n,k,d]\!]$ stabilizer code is locally equivalent to an $[\![n,k,d]\!]$ graph code with a well-defined graph structure. 
    Furthermore, an $[\![n,k,d]\!]$ stabilizer code can be mapped onto a $[\![4n,2k,2d]\!]$ Calderbank-Shor-Steane (CSS) code
    possessing elegant structure, e.g., the toric code and the dual-containing code. 
    The left figure on the last line illustrates a $[\![32,2,4]\!]$ toric code, while the right figure shows the $[\![7,1,3]\!]$ Steane code, which belongs to the dual-containing code family. 
    }
    \label{fig:task description}
\end{figure}

On the other hand, Quantum Error Correction (QEC) is crucial for fault-tolerant quantum computation~\cite{gottesman1997stabilizer,shor1995scheme,1996multipleparticle,calderbank1996good,kitaev2003faulttoleranta,chen2021exponential,krinner2022realizing,takeda2022quantum,acharya2023suppressing,ni2023beating,bluvstein2024logical}, as it protects quantum information from noise. 
QEC achieves this by encoding logical qubits into physical qubits and constructing a corresponding code space, 
as illustrated in Fig.~\ref{fig:task description}(a). 
Intuitively, ensuring the effectiveness of the QEC process—including both encoding logical states and performing logical operations—is crucial for achieving fault tolerance in noisy, real-world quantum processors~\cite{ryan-anderson2021realizationa,rodriguez-blanco2021efficient,miao2023overcoming,zhang2012experimental}. 
A critical challenge arises: How can we certify the reliability of the encoding process in actual quantum devices? Can we verify whether errors have occurred during logical operations by certifying the corresponding code subspace rather than diagnosing syndromes? 
Recently, Baccari \textit{et al.}~\cite{baccari2020deviceindependent} partially addressed this certification problem 
by presenting the first self-testing protocols for the five-qubit and toric code subspaces.

In this work, we move from quantum state verification to \emph{quantum subspace verification}, 
aiming to determine whether a prepared quantum state belongs to a certain entangled subspace using local measurements.
Since entangled subspaces exhibit much more complex entanglement structures than entangled states, this task is highly non-trivial.
Note that quantum subspace verification has been previously mentioned in~\cite{zhu2024efficient, chen2023efficient}, 
where strategies were designed to verify ground states of local Hamiltonians. 
Here, we focus on the stabilizer code subspaces and the main contributions are summarized in Fig.~\ref{fig:task description}.
First, we establish a general framework of subspace verification. 
Then, we propose two efficient verification strategies for general stabilizer code subspaces. 
To further enhance verification efficiency, we examine two large classes of stabilizer codes: 
\emph{graph} codes and \emph{Calderbank-Shor-Steane} (CSS) codes, which are of practical interests.
For graph codes, we develop a new graph structure and introduce a coloring strategy for verifying the corresponding subspace. 
This coloring strategy requires fewer measurement settings than the aforementioned strategies. 
For CSS codes, we design an XZ strategy that requires only two measurement settings. 
Remarkably, all the proposed strategies employ a limited number of Pauli measurements, are non-adaptive, 
and work generally on mixed states,
enabling efficient experimental certification of logical qubits and logical operations in quantum computers.

\section{General framework of subspace verification}

Let $\cV:=\spn{\ket{\psi_j}}$ be the subspace spanned 
by a set of orthonormal states $\{\ket{\psi_j}\}_j$, 
where $\ket{\psi_j}$ lies in an $n$-partite Hilbert space $\cH=\otimes_{\ell=1}^n\cH_\ell$. 
By definition, $\cV\subseteq\cH$. 
Let $\density{\cV}$ be the set of density operators acting on $\cV$
and $\Pi:=\sum_j\proj{\psi_j}$ be the projector onto $\cV$. 
It should be noted that $\sigma\in\density{\cV}$ if and only if $\tr[\Pi\sigma] = 1$, 
as proved in Appendix~\ref{app:proof of subspace verification}. 
We can now formally define the quantum subspace verification task:
Given a quantum computer $\cD$ designed to produce states in $\cV$ and $N$ copies of states $\sigma_1, \sigma_2, \ldots, \sigma_N$ generated by $\cD$, 
the objective is to distinguish between the following two cases: 
(i) \textbf{Good}: for all $i\in[N]$, $\tr[\Pi\sigma_i] = 1$;
(ii) \textbf{Bad}:  for all $i\in[N]$, $\tr[\Pi\sigma_i] \leq 1 - \epsilon$ for some fixed $\epsilon$. 
In the following, we discuss how to complete this task. 

Suppose that we have access to a set of POVM elements $\cM$. 
Define a probability mass $\mu: \cM \to [0, 1]$, $\sum \mu(M) = 1$. 
For each state preparation, we pick a POVM element $M\in\cM$ with probability $\mu(M)$ 
and consider the corresponding two-outcomes POVMs $\{M, \1-M\}$, 
where $M$ has output ``pass'' and $\1-M$ has output ``fail''. We call $M$ a \textit{test operator}. 
The expected probability of a randomly generated quantum state $\sigma$ passing the test reads 
\begin{align}
    \Pr\left\{\text{``pass''}\vert \sigma\right\} 
    = \sum_{M\in\cM} \mu(M) \tr[M \sigma]
    = \tr[\Omega \sigma], 
\end{align}
where the \emph{verification operator} of this strategy is defined as 
\begin{align}\label{eq:verification-operator}
    \Omega := \sum_{M\in\cM} \mu(M) M. 
\end{align}

To satisfy the requirement of the verification task,
we impose two conditions on the verification operator $\Omega$: \emph{perfect completeness} condition and \emph{soundness} condition.
The perfect completeness condition requires that
$\forall\sigma\in\density{\cV}, \tr[\Omega \sigma] = 1$.
Intuitively, this condition guarantees that states in the target subspace $\cV$ can always pass the test.
This condition can be equivalently characterized using the projector $\Pi$ associated with $\cV$ as 
$\tr[\Omega\Pi]=\rank(\Pi)$, where $\rank(X)$ denotes the rank of the operator $X$.  
Detailed arguments and proofs is given in Appendix~\ref{app:proof of subspace verification}. 
Then, let's consider the soundness condition. 
We show that the worst-case passing probability $p(\Omega)$, defined as 
\begin{align}
    p(\Omega) := \max_{\sigma:\tr[\Pi\sigma]\leq 1-\epsilon}\Pr\{\text{``pass''}|\sigma\}, 
\end{align}
in the \textbf{Bad} case is uniquely 
determined by the largest eigenvalue of the projected effective verification operator,
as elucidated in the following theorem. 
\begin{theorem}
\label{theorem:worst-case passing probability}
It holds that
\begin{align}
    p(\Omega) 
= \max_{\sigma:\tr[\Pi\sigma]\leq1-\epsilon}\tr[\Omega\sigma]
= 1 - (1-\lambda_{\max}(\wh{\Omega}))\epsilon,
\end{align}
where $\wh{\Omega}:=(\1-\Pi)\Omega(\1-\Pi)$ is the 
projected effective verification operator 
and $\lambda_{\max}(X)$ denotes the maximal eigenvalue of the Hermitian operator $X$.
\end{theorem}
The proof can be found in Appendix~\ref{app:proof of subspace verification}. 
Therefore, the probability of accepting the \textbf{Bad} case is bounded as follows,   
\begin{align}
    \Pr\left\{\text{``accept''}\vert \sigma_1, \cdots, \sigma_N\right\} 
    \leq (1 - \nu(\Omega)\epsilon)^N, 
\end{align}
where $\nu(\Omega) := 1 - \lambda_{\max}(\wh{\Omega})$ is the \emph{spectral gap}. 
To achieve a confidence level $1-\delta$, it suffices to take 
\begin{align}\label{eq:sample-complexity}
    N(\Omega) \geq \frac{1}{\nu(\Omega)}\times\frac{1}{\epsilon}\ln\frac{1}{\delta} 
\end{align}
state copies. 
This inequality provides a guideline for constructing efficient verification strategies by maximizing $\nu(\Omega)$. 

As evident from Eqs.~\eqref{eq:verification-operator} and~\eqref{eq:sample-complexity}, 
the sample complexity of a verification strategy $\Omega$ depends heavily on the set of available measurements, $\cM$. 
Consider the extreme case where the strategies are constructed from measurements without any restrictions, referred to as \emph{global} or \emph{entangled measurements}, in contrast to the \emph{local measurements} that will be investigated.
The complexity of this globally optimal strategy serves as a reasonable benchmark for 
verification strategies with restricted measurements. 
Define the test POVM $\{\Omega_{g}, \1 - \Omega_{g}\}$ with $\Omega_{g} = \Pi$,
which satisfies $\nu(\Omega_{g}) = 1$. 
Therefore, it suffices to take $N \geq 1/\epsilon\ln 1/\delta$ state copies to achieve a confidence level $1-\delta$. 
However, it is important to note that the globally optimal verification strategy necessitates the use of entangled measurements if the target subspace is entangled (in which case there is at least one entangled basis state). 
Implementing entangled measurements is experimentally challenging. 
In the following, we discuss subspace verification under local and non-adaptive measurement restrictions, where each POVM element $M$ is a local projector 
and fixed \emph{a prior}, rather than being chosen based on prior measurement settings and/or their outcomes. 
Strategies under these restrictions are much more experimentally feasible. 
Notably, we show that these severely restrictions incur only a \emph{constant-factor} penalty compared to the globally optimal strategy. 

\section{Stabilizer subspace verification}

Let $\{I, X, Y, Z\}$ be the single-qubit Pauli operators. 
Let $\cP_n$ denotes the Pauli group on $n$ qubits, 
consisting of $n$-fold tensor products of $I, X, Y, Z$ with the overall factors $\pm 1$ or $\pm i$. 
An $[\![n, k, d]\!]$ stabilizer code is defined by a commutative group $\cS_{k}\subseteq \cP_n$ acting on the state space of $n$ physical qubits~\cite{gottesman1997stabilizer,breuckmann2021quantum}. 
The group $\cS_{k}$ has $n-k$ independent generators, labeled as $\cG_{k} := \{S_1, S_2, \cdots, S_{n-k}\}$. 
A stabilizer code $\cS_{k}$ induces a stabilizer subspace $\cV$, defined to be the $+1$ eigenspace of $\cS_{k}$ and can be interpreted
as the state space of $k$ logical qubits~\cite{gottesman1997stabilizer,shor1995scheme,1996multipleparticle,calderbank1996good,kitaev2003faulttoleranta}. 
Thus, verifying the stabilizer subspace is equivalent to verifying logical qubits~\cite{gottesman1997stabilizer}, 
a critical task as we are marching on the era of fault-tolerant quantum computing.
In the following, we provide two efficient verification strategies for 
verifying a general stabilizer subspace $\cV$ induced by $\cS_k$. 
These strategies utilize only Pauli measurements, which have two outcomes, $+1$ or $-1$, indicating whether the state lies in the positive or negative eigenspace of the corresponding Pauli operator.

\textbf{Strategy I:} This strategy involves choosing a stabilizer operator $P$ from $\cS_k\setminus\{\1\}$ 
uniformly at random and measure the state with $P$. 
If the measurement outcome is $+1$, the state passes this test; otherwise, it fails.
Mathematically, the verification operator of this strategy is given by 
\begin{align}
    \Omega_{\rm I} := \frac{1}{2^{n-k}-1}\sum_{P\in\cS_k\setminus\{\1\}} P^+, 
    \label{eq:stabilizer protocol 1}
\end{align}
where $P^+:=(P+\1)/2$ is the projector onto the positive eigenspace of the stabilizer operator $P$. 
We prove in Appendix~\ref{app:proof of stabilizer subspace verification} that 
$\Omega_{\rm I}$ is a valid verification strategy for $\cV$, 
the spectral gap of $\Omega_{\rm I}$ is $\nu(\Omega_{\rm I}) = 2^{n-k-1}/(2^{n-k} - 1)$,
and uniform random sampling is optimal when only the set of measurements $\cS_k$ is available.
Therefore, to achieve a confidence level $1-\delta$, it suffices to take 
\begin{align}\label{eq:Omega-1-sample-complexity}
    N(\Omega_{\rm I}) = \frac{2^{n-k}-1}{2^{n-k-1}}\frac{1}{\epsilon}\ln\frac{1}{\delta} \approx 2\frac{1}{\epsilon}\ln\frac{1}{\delta}
\end{align}
state copies, which is independent of the number of logical qubits $k$ and physical qubits $n$. 
Moreover, this strategy necessitates at most twice as many copies as the globally optimal verification. 
However, \textbf{Strategy I} requires implementing $2^{n-k} - 1$ different Pauli measurement settings, 
which increases exponentially with the number of generators $n - k$. 
This drawback motivates the second strategy, which involves exponentially fewer measurement settings. 

\textbf{Strategy II:} 
This strategy chooses a stabilizer generator $S$ from $\cG_k$ uniformly at random 
for each state and performs the corresponding measurement. 
The state passes only if the measurement outcome is $+1$.  
Mathematically, the verification operator of this strategy is given by
\begin{align}
    \Omega_{\rm II} := \frac{1}{n-k}\sum_{S\in\cG_k} S^+, 
    \label{eq:stabilizer protocol 2}
\end{align}
where $S^+ := (S+\1)/2$ is the projector onto the positive eigenspace of the stabilizer generator $S$. 
In Appendix~\ref{app:proof of stabilizer subspace verification}, we prove that 
$\Omega_{\rm II}$ is a valid verification strategy for $\cV$, 
the spectral gap of $\Omega_{\rm II}$ is $\nu(\Omega_{\rm II}) = 1/(n-k)$,
and uniform random sampling is optimal when only the set of measurements $\cG_k$ is available.
Therefore, to achieve a confidence level of $1 - \delta$, it suffices to choose
\begin{align}
    N(\Omega_{\rm II}) = (n-k)\frac{1}{\epsilon}\ln\frac{1}{\delta}.
\end{align}
\textbf{Strategy II} requires exponentially fewer measurement settings compared to \textbf{Strategy I} 
while consuming only $(n-k)/2$ times more state copies than \textbf{Strategy I}. 
The results of these two strategies are consistent with those in~\cite{pallister2018optimala}, 
which investigated the special case $k = 0$, corresponding solely to an entangled state. 
Although these two strategies work for all stabilizer subspaces, more efficient strategies warrant further investigation. 
In the following sections, we explore two large classes of stabilizer codes that are of practical interests 
and propose subspace verification strategies requiring fewer measurement settings and consuming fewer state samples.

\section{Graph subspace verification}

A stabilizer code can be represented as a graph code 
using local unitary transformations~\cite{schlingemann2001stabilizer}, 
and graph codes are of primary importance in photonic quantum technologies~\cite{bell2023optimizing}. 
An $[\![n,k,d]\!]$ graph code $\mathscr{G} = (G, W)$ is determined by an undirected graph $G=(V,E)$ 
and an Abelian group $W := \langle\bm{w}_1, \cdots, \bm{w}_k\rangle\subseteq \bZ_2^n$, where $V\equiv [n]:=\{1,\cdots,n\}$. 
The graph $G$ induces a \emph{graph state} $\ket{G}$, 
which is an $n$-qubit stabilizer state determined by the stabilizer generators $\cG_G=\{S_i\}_{i=1}^n$, 
where $S_i = X_i \prod_{(i, j)\in E} Z_j$, and $X_i$ means $X$ in the $i$-th qubit.
The corresponding \emph{graph subspace} $\cV_{\mathscr{G}}$ is spanned by an orthogonal basis $\{G_{\bm{w}}: \bm{w} \in W\}$, where 
\begin{align}\label{eq:graph basis}
    \ket{G_{\bm{w}}} := Z^{\bm{w}} \ket{G} = \prod_{i=1}^n Z_i^{w_i} \ket{G}. 
\end{align}
A graph code can also be represented in graphical form, with each $\bm{w}$ visualized as a logical qubit, as shown in Fig.~\ref{fig:graph subspace verification}.
In the following, we propose an efficient verification strategy for $\cV_{\mathscr{G}}$ in two steps. 
First, we construct a new graph structure of the graph code $\mathscr{G}$ consisting of $n-k$ vertices. 
Then, we introduce a coloring strategy for the graph code based on the new graph structure, 
inspired by hypergraph state verification strategies~\cite{zhu2019efficientb}. 
For the sake of easy understanding, we focus on the single logical qubit case, i.e., $k=1$, 
with the analysis for multi-logical qubits case detailed in Appendix~\ref{app:proof of the graph subspace}. 

\subsection{Step 1: Design a new graph} 

Since $k=1$, there is only one element $\bm{w}\neq\bm{0}\in W$. 
Let ${\rm supp}(\bm{x}) := \{l:\bm{x}_l\neq 0\}\subseteq [n]$ for $\bm{x}\in \bZ_2^n$.
We design a new graph $G' = (V', E')$ that induces $n-1$ stabilizer generators of the graph code $\mathscr{G}$ via $\bm{w}$.
The vertices in $V'$ and their associated stabilizer generators are constructed as follows:
\begin{itemize}
    \item If $\vert{\rm supp}(\bm{w})\vert=1$, 
            then $V' = V \setminus {\rm supp}(\bm{w})$ and $\forall a\in V'$, $S_a' = S_a$.
    \item Else, we first set $V' = V \setminus {\rm supp}(\bm{w})$ and $\forall a\in V'$, $S_a' = S_a$. 
          Then, we sort the vertices in ${\rm supp}(\bm{w})$ in the ascending order and group them pair-wisely. 
          For example, ${\rm supp}(\bm{w})=\{1,2,5\}$ yields $\{\{1,2\}, \{2,5\}\}$.
          Finally, we treat each pair $\{a,\bar{a}\}$ as a new vertex and add it to $V'$.
          The associated stabilizer $S'_{\{a,\bar{a}\}}:=S_aS_{\bar{a}}$.
\end{itemize}
By construction, $\vert V'\vert = n-1$ and vertices in $V'$ are directly related to vertices in $V$.
We show in Appendix~\ref{app:proof of the graph subspace} that such 
constructed stabilizers are indeed stabilizer generators of $\mathscr{G}$.
Note that the construction of stabilizers here is different from the 
standard stabilizers construction of graph states.

Define $\delta_{a,b} = 1$ if $(a,b)\in E$ else $\delta_{a,b}=0$. 
The set of edges $E'$ is constructed as follows: For arbitrary $v_1,v_2\in V'$, 
\begin{enumerate}
    \item If $v_1 = a$ and $v_2 = b$, $(v_1,v_2)\in E'$ if $\delta_{a,b}=1$; 
    \item If $v_1=a$ and $v_2=(b,\bar{b})$, $(v_1, v_2)\in E'$ if $\delta_{a,b} = 1$ or $\delta_{a,\bar{b}}=1$; 
    \item If $v_1=\{a,\bar{a}\}$ and $v_2=\{b,\bar{b}\}$, $(v_1,v_2)\not\in E'$ if at least one of the following two conditions holds: 
     \begin{enumerate}
        \item $v_1\cap v_2=\emptyset$ and $\delta_{a,b} = \delta_{a,\bar{b}} = \delta_{\bar{a},b} = \delta_{\bar{a},\bar{b}}$; 
        \item $v_1\cup v_2=\{a,b,c\}$ and $\delta_{a,b} = \delta_{a,c} = \delta_{b,c}$. 
    \end{enumerate}
\end{enumerate}
These rules are visualized in Appendix~\ref{app:proof of the graph subspace}. 
Note that the edges in $E'$ do not contribute to the construction of stabilizers but 
determine the construction of the verification strategy.

To deepen understanding of the construction, which is a bit complicated,
we present two representative examples of graph codes~\cite{cafaro2014schemea} in Fig.~\ref{fig:graph subspace verification},
explicitly detailing the original graphs, the new graphs, and the relationships between these graphs.

\begin{figure}[!htbp]
    \centering
    \includegraphics[width=\linewidth]{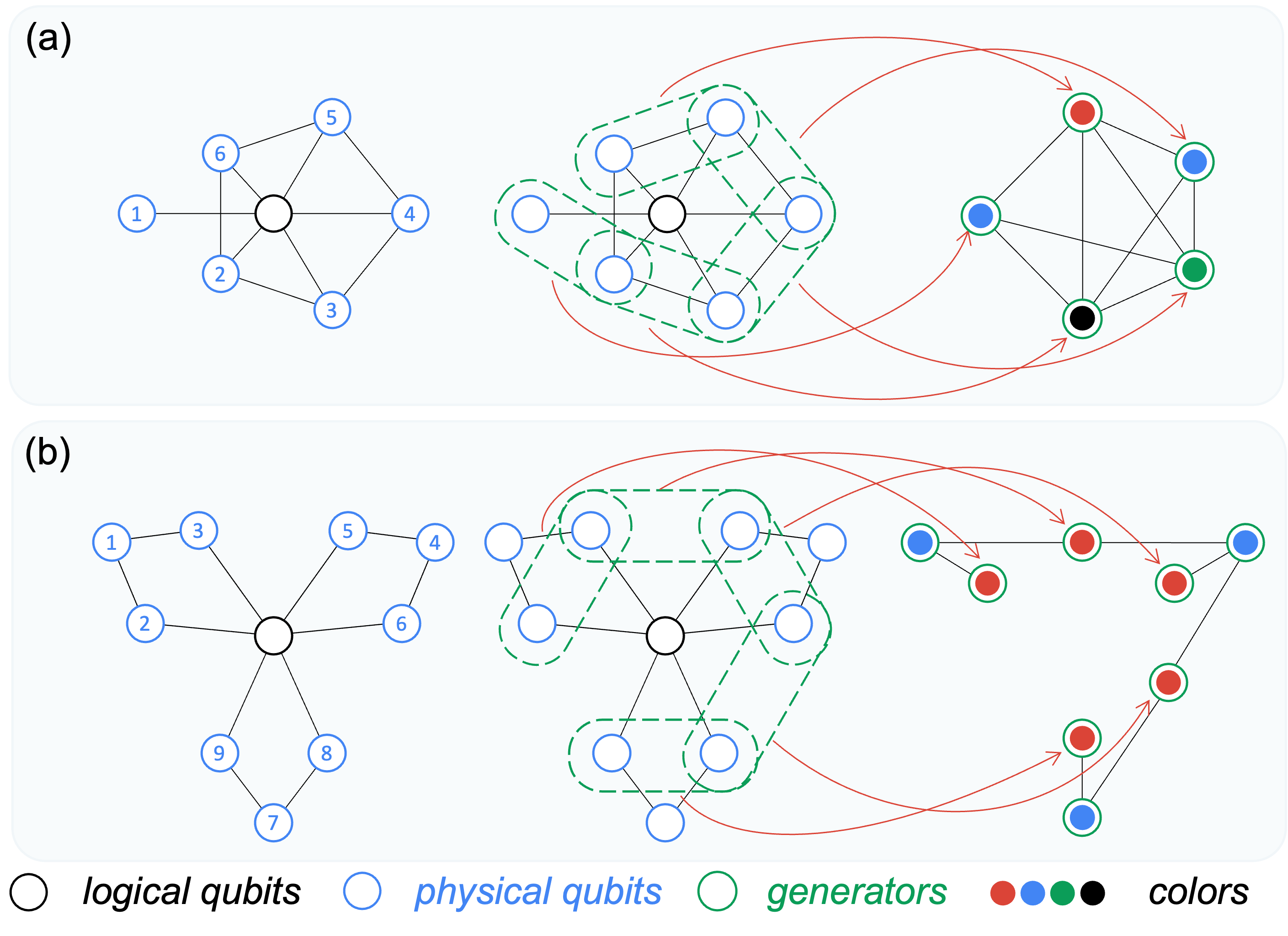}
    \caption{\raggedright 
    Examples of graph subspace verification: 
    (a) a $[\![6,1,3]\!]$ graph subspace and 
    (b) a $[\![9,1,3]\!]$ graph subspace~\cite{cafaro2014schemea}. 
    Figures in left show the original graph structure of these two codes, where (a) $\bm{w}=111111$ and (b) $\bm{w}=011011011$. 
    Figures in middle show the stabilizer generators: (a) $\{S_1 S_2, S_2 S_3, S_3 S_4, S_4 S_5, S_5 S_6\}$; 
    (b) $\{S_1, S_4, S_7, S_2 S_3, S_3 S_5, S_5 S_6, S_6 S_8, S_8 S_9\}$. 
    Figures in right show the new graph structures derived from these graph codes.
    Using the coloring strategy, these new graphs can be colored with $4$ and $2$ colors respectively. 
    This indicates that the corresponding verification strategies require $4$ and $2$ measurement settings: 
    (a) \red{\texttt{ZZZZYY}}, \blue{\texttt{XXZYYZ}}, \green{\texttt{ZZYYZZ}}, and \texttt{ZYYZZZ}; 
    (b) \red{\texttt{ZXXZXXZXX}} and \blue{\texttt{XZZXZZXZZ}}. 
    }
    \label{fig:graph subspace verification}
\end{figure}

\subsection{Step 2: Construct a coloring strategy}

Based on the new graph $G'$, 
we propose the following coloring verification strategy for the graph subspace $\cV_{\mathscr{G}}$. 
Let $\mathscr{A} := \{A_1, A_2, \cdots, A_m\}$ be an independence cover of $G'$ 
composed of $m$ nonempty independent sets~\cite{scheinerman2013fractional,zhu2019efficientb}. 
For each $A_\ell\in\mathscr{A}$, we construct a test operator $P_\ell$ in the following way: 
For each $v\in A_\ell$,
\begin{enumerate}
    \item If $v=a$, perform $X$ measurement on the $a$-th qubit; 
    \item If $v=(a,\bar{a})$ and $\delta_{a,\bar{a}}=1$, perform $Y$ measurements on $a$-th 
            and $\bar{a}$-th qubits; If $v=(a,\bar{a})$ and $\delta_{a,\bar{a}}=0$, 
            perform $X$ measurements on these two qubits. 
    \item Perform $Z$ measurement on other not involved qubits.
\end{enumerate}
The state passes the test only if the measurement result of $S'_a$ is $+1$ for all $a\in A_\ell$.
As shown in Appendix~\ref{app:proof of the graph subspace}, the test operator takes the form $P_\ell = \prod_{a \in A_\ell} (\1 + S'_a)/2$. 
The verification strategy is constructed by sampling from the set $\{P_\ell\}_{\ell=1}^m$ uniformly at random and 
the verification operator is given by
\begin{align}\label{eq:graph-subspace-protocol}
    \Omega_{\mathscr{G}} = \frac{1}{m}\sum_{\ell=1}^m P_\ell. 
\end{align}
In Appendix~\ref{app:proof of the graph subspace}, we prove that 
$\Omega_{\mathscr{G}}$ is a valid verification strategy for the graph code $\mathscr{G}$, 
the spectral gap satisfies $\nu(\Omega_{\mathscr{G}}) = 1/m$, and
uniform random sampling is optimal when only the set of measurements $\{P_\ell\}_{\ell=1}^m$ is available. 
Consequently, to achieve a confidence level of $1-\delta$, it suffices to take
\begin{align}
    N(\Omega_{\mathscr{G}}) = m\frac{1}{\epsilon}\ln\frac{1}{\delta}.
\end{align}
This result is consistent with~\cite{zhu2019efficientb}, which investigated hypergraph state verification. 
Recall that \textbf{Strategy II} requires $n-1$ measurement settings to verify this graph subspace. 
Since $m \leq n-1$ by the definition of the independence cover, the new strategy $\Omega_{\mathscr{G}}$ 
can verify the graph subspace using even fewer measurement settings and state copies, 
underscoring the superiority of the coloring strategy. 
Surprisingly, as demonstrated in Appendix~\ref{app:proof of the graph subspace}, 
this advantage persists in multi-logical qubit cases, 
which can be efficiently verified using a coloring strategy with $m \leq n-k$ measurement settings. 
This superiority is rooted in the fact that many test operators in \textbf{Strategy I} and \textbf{Strategy II} can be conducted in parallel,
thanks to the graph structure. This observation motivates an efficient protocol for the CSS subspace considered below. 

\begin{table*}[!htbp]
\centering
\renewcommand{\arraystretch}{1.5} 
\setlength{\tabcolsep}{12pt} 
\setlength\heavyrulewidth{0.3ex}  
\begin{tabular}{@{}cccc@{}}
\toprule
\textbf{Subspace} & \textbf{Strategy} & \textbf{\# Measurement Settings} & \textbf{Sample Complexity} \\ \midrule
\multirow{2}{*}{Stabilizer subspace} 
    & $\Omega_{\rm I}$ in Eq.~\eqref{eq:stabilizer protocol 1}  
    & $2^{n-k}-1$ & $2/\epsilon\ln(1/\delta)$ \\
    & $\Omega_{\rm II}$ in Eq.~\eqref{eq:stabilizer protocol 2} 
    & $n-k$ & $(n-k)/\epsilon\ln(1/\delta)$ \\
Graph subspace & $\Omega_{\mathscr{G}}$ in Eq.~\eqref{eq:graph-subspace-protocol} & $m\leq n-k$ & $m/\epsilon\ln(1/\delta)$ \\
CSS subspace & $\Omega_{\rm CSS}^{\rm XZ}$ in Eq.~\eqref{eq:CSS-subspace-protocol} & $2$ & $2/\epsilon\ln(1/\delta)$ \\
Dual-containing subspace & $\Omega_{\rm DC}^{\rm XYZ}$ in Eq.~\eqref{eq:DC-subspace-protocol} & $3$ & $1.5/\epsilon\ln(1/\delta)$ \\
Arbitrary & $\Omega_g$ (globally optimal) & $1$ & $1/\epsilon\ln(1/\delta)$ \\
\bottomrule
\end{tabular}
\caption{\raggedright
    Verification strategies for stabilizer subspace induced by an $[\![n,k,d]\!]$ quantum stabilizer code, 
    where $n$ is the number of physical qubits, $k$ is the number of logical qubits, and $d$ is the code distance. 
    The number of measurement settings is the number of required test operators in the strategy.
    The sample complexity is the minimal number of state copies consumed by the strategy
    to verify the target subspace with fidelity parameter $1-\epsilon$ and confidence level $1-\delta$.}
\label{tab:summarize strategies}
\end{table*}

\section{CSS subspace verification}

A large class of stabilizer codes is the CSS code~\cite{Nielsen_Chuang_2010,breuckmann2021quantum}. 
Notably, any $[\![n, k, d]\!]$ stabilizer code can be mapped onto a $[\![4n,2k,2d]\!]$ CSS code~\cite{haah2011local,breuckmann2021quantum}. 
An $[\![n, k, d]\!]$ CSS code can be described by two matrices 
$H_X\in\bZ_2^{k_X\times n}$ and $H_Z\in\bZ_2^{k_Z\times n}$, 
with the condition that $H_X H_Z^T = \bm{0}$~\cite{mackay2004sparse}. 
It is has $k_X+k_Z$ stabilizer generators of the form $X^{\bm{c}_X}$ and $Z^{\bm{c}_Z}$, 
where $P^{\bm{c}} := \bigotimes_{i \in [n]} P^{c_i}$ and $\bm{c}_X$ and $\bm{c}_Z$ are the rows of $H_X$ and $H_Z$, respectively. 
It determines a CSS subspace $\cV_{\rm CSS}$ with dimension $2^k$, where $k=n-(k_X+k_Z)$.
\textbf{Strategy II} requires $n-k = k_X+k_Z$ measurement settings to verify $\cV_{\rm CSS}$. 
However, by utilizing the structure information of the CSS code generators, 
we propose a more efficient verification strategy for $\cV_{\rm CSS}$, termed the XZ strategy, which requires only $2$ measurement settings. 
The strategy works as follows. 
We choose a Pauli operator $P\in\{X, Z\}$ uniformly at random and perform the measurement $P^{\otimes n}$.
The state passes the test if the measurement result of $P^{\bm{c}_P}$ is $+1$ for all $\bm{c}_P$. 
Mathematically, the verification operator is given by
\begin{align}\label{eq:CSS-subspace-protocol}
    \Omega_{\rm CSS}^{\rm XZ} = \frac{1}{2} 
    \sum_{P\in\{X,Z\}} M_P, 
\end{align}
where $M_P:=\prod_{\bm{c}_P} (\1 + P^{\bm{c}_P})/2$. 
We prove in Appendix~\ref{app:proof of the CSS subspace} that
$\Omega_{\rm CCS}^{\rm XZ}$ is a valid verification strategy for the CSS subspace, 
with the spectral gap satisfying $\nu(\Omega_{\rm CSS}^{\rm XZ}) = 1/2$. 
Moreover, uniform random sampling is optimal 
when the set of available measurements is $\{X^{\otimes n}, Z^{\otimes n}\}$. 
Interestingly, the spectral gap is independent of the number of logical qubits $k$ and physical qubits $n$.
Therefore, to achieve a confidence level $1-\delta$, it suffices to take 
\begin{align}
    N(\Omega_{\rm CSS}^{\rm XZ}) = 2\frac{1}{\epsilon}\ln\frac{1}{\delta}, 
\end{align}
which is independent of the subspace size and approximately equals to the required state copies of 
\textbf{Strategy I} in Eq.~\eqref{eq:Omega-1-sample-complexity}. 
In Appendix~\ref{app:proof of the CSS subspace}, we propose a concrete verification strategy 
for the subspace induced by the toric code, which is a special type of CSS code. 
Below, we present an even more efficient strategy for the subspace induced by the dual-containing code.

\subsection{Dual-containing subspace verification}

A well-known special case of CSS code is the dual-containing code with $H_X = H_Z$, e.g., the $[\![7,1,3]\!]$ Steane code~\cite{Nielsen_Chuang_2010,mackay2004sparse}. 
The corresponding subspace is termed the \emph{dual-containing subspace} and denoted by $\cV_{\rm DC}$. 
To verify $\cV_{\rm DC}$, we can directly apply $\Omega_{\rm CSS}^{\rm XZ}$,
which requires $2$ measurement settings and $2/\epsilon \ln(1/\delta)$ state copies. 
However, we propose an even more efficient strategy using $3$ measurement settings to further reduce the number of consumed state copies. 
The strategy works as follows. 
We choose a Pauli operator $P\in\{X, Z, Y\}$ uniformly at random and perform measurement $P^{\otimes n}$ on all physical qubits.
The state passes if the measurement result of $P^{\bm{c}}$ is $+1$ for all $\bm{c}$, where $\bm{c}$ represents the rows of $H_X$. 
The verification operator for this strategy is given by 
\begin{align}\label{eq:DC-subspace-protocol}
    \Omega_{\rm DC}^{\rm XYZ} = \frac{1}{3} 
    \sum_{P'\in\{X, Z, iY\}} M_{P'}, 
\end{align}
where $M_{P'} := \prod_{\bm{c}} (\1 + P'^{\bm{c}}) / 2$. 
We prove in Appendix~\ref{app:proof of the CSS subspace} that
$\Omega_{\rm DC}^{\rm XYZ}$ is a valid verification strategy for the dual-containing subspace, 
with the spectral gap satisfying $\nu(\Omega_{\rm DC}^{\rm XYZ}) = 2/3$. 
Furthermore, uniform random sampling is optimal when the set of available measurements is 
$\{X^{\otimes n},Y^{\otimes n},Z^{\otimes n}\}$. Thus, it suffices to take 
\begin{align}
    N(\Omega_{\rm DC}^{\rm XYZ}) = \frac{3}{2}\frac{1}{\epsilon}\ln\frac{1}{\delta}
\end{align}
state copies to achieve a confidence level of $1 - \delta$. 
Compared to $\Omega_{\rm CSS}^{\rm XZ}$ in Eq.~\eqref{eq:CSS-subspace-protocol}, 
this new strategy reduces the required number of state copies by one-quarter 
at the cost of adding an additional measurement setting, specifically performing $Y$ measurements on all qubits. 

\section{Conclusions}

We have established a general framework for verifying entangled subspaces, 
enabling efficient certification of whether an unknown state belongs to a specific subspace. 
Then, we focus on the entangled subspaces spanned by stabilizer codes, 
which constitute the fundamental cornerstone for the era of fault-tolerant quantum computing.  
The obtained results are summarized in Table~\ref{tab:summarize strategies}.
For a general $[\![n,k,d]\!]$ stabilizer code, we have proposed two efficient strategies to verify the corresponding stabilizer subspace. 
\textbf{Strategy I} is derived from the full stabilizer group, requires $2^{n-k}-1$ measurement settings, and consumes $2/\epsilon \ln(1/\delta)$ state copies. 
\textbf{Strategy II} is constructed solely from the stabilizer generators, requires $n-k$ measurement settings, and consumes $(n-k)/\epsilon \ln(1/\delta)$ states copies. 
To further improve the verification efficiency, we have investigated two special types of stabilizer codes, namely graph codes and CSS codes, which are of practical interests. 
For an $[\![n,k,d]\!]$ graph code, we have introduced a new graph structure of the graph code and propose a coloring strategy for verifying the corresponding graph subspace. 
This strategy uses $m\leq n-k$ measurement settings and requires $m/\epsilon \ln(1/\delta)$ state copies. 
For an $[\![n,k,d]\!]$ CSS code, we have proposed an XZ strategy that uses only $2$ measurement settings and consumes $2/\epsilon \ln(1/\delta)$ state copies, regardless of the code size. 
This strategy achieves \emph{nearly} the same performance as the globally optimal strategy in terms of both the required number of measurement settings and state copies.
For the special dual-containing codes, we have proposed an even more efficient strategy that uses 
just $3$ measurement settings and consumes $1.5/\epsilon\ln(1/\delta)$ state copies. 
This work contributes the first systematic study of efficient verification of stabilizer code subspaces with local measurements,
enabling experimentally efficient certification of both logical qubits and logical operations in noisy quantum computers.

Several key questions regarding quantum subspace verification remain unresolved, 
with one of the most prominent being the design and proof of \emph{optimal} verification strategies 
using local measurements. In existing approaches, the quantum states are typically destroyed during the measurement process, rendering them unusable for subsequent tasks. Thus, identifying verification strategies that employ quantum nondemolition measurements, 
which preserve the post-measurement states, is both intriguing and significant~\cite{liu2021universallya}.

\prlsection{Note.}After the completion of this work, 
we have become aware of a related work by Chen \textit{et al.}~\cite{chen2024quantuma}.

\section*{Acknowledgements}

This work was supported by 
the National Natural Science Foundation of China (Grant Nos. 62471126),  the Jiangsu Key R\&D Program Project (Grant No. BE2023011-2), 
the SEU Innovation Capability Enhancement Plan for Doctoral Students (Grant No. CXJH\_SEU 24083), 
the National Key Research and Development Program of China (Grant No. 2022YFF0712800), 
the Innovation Program for Quantum Science and Technology (Grant No. 2021ZD0301500), 
and the Fundamental Research Funds for the Central Universities (Grant No. 2242022k60001).


%

\makeatletter
\newcommand{\appendixtitle}[1]{\gdef\@title{#1}}
\newcommand{\appendixauthor}[1]{\gdef\@author{#1}}
\newcommand{\appendixaffiliation}[1]{\gdef\@affiliation{#1}}
\newcommand{\appendixdate}[1]{\gdef\@date{#1}}
\makeatother

\makeatletter%
\newcommand{\appendixmaketitle}{%
\begin{center}%
\vspace{0.4in}%
{\Large \@title \par}%
\end{center}%
\par%
}%
\makeatother%

\setcounter{secnumdepth}{2}
\appendix
\widetext
\newpage

\appendixtitle{\bf 
Supplemental Material for\\``\thetitle''}
\appendixmaketitle
\vspace{0.2in}


In this supplementary information,we provide details on 
\begin{itemize}
    \item In Appendix~\ref{app:proof of subspace verification}, 
    we provide details of \textbf{\emph{General framework of subspace verification}}. 
    \item In Appendix~\ref{app:proof of stabilizer subspace verification}, we prove that \textbf{Strategy I} and \textbf{Strategy II} are valid verification strategies for general stabilizer subspace. 
    \item In Appendix~\ref{app:proof of the graph subspace}, we provide the details of steps of graph subspace verification. 
    Especially, we consider the more general case where
    the subspace induced by $[\![n,k,d]\!]$ graph codes with $k\geq 2$. 
    \item In Appendix~\ref{app:proof of the CSS subspace}, we provide the details of CSS subspace verification,
            including toric codes and dual-containing codes as special cases.
\end{itemize}

\section{General framework of subspace verification}
\label{app:proof of subspace verification}

Consider a quantum computer $\cD$ that produces $N$ copies of $n$-qubit states, denoted 
$\sigma_1, \sigma_2, \cdots, \sigma_N$. 
The task of quantum state verification seeks to answer the question: 
\begin{center}
``Are the states $\sigma_i$ generated by $\cD$ equal to a fixed state $\ketbra{\psi}$?''
\end{center}
Similarly, the quantum subspace verification aims to answer the question: 
\begin{center}
``Are the states $\sigma_i$ generated by $\cD$ contained within a specific subspace $\cV$?''
\end{center}
In the following, we set up a formal framework for general subspace verification strategies. 

\subsection{Task of subspace verification}

To mathematically verify whether a state $\sigma\in\mathscr{D}(\cV)$, 
where $\cV := {\rm span}\{\ket{\psi_j}\}_j$, 
we define the projector $\Pi:=\sum\proj{\psi_j}$ and provide the following lemma. 

\begin{lemma}\label{lemma:condition for verifying subspace}
Let $\sigma\in\density{\cH}$ be a density operator. It holds that
$\tr[\Pi\sigma] = 1$ if and only if $\sigma\in\density{\cV}$.  
\end{lemma}

\begin{proof}[Proof of Lemma~\ref{lemma:condition for verifying subspace}]
The \textbf{\emph{necessity}} is obvious.
If $\sigma\in\text{span}\{\ket{\psi_j}\}$, then we have 
\begin{align}
    \sigma = \sum_{jl} \sigma_{jl} \ket{\psi_j}\!\bra{\psi_l}
    \quad \Rightarrow \quad 
    \tr[\Pi \sigma] = \sum_j \bra{\psi_j}\sigma\ket{\psi_j} = 1. 
\end{align}
Now we turn to show the \textbf{\emph{sufficiency}}. 
For an arbitrary matrix $\sigma$, it can be written as 
\begin{align}
    \sigma = \sum_i \lambda_i \ketbra{\phi_i}, \quad 
    \sum_i \lambda_i = 1, \quad \lambda_i \geq 0. 
\end{align}
For each eigenstate $\ket{\phi_i}$, we have 
\begin{align}
    \ket{\phi_i} = \sin\theta_i\ket{\Psi_i} + \cos\theta_i\ket{\Psi_i^\bot},  
\end{align}
where $\sum_j |\braket{\psi_j}{\Psi_i}|^2 = 1$, and $\ket{\Psi_i^\bot}$ is orthogonal to $\ket{\Psi_i}$. 
Then, $\sigma$ can also be written as 
\begin{align}
    \sigma = \sum_i \lambda_i \left(\sin^2\theta_i\ketbra{\Psi_i} + \sin\theta_i\cos\theta_i\ket{\Psi_i}\!\bra{\Psi_i^\bot} + \sin\theta_i\cos\theta_i\ket{\Psi_i^\bot}\!\bra{\Psi_i} + 
    \cos^2\theta_i\ketbra{\Psi_i^\bot}\right). 
\end{align}
With the trace constraint, we have 
\begin{align}
    \sum_i \lambda_i \sin^2\theta_i = 1 
    \quad \Rightarrow \quad 
    \sin\theta_i = 1, \; \forall\; \theta_i, 
    \quad \Rightarrow \quad 
    \sigma = \sum_i \lambda_i \ketbra{\Psi_i}, 
\end{align}
which hints that $\sigma\in\spn{\ket{\psi_i}}$. 
\end{proof}
With the help of the above lemma, we can now formally define the quantum subspace verification 
task---Given a quantum device $\cD$ that is designed to produce states in $\cV$, distinguish between the following two cases:
\begin{enumerate}
    \item \textbf{Good}: for all $i\in[N]$, $\tr[\Pi\sigma_i] = 1$;
    \item \textbf{Bad}:  for all $i\in[N]$, $\tr[\Pi\sigma_i] \leq 1 - \epsilon$ for some fixed $\epsilon$. 
\end{enumerate}

\subsection{The conditions for verification operator}

To complete this distinguish task, we consider randomly pick a POVM element $M\in\cM$ with some probability and consider the corresponding two-outcomes POVMs $\{M, \1-M\}$, 
where $M$ has output ``pass'' and $\1-M$ has output ``fail''. 
Moreover, we define a probability mass $\mu: \cM \to [0, 1]$, $\sum \mu(M) = 1$. 
The probability of a generated quantum state $\sigma$ passing the test can be expressed as 
\begin{align}
    \Pr\left\{\text{``pass''}\vert \sigma\right\} 
    = \sum_{M\in\cM} \mu(M) \tr[M \sigma]
    = \tr[\Omega \sigma], 
\end{align}
where the \emph{verification operator} of this protocol is defined as 
\begin{align}
    \Omega := \sum_{M\in\cM} \mu(M) M. 
\end{align}
To satisfy the requirement of the verification task,
we impose two conditions on the verification operator $\Omega$: 
\begin{enumerate}
    \item \textbf{\emph{Perfect completeness}} condition: 
    states within the target subspace $\cV$ will always pass the test, that is, 
    \begin{align}
        \tr[\Omega \sigma] = 1,\quad\forall\;\sigma\in\mathscr{D}(\cV). 
    \end{align}
    \item \textbf{\emph{Soundness}} condition: states in the \textbf{Bad} case can be rejected with high probability. 
\end{enumerate}

\subsubsection{Perfect completeness condition}

The perfect completeness condition can be equivalently characterized using the projector $\Pi$ associated with $\cV$ as follows. 
\begin{lemma}
\label{lemma:perfect completeness condition}
The perfect completeness condition can be equivalently characterized as
\begin{align}
    \tr[\Omega\Pi]=\rank(\Pi),
\end{align}
where $\rank(\Pi)$ is the rank of the projector.
\end{lemma}

\begin{proof}[Proof of Lemma~\ref{lemma:perfect completeness condition}]
With perfect completeness condition, there exist a set of orthogonal bases $\{\ket{\psi_l^\bot}\}_l$ in the complementary subspace of the target subspace, 
such that $\Omega$ can be written as 
\begin{align}
    \Omega = \Pi + \sum \omega_l\ketbra{\psi^\bot_l},  
    \label{eq:perfect completeness condition of omega}
\end{align}
otherwise $\forall\;\sigma\in\spn{\ket{\psi_j}}, \tr[\Omega\sigma] = 1$ does not hold. 
We define the projected effective verification operator as 
\begin{align}
    \wh{\Omega} 
    := (\1-\Pi)\Omega(\1-\Pi) 
    = \sum \omega_l\ketbra{\psi^\bot_l}. 
\end{align}
Therefore, we have 
\begin{align}
    \tr[\Omega\Pi] 
    = \tr[\Pi^2] + \tr[\wh{\Omega}\Pi] 
    = \tr[\Pi] 
    = \rank(\Pi). 
\end{align}
\end{proof}

\subsubsection{Soundness condition}

Next, we consider the soundness condition. 
We demonstrate that the worst-case passing probability, $p(\Omega)$, defined as 
\begin{align}
    p(\Omega) := \max_{\sigma:\tr[\Pi\sigma]\leq 1-\epsilon}\Pr\{\text{``pass''}|\sigma\}, 
\end{align}
in the \textbf{Bad} case is uniquely determined by the largest eigenvalue of the projected effective verification operator. 
Specifically, 
\begin{align}
    p(\Omega) = \max_{\sigma:\tr[\Pi\sigma]\leq1-\epsilon}\tr[\Omega\sigma]
= 1 - (1-\lambda_{\max}(\wh{\Omega}))\epsilon. 
\end{align}
We provide the proof as follows. 
\begin{proof}[Proof of the Theorem~\ref{theorem:worst-case passing probability}]
For a fixed set $\{\ket{\psi_l^\bot}\}$, 
an arbitrary quantum state $\sigma$ with $\tr[\Pi\sigma] = r$ can always be written as 
\begin{align}
    \sigma = r \Psi + (1-r) \Psi^\bot + \sum_{jl} \left(c_{jl}\ket{\psi_j}\!\bra{\psi_l^\bot} + c_{jl}^* \ket{\psi_l^\bot}\!\bra{\psi_j}\right), 
\end{align}
where $\Psi$ and $\Psi^\bot$ are the states in the $\spn{\ket{\psi_j}}$ and $\spn{\ket{\psi_l^\bot}}$, respectively. 
Then, such a state will pass the test with probability 
\begin{align}
    \Pr\{\text{``pass''}\vert\sigma\} 
    &= \tr[\Omega\sigma] \\
    &= r\tr[\Omega\Psi] + (1 - r)\tr[\wh{\Omega}\Psi^\bot] \\
    &\leq r + (1 - r)\lambda_{\max}(\wh{\Omega}). 
\end{align}
The above inequality becomes an equality if 
\begin{align}
    \Psi^\bot = \ketbra{\psi^\bot_{\max}}, 
\end{align}
where $\ket{\psi_{\max}^\bot}$ is the eigenstate of $\wh{\Omega}$ corresponding to the largest eigenvalue $\lambda_{\max}(\wh{\Omega})$.
Thus, for a fixed $\Omega$, 
\begin{align}
    \max_{\sigma:\tr[\Pi\sigma] = r} \Pr\{\text{``pass''}\vert\sigma\} 
    = r + (1 - r)\lambda_{\max}(\wh{\Omega}), 
\end{align}
which is achieved by any density matrix of the form 
\begin{align}
    \sigma 
    = r\Psi + (1-r)\ketbra{\psi_{\max}^\bot} + \sum_{jl} \left(c_{jl}\ket{\psi_j}\!\bra{\psi_l^\bot} + c_{jl}^* \ket{\psi_l^\bot}\!\bra{\psi_j}\right). 
\end{align}
Note that the pure state $\sigma = \ketbra{\phi}$ for 
\begin{align}
    \ket{\phi} = \sqrt{r}\ket{\psi'} + \sqrt{1-r} \ket{\psi^\bot_{\max}}, 
\end{align}
where $\ket{\psi'}$ is the linear combination of vectors $\ket{\psi_j}$, is of this form. 
Therefore, we can only consider pure states in the following analysis. 

Now, for a fixed $\bar{\epsilon} \geq \epsilon> 0$, we define a state $\sigma = \ketbra{\phi_{\bar{\epsilon}}}$ with 
$\ket{\phi_{\bar{\epsilon}}} = \sqrt{1-\bar{\epsilon}}\ket{\phi} + \sqrt{\bar{\epsilon}}\ket{\phi^\bot}$, 
where $\ket{\phi}$ is the linear combination of vectors $\{\ket{\psi_j}\}_j$ 
and $\braket{\phi}{\phi^\bot} = 0$. 
Then, we define that the worst-case passing probability as 
\begin{align}
    p(\Omega) 
    &:= \max_{\sigma:\tr[\Omega\sigma]\leq 1-\epsilon} 
    \Pr\{\text{``pass''}\vert\sigma\} \\
    &= \max_{\sigma:\tr[\Omega\sigma]\leq 1-\epsilon} 
    \tr[\Omega\sigma] \\
    &= \max_{\bar{\epsilon} \geq \epsilon, \ket{\phi^\bot}} 1 - \bar{\epsilon} + \bar{\epsilon}\bra{\phi^\bot}\wh{\Omega}\ket{\phi^\bot} \\
    &= 1 - (1 - \lambda_{\max}(\wh{\Omega})) \epsilon. 
\end{align}
\end{proof}
Therefore, the probability of accepting the \textbf{Bad} case is bounded as follows,   
\begin{align}
    \Pr\left\{\text{``accept''}\vert \sigma_1, \cdots, \sigma_N\right\} 
    \leq (1 - \nu(\Omega)\epsilon)^N. 
\end{align}
We want this probability to be bounded from above by $\delta > 0$, that is 
\begin{align}
    (1 - \nu(\Omega)\epsilon)^N \leq \delta. 
\end{align} 
Thus, we have 
\begin{align}
    N \geq \frac{1}{\ln(1-\nu(\Omega)\epsilon)^{-1}} \ln\frac{1}{\delta} 
    \approx \frac{1}{\nu(\Omega)}\times\frac{1}{\epsilon}\ln\frac{1}{\delta}. 
\end{align}
This inequality provides a guideline for constructing efficient verification by maximizing $\nu(\Omega)$ so that we can use less state copies. 

\section{Proof of stabilizer subspace verification}
\label{app:proof of stabilizer subspace verification}

In this section, we prove the two verification strategies of stabilizer subspace. 
For the subspace $\cV$ determined by $\cG_k$, we define a set of orthogonal bases in $\cV$ as 
\begin{align}
    \{\ket{\psi_1}, \cdots, \ket{\psi_{2^{k}}}\}. 
\end{align}
The set of stabilizer operators is defined as $\cS_k$, and the projector onto $\cV$ can be defined in the following two ways, 
\begin{align}
    \Pi_{\cV} = \frac{1}{2^{n-k}} \sum_{P\in\cS_k}P = \sum_{j=1}^{2^{k}} \ketbra{\psi_j}. 
\end{align}
With Eq.~\eqref{eq:perfect completeness condition of omega}, we know that a feasible verification strategy $\Omega$ must be in the following form: 
\begin{align}
    \Omega = \sum_{j=1}^{2^{k}} \ketbra{\psi_j} + \sum_{l=1}^{2^n-2^{k}} \omega_l \ketbra{\psi^\bot_l} 
    \label{eq:form of the stabilizer verification strategy}
\end{align}
where $\{\ket{\psi_1^\bot}, \cdots, \ket{\psi_{2^n-2^{k}}^\bot}\}$ is a set of orthogonal bases in complementary subspace of $\cV$. 
Additionally, the spectral gap of $\Omega$ is 
\begin{align}
    \nu(\Omega) = 1 - \max_l\;\omega_l. 
\end{align}

\subsection{Proof of Strategy I} 

Now, we begin our proof. 
Firstly, we show that $\Omega_{\rm I}$ defined is a valid verification strategy and compute the corresponding spectral gap. 
For each $P\in\cS_k\setminus\{\1\}$, we have 
\begin{align}
    P = P^+ - P^-, \quad 
    P^+ + P^- = \1, 
\end{align}
where $P^+\;(P^-)$ is the projector onto the positive (negative) eigenspace of $P$. 
With the above decomposition, we have 
\begin{align}
    \sum_{j=1}^{2^{k}}\ketbra{\psi_j}
    &= \frac{1}{2^{n-k}} \sum_{P\in\cS_k\setminus\{\1\}} \left(P^+ - P^-\right) + \frac{1}{2^{n-k}} \1 \\
    &= \frac{1}{2^{n-k}} \sum_{P\in\cS_k\setminus\{\1\}} \left(2P^+ - \1\right) + \frac{1}{2^{n-k}} \1 \\
    &= \frac{1}{2^{n-k-1}} \sum_{P\in\cS_k\setminus\{\1\}} P^+ - \left(1 - \frac{1}{2^{n-k-1}}\right)\1. 
\end{align}
Then, we have 
\begin{align}
    \sum_{P\in\cS_k\setminus\{\1\}} P^+ 
    &= 2^{n-k-1}\sum_{j=1}^{2^{k}}\ketbra{\psi_j} + \left(2^{n-k-1} - 1\right)\1 \\
    &= 2^{n-k-1}\sum_{j=1}^{2^{k}}\ketbra{\psi_j} + \left(2^{n-k-1} - 1\right)\left(\sum_j \ketbra{\psi_j} + \sum_l \ketbra{\psi^\bot_l}\right) \\
    &= (2^{n-k} - 1) \sum_{j=1}^{2^{k}}\ketbra{\psi_j} + \left(2^{n-k-1} - 1\right)\sum_l \ketbra{\psi^\bot_l}. 
\end{align}
Finally, we derive the desired equation 
\begin{align}
    \frac{1}{2^{n-k}-1}\sum_{P\in\cS_k\setminus\{\1\}} P^+ 
    &= \sum_{j=1}^{2^{k}}\ketbra{\psi_j} + \frac{2^{n-k-1} - 1}{2^{n-k}-1}\sum_l \ketbra{\psi_l^\bot}, 
\end{align}
which hints the $\Omega_{\rm I}$ defined in Eq.~\eqref{eq:stabilizer protocol 1} is feasible and 
\begin{align}
    \nu(\Omega_{\rm I}) = 1 - \frac{2^{n-k-1}-1}{2^{n-k}-1} = \frac{2^{n-k-1}}{2^{n-k}-1}. 
\end{align}

Then, we prove that the uniform sample distribution is the optimal one. 
Define the probability of measurement $P$ to be sampled as $\mu(P)$, $P\in\cS_k\setminus\{\1\}$. 
The corresponding verification operator is 
\begin{align}
    \Omega_{{\rm I},\mu} := \sum_{P\in\cS_k\setminus\{\1\}} \mu(P) P^+. 
\end{align}
We have $\tr[\Omega_{{\rm I},\mu}] = 2^{n-1}$. 
With Eq.~\eqref{eq:form of the stabilizer verification strategy}, we know that 
\begin{align}
    \tr[\Omega_{{\rm I},\mu}] = 2^{k} + \sum_l \omega_l = 2^{n-1}. 
\end{align}
Therefore, we have 
\begin{align}
    \max_l \omega_l \geq \frac{2^{n-1} - 2^{k}}{2^n - 2^{k}} 
    = \frac{2^{n-k-1}-1}{2^{n-k} - 1}, 
\end{align}
the equality holds when $\mu(P) = \frac{1}{2^{n-k}-1}$ for all $P\in\cS_k\setminus\{\1\}$, i.e., $\Omega_{{\rm I},\mu} = \Omega_{\rm I}$.

\subsection{Proof of Strategy II}


Here, we define a verification strategy with probability distribution $\mu$ as 
\begin{align}
    \Omega_{{\rm II}, \mu} := \sum_{S\in\cG_k} \mu(S) S^+, 
\end{align}
i.e., each measurement $S\in\cG_k$ is sampled with probability $\mu(S)$. 
Then, we have 
\begin{align}
    \tr[\Omega_{{\rm II}, \mu} \Pi_V] 
    &= \frac{1}{2^{n-k}} \sum_{S\in\cG_k, P\in\cS_k} \mu(S) \tr[PS^+] \\
    &= \frac{1}{2^{n-k}} \sum_{S\in\cG_k, P\in\cS_k} \frac{1}{2}\mu(S)\left(\tr[P] + \tr[P S]\right) \\ 
    &= \frac{1}{2^{n-k+1}} \left[\sum_{S\in\cG_k, P\in\cS_k}\mu(S)\tr[P] + \sum_{S\in\cG_k, P\in\cS_k} \tr[PS]\right] \\
    &= \frac{1}{2^{n-k+1}} \left[2^n + \sum_{S\in\cG_k} \mu(S)\cdot 2^n\right] \\
    &= \frac{2^{n+1}}{2^{n-k+1}} = 2^{k} = \rank(\Pi_V). 
\end{align}
Thus, $\Omega_{{\rm II}, \mu}$ must satisfy perfect completeness condition defined in Lemma~\ref{lemma:perfect completeness condition}. 
Subsequently, we analyze the spectral gap of $\Omega_{{\rm II},\mu}$. 
We define a set of stabilizer generators with size $n$, 
\begin{align}
    \cG_n = \{\underbrace{S_1, \cdots, S_{n-k}}_{\cG_k}, S_{n-k+1}, \cdots, S_n\}. 
\end{align}
Then, we can construct a set of orthogonal bases $\ket{C_w}$ with $n$-bit strings $\{\bm{w}\}$, where 
\begin{align}
    \ketbra{C_{\bm{w}}} = \prod_{j=1}^n \frac{\1 + (-1)^{\bm{w}_j} S_j}{2}, \quad 
    \bm{w} \in \bZ_2^n. 
\end{align}
Obviously, $\ket{C_{\bm{w}}}$ is also a stabilizer state for all $\bm{w}$~\cite{dangniam2020optimal}. 
And there is a subset $W\subseteq \bZ_2^n$, for all $\bm{w}\in W$, $\ket{C_{\bm{w}}}\in \cV$. 
In other word, for a fixed $\bm{w}\in W$, the first $n-k$ bits of it are all zeros. 
Then, we can define arbitrary state in $\cV^\bot$ as 
\begin{align}
    \ket{\Psi^\bot} = \sum_{\bm{w}\in W^\bot} \alpha_w \ket{C_w  }, \quad 
    \sum_{\bm{w}} |\alpha_{\bm{w}}|^2 = 1, 
\end{align}
and we have 
\begin{align}
    \bra{\Psi^\bot}\Omega_{{\rm II}, \mu}\ket{\Psi^\bot} 
    &= \sum_{i=1}^{n-k} \mu(S_i) \bra{\Psi^\bot}S_i^+\ket{\Psi^\bot} \\
    &= \sum_{i=1}^{n-k} \mu(S_i) \sum_{\bm{w}, \bm{w}'\in W^\bot} \alpha_{\bm{w}}^* \alpha_{\bm{w}'} \bra{C_{\bm{w}}}S_i^+\ket{C_{\bm{w}'}} \\
    &= \sum_{i=1}^{n-k} \mu(S_i) \sum_{\bm{w}, \bm{w}'\in W^\bot} \alpha_{\bm{w}}^* \alpha_{\bm{w}'} \delta_{\bm{w} \bm{w}'} \epsilon_{i, \bm{w}'} \\
    &= \sum_{i=1}^{n-k} \mu(S_i) \sum_{\bm{w}\in W^\bot} |\alpha_{\bm{w}}|^2 \epsilon_{i, \bm{w}}, 
\end{align}
where $W^\bot = \bZ_2^n \setminus W$, $S_i^+\ket{C_{\bm{w}}} = \epsilon_{i,\bm{w}} \ket{C_{\bm{w}}}$, and 
\begin{align}
    \epsilon_{i, \bm{w}} = \begin{cases}
        1 \quad \text{if $i$-th bit of $\bm{w}$ is 0} \\
        0 \quad \text{else}
    \end{cases}. 
\end{align}
Therefore, we have 
\begin{align}
    \max_{\Psi^\bot} \bra{\Psi^\bot}\Omega_{{\rm II}, \mu}\ket{\Psi^\bot} 
    &= \max_{\Psi^\bot} \sum_{\bm{w}\in W^\bot} |\alpha_{\bm{w}}|^2 \left(\sum_{i=1}^{n-k} \mu(S_i) \epsilon_{i,\bm{w}}\right) \\
    &= \max_{\bm{w}\in W^\bot} \sum_{i=1}^{n-k} \mu(S_i) \epsilon_{i, w} \\
    &= 1 - \min_{S_i} \mu(S_i), 
\end{align}
with the fact that for the first $k$ bits of $\bm{w}\in W^\bot$, there are at most $n-k-1$ bits equal to $0$. 
Then, we have 
\begin{align}
    \max_{\Psi^\bot} \bra{\Psi^\bot}\Omega_{{\rm II}, \mu}\ket{\Psi^\bot} 
    \geq 1 - \frac{1}{n-k}, 
\end{align}
the equality holds when $\mu(S_i) = \frac{1}{n-k}$ for all $i\in[n-k]$. 
Therefore, $\Omega_{\rm II}$ defined in Eq.~\eqref{eq:stabilizer protocol 2} is the optimal one, and we have 
\begin{align}
    \nu(\Omega_{\rm II}) = \frac{1}{n-k}. 
\end{align}

\section{Proof of graph subspace verification}
\label{app:proof of the graph subspace}

A graph subspace $\cV_{\mathscr{G}}$ is defined by a graph $G = (V, E)$ and an Abelian group $W = \langle\bm{w}_1, \cdots, \bm{w}_k\rangle$, 
such that the dimension of $\cV_{\mathscr{G}}$ is $2^k$. 
The graph subspace $\cV_{\mathscr{G}}$ has a set of orthogonal bases $\{G_{\bm{w}}: \bm{w}\in W\}$, where 
\begin{align}
    \ket{G_{\bm{w}}} := Z^{\bm{w}} \ket{G} = \prod_{i} Z_i^{w_i} \ket{G}. 
\end{align}
Each $\bm{w}\in W$ represents a logical qubit and connected to the $a$-th physical qubit if $a\in {\rm supp}(\bm{w})$. 
In the following, we will demonstrate that the problem of graph subspace verification can be reduced to the problem of coloring a graph with $n-k$ vertices. 

\subsection{The evolution of generators}

Recall that a graph $G$ with $n$ vertices has $n$ stabilizer generators. 
However, with the introduction of $W$, the number of stabilizer generators is reduced to $n-k$. 
This raises the question of which generators will be deleted or modified. 

With the constraints imposed by $W$, 
the set of generators must take the following form: 
\begin{align}
    \cG_{\mathscr{G}} &:= \{S^{\bm{y}}:  \forall\;i\in[k],|{\rm supp}(\bm{y}) \cap {\rm supp}(\bm{w}_i)| \text{ is even }\}
    \label{eq:new stabilizer generators}
\end{align}
where 
\begin{align}
    S^{\bm{y}} = \prod_{a\in{\rm supp}(\bm{y})} S_a. 
\end{align}
Specially, let $\{S^{\bm{0}}\} = \emptyset$. 
A set of operator $\{S^{\bm{y}}\}$ is said to be independent if $\{\bm{y}\}$ is a set of linear independent bases. 
We can then divide the stabilizer generators of the graph code $\mathscr{G}$ into the following two parts: 
\begin{align}
    \cG_{\mathscr{G}} &= \cG_{\overline{W}} \cup \cG_{W}.  
\end{align}
The first part, $\cG_{\overline{W}}$, can be expressed in two equivalent forms: 
\begin{align}
    \cG_{\overline{W}} 
    :&= \{S^{\bm{y}}: \forall\;i\in[k], {\rm supp}(\bm{y})\cap{\rm supp}(\bm{w}_i) = \emptyset \text{ and } |{\rm supp}(\bm{y})|=1\} 
\end{align}
or 
\begin{align}
   \cG_{\overline{W}} = \{S_a:a\in \overline{W}\}, 
\end{align}
where $\overline{W} = \{a:a\not\in {\rm supp}(\bm{w}_i),\;\forall\;i\in[k]\}$,
representing the set of physical qubits that are not connected to any logical qubits. Therefore, the set $\cG_{\overline{W}}$ can be directly obtained from $G$ and $W$. 
Next, we focus on constructing $\cG_W$, which will allow us to fully determine the set of generators for the graph code $\mathscr{G}$. 

\begin{figure}[!htbp]
    \centering
    \includegraphics[width=0.7\linewidth]{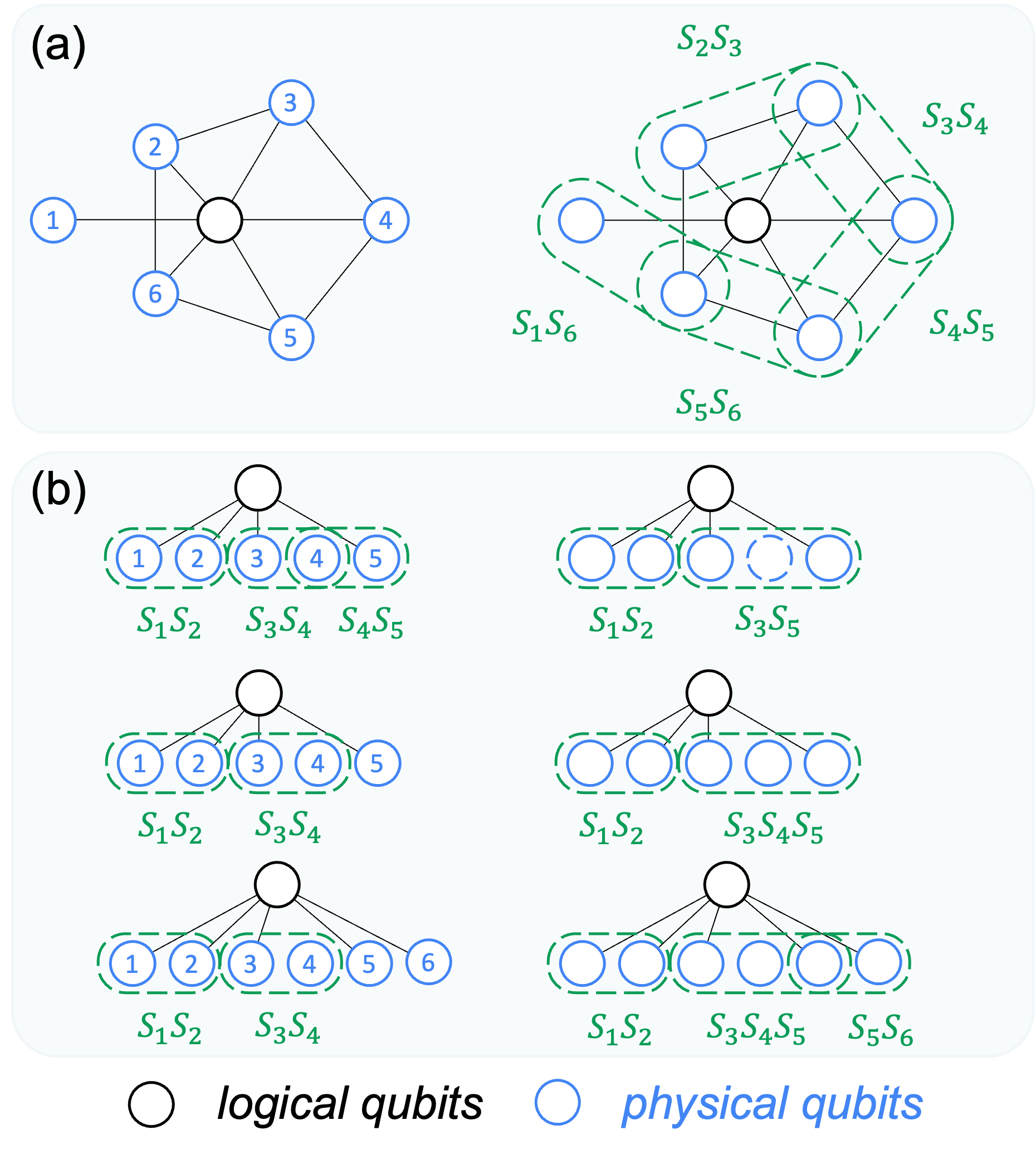}
    \caption{\raggedright 
    Examples of the evolution of generators. 
    (a) $k=1$ case. The left figure illustrates the original graph codes, while the right figure shows the generators corresponding to these original graph codes. 
    Note that the generators defined here differ from those described in the main text; 
    specifically, we do not require ``ascending order'' here.
    (b) $k\geq 2$ case. 
    The left figures depict the generators of $\cS_{G, W'}$, and we examine the impact of the vertex $\bm{w}_k$. 
    The right figures display the new generators obtained after incorporating $\bm{w}_k$. 
    }
    \label{fig:appendix-graph code-step1}
\end{figure}

\begin{itemize}
    \item For the $k=1$ case, the situation is straightforward. 
    If $|{\rm supp}(\bm{w}_1)| = 1$, then $\cS_W = \emptyset$. 
    Otherwise, we construct a set of independent generators $\{S^{\bm{y}}\}$, 
    where $|{\rm supp}(\bm{y})| = 2$ and ${\rm supp}(\bm{y}) \cap \overline{W} = \emptyset$ for each $\bm{y}$. 
    This construction is visualized in Fig.~\ref{fig:appendix-graph code-step1}(a). 
    \item For the $k\geq 2$ case, suppose we have constructed a set of generators with 
    $W' = \langle \bm{w}_1, \cdots, \bm{w}_{k-1}\rangle$, denoted by $\cG_{G, W'} = \cG_{\overline{W'}} \cup \cG_{W'}$. 
    We then need to consider the influence of $\bm{w}_k$, resulting in the following cases:
    \begin{enumerate}
        \item If $|{\rm supp}(\bm{w}_k) \cap \overline{W'}| = 0$, i.e., $\overline{W} = \overline{W'}$, 
        \begin{align}
            \cG_{G, W'}  = \cG_{\overline{W'}} \cup \cG_{W'} 
            \xrightarrow[]{\bm{w}_k} 
            \cG_{\overline{W}} \cup \cG_{W}^o \cup \cG_{W}^e, 
        \end{align}
        where 
        \begin{align}
            \cG_{W}^e 
            &:= \left\{S^{\bm{y}}\in \cG_{W'}: |{\rm supp}(\bm{y}) \cap {\rm supp}(\bm{w}_k)|\text{ is even}\right\} \\
            \cG_{W}^o 
            &:= \left\{ S^{\bm{y}} S^{\bm{y}'}: S^{\bm{y}},S^{\bm{y}'} \in \cG_{W'}\setminus\cG_W^e\right\}. 
        \end{align}
        Note that $\cG_W^o$ contains $|\cG_{W'}\setminus\cG_W^e|-1$ independent generators. 
        \item If $|{\rm supp}(\bm{w}_k) \cap \overline{W'}| = 1$, 
        suppose ${\rm supp}(\bm{w}_k) \cap \overline{W'} = \{a\}$, 
        \begin{align}
            \cG_{G, W'}  = \cG_{\overline{W'}} \cup \cG_{W'} 
            \xrightarrow[]{\bm{w}_k} 
            \cG_{\overline{W}} \cup \cG_{W}^o \cup \cG_{W}^e. 
        \end{align}
        where 
        \begin{align}
            \cG_{W}^e 
            &:= \left\{S^{\bm{y}}\in \cG_{W'}: |{\rm supp}(\bm{y}) \cap {\rm supp}(\bm{w}_k)|\text{ is even}\right\} \\
            \cG_{W}^o 
            &:= \left\{ S_a S^{\bm{y}}: S^{\bm{y}} \in \cG_{W'}\setminus\cG_W^e\right\}. 
        \end{align}
        \item If $|{\rm supp}(\bm{w}_k) \cap \overline{W'}| \geq 2$, 
        select one vertex $a\in {\rm supp}(\bm{w}_k) \cap \overline{W'}$, 
        \begin{align}
            \cG_{G, W'}  = \cG_{\overline{W'}} \cup \cG_{W'} 
            \xrightarrow[]{\bm{w}_k} 
            \cG_{\overline{W}} \cup \cG_{\overline{W'}\setminus\overline{W}} \cup 
            \cG_{W}^o \cup \cG_{W}^e. 
        \end{align}
        where 
        \begin{align}
            \cG_{\overline{W'}\setminus\overline{W}}
            &:= \{S^{\bm{y}}: |{\rm supp}(\bm{y})| = 2 \text{ and } {\rm supp}(\bm{y}) \subset \overline{W'}\setminus\overline{W}\} \\
            \cG_{W}^e 
            &:= \left\{S^{\bm{y}}\in \cG_{W'}: |{\rm supp}(\bm{y}) \cap {\rm supp}(\bm{w}_k)|\text{ is even}\right\} \\
            \cG_{W}^o 
            &:= \left\{ S_a S^{\bm{y}}: S^{\bm{y}} \in \cG_{W'}\setminus\cG_W^e\right\}. 
        \end{align}
        Note that $\cG_{\overline{W'}\setminus\overline{W}}$ contains independent $|\overline{W'}\setminus\overline{W}|-1$ generators. 
    \end{enumerate}
    This construction can be visualized in Fig.~\ref{fig:appendix-graph code-step1}(b). 
\end{itemize}
As observed, for each generator associated with $W$, the number of generators in the graph code is reduced by $1$. 
Consequently, we ultimately obtain $n-k$ generators. 

\subsection{Graph structure of the graph code}

In this subsection, we analyze the graph structure of the graph code.
We define a new graph $G' = (V', E')$, where $|V'| = n-k$, and each vertex represents a stabilizer generator, as defined in Eq.~\eqref{eq:new stabilizer generators}. 
We now need to determine how to connect these vertices, labeled as $\{\bm{y}\}$. 
To clarify our reconnection rules, we first explain the physical meaning of an edge in the graph. 

For a graph state $\ket{G}$ and its corresponding graph $G$, 
$a$-th and $b$-th qubits are connected in $G$ (i.e., $(a,b)\in E$), 
if and only if the corresponding stabilizer generators $S_a$ and $S_b$ anti-commute on these two qubits. 
Similarly, in the graph $G'$, $(\bm{y},\bm{y}')\in E'$ if and only if $S^{\bm{y}}$ and $S^{\bm{y}'}$ anti-commute on the related qubits. This can be verified as follows:
\begin{enumerate}
    \item For $a\in{\rm supp}(\bm{y})$ and $\bm{y}'$, define the function
        \begin{align}
            f(a, \bm{y'}) 
            := |\{b\in{\rm supp}(\bm{y'}): (a,b)\in E\}|. 
            \label{eq:edge count in graph}
        \end{align}
        This function counts the number of $b\in{\rm supp}(\bm{y'})$ that are connected to $a$ in $G$. 
    \item $(\bm{y},\bm{y}')\in E'$ if one of the following conditions holds: 
    \begin{enumerate}
        \item $\exists\;a\in {\rm supp}(\bm{y})\cap{\rm supp}(\bm{y}')$ such that $f(a,\bm{y}) + f(a,\bm{y}')$ is odd; 
        \item $\exists\;a\in {\rm supp}(\bm{y})$ but $\not\in{\rm supp}(\bm{y}')$ such that $f(a,\bm{y}')$ is odd; 
        \item $\exists\;a\in {\rm supp}(\bm{y}')$ but $\not\in{\rm supp}(\bm{y})$ such that $f(a,\bm{y})$ is odd. 
    \end{enumerate}
\end{enumerate}

Specially, for the $k=1$ case, the rules of connection can be summarized as follows. 
Define $\delta_{a,b} = 1$ if $(a,b)\in E$; otherwise $\delta_{a,b}=0$. 
The edges in $G'$ are then determined as follows, 
\begin{enumerate}
    \item If ${\rm supp}(\bm{y}_1) = \{a\}$ and ${\rm supp}(\bm{y}_2) = \{b\}$, $(\bm{y}_1,\bm{y}_2)\in E'$ if $\delta_{a,b}=1$; 
    \item If ${\rm supp}(\bm{y}_1) = \{a\}$ and ${\rm supp}(\bm{y}_2) = \{b,\bar{b}\}$, $(\bm{y}_1, \bm{y}_2)\in E'$ if $\delta_{a,b} = 1$ or $\delta_{a,\bar{b}}=1$; 
    \item If ${\rm supp}(\bm{y}_1) = \{a,\bar{a}\}$ and ${\rm supp}(\bm{y}_2) = \{b,\bar{b}\}$, $(\bm{y}_1,\bm{y}_2)\not\in E'$ if at least one of the following two conditions holds: 
     \begin{enumerate}
        \item ${\rm supp}(\bm{y}_1)\cap {\rm supp}(\bm{y}_2)=\emptyset$ and $\delta_{a,b} = \delta_{a,\bar{b}} = \delta_{\bar{a},b} = \delta_{\bar{a},\bar{b}}$; 
        \item ${\rm supp}(\bm{y}_1)\cup {\rm supp}(\bm{y}_2)=\{a,b,c\}$ and $\delta_{a,b} = \delta_{a,c} = \delta_{b,c}$. 
    \end{enumerate}
\end{enumerate}
\begin{figure}[!htbp]
    \centering
    \includegraphics[width=0.8\linewidth]{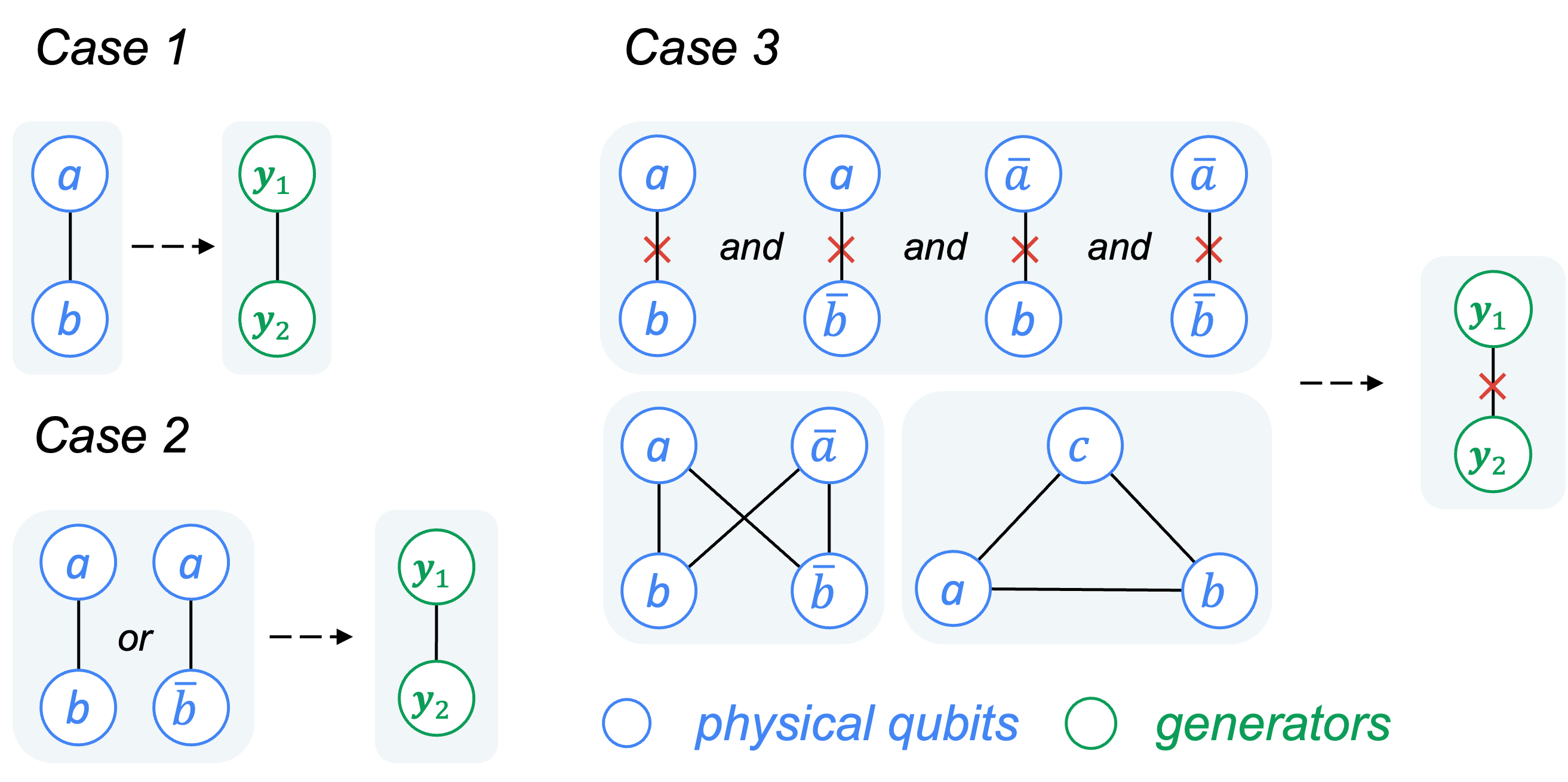}
    \caption{\raggedright The connection rules of each two vertices when $k=1$. 
    There are three cases. 
    For each case, the edges in the new graph $G'$ (right one) is determined by the edges in the original graph $G$ (left one).
    }
    \label{fig:appendix-graph code-step2}
\end{figure}

\subsection{Coloring strategy}  

With the above analysis, the graph subspace $\cV$ can be described using a new graph with $n-k$ vertices. 
Inspired by the coloring strategy initially applied to hypergraph states~\cite{zhu2019efficientb}, we propose a verification strategy for this graph subspace, which we term the \emph{coloring strategy}. This strategy is based on the coloring problem of the graph. 
Let $\mathscr{A} = \{A_1, A_2, \cdots, A_m\}$ be an independence cover of $G'$,  composed of $m$ nonempty independent sets, such that 
$\bigcup_\ell A_\ell = \{\bm{y}\}$ and $A_\ell\cap A_\iota = \emptyset$ for all $\ell\neq \iota$. 
For each set $A_\ell$, we construct a test operator $P_\ell$ as follows:
\begin{enumerate}
    \item For $\bm{y}\in A_\ell$ and $a\in{\rm supp}(\bm{y})$, 
    if $f(a,\bm{y})$ is even, perform an $X$ measurement on the $a$-th qubit; 
    otherwise, perform a $Y$ measurement on the $a$-th qubit, 
    where $f(\cdot)$ is defined in Eq.~\eqref{eq:edge count in graph}.
    \item For other qubits, perform $Z$ measurements. 
    \item The state passes only if the measurement result of $S^{\bm{y}}$ is $+1$ for all $\bm{y}\in A_\ell$. 
\end{enumerate}
Therefore, the corresponding test operator is defined as 
\begin{align}
    P_\ell = \prod_{\bm{y}\in A_\ell} \frac{\1 + S^{\bm{y}}}{2}. 
\end{align}
Suppose the test $P_\ell$ is performed with probability $\mu_\ell$. 
Our coloring strategy is characterized by the weighted independence cover $(\mathscr{A}, \mu)$. 
The corresponding verification operator is given by 
\begin{align}
    \Omega(\mathscr{A}, \mu) = \sum_{\ell} \mu_\ell P_\ell. 
    \label{eq:verification operator of coloring protocol}
\end{align}

\subsubsection{Single logical qubit case} 

\begin{figure}[!htbp]
    \centering
    \includegraphics[width=0.8\linewidth]{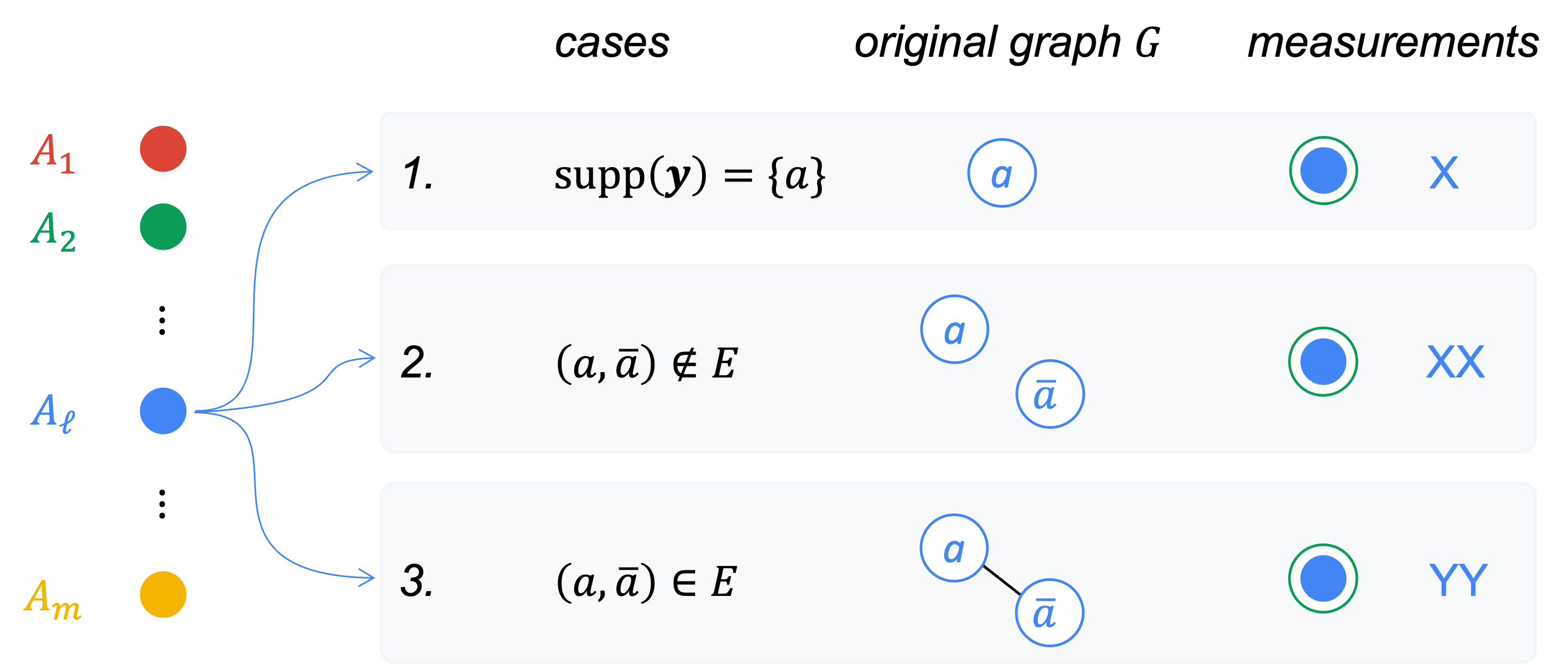}
    \caption{\raggedright The measurement settings for different independent set $A_\ell$ when $k=1$. 
    For each $A_\ell$, The measurement settings are depend on original graph $G$, as shown in black vertices.
    We perform $X$ measurements in first two cases, and perform $Y$ measurements in the last one. 
    For other not involved qubits, we perform $Z$ measurements. 
    }
    \label{fig:appendix-graph code-step3}
\end{figure}

Here, we show the coloring strategy for $k=1$ case concretely. 
For each independent set $A_\ell$, we can construct the following measurement setting, 
\begin{enumerate}
    \item If ${\rm supp}(\bm{y})=\{a\}$, perform $X$ measurement on the $a$-th qubit; 
    \item If ${\rm supp}(\bm{y})=\{a,\bar{a}\}$ and $\delta_{a,\bar{a}}=1$, 
    perform $Y$ measurements on $a$-th and $\bar{a}$-th qubits; 
    If ${\rm supp}(\bm{y})=\{a,\bar{a}\}$ and $\delta_{a,\bar{a}}=0$, 
    perform $X$ measurements on these two qubits. 
    \item Perform $Z$ measurement on other not involved qubits.
\end{enumerate}
A visual illustration can be found in Fig.~\ref{fig:appendix-graph code-step3}.

\subsection{Efficiency of the coloring strategy}

To analyze the efficiency of the coloring strategy, we first introduce the following lemma, which will be instrumental in the subsequent analysis. 

\begin{lemma}\label{lemma:product of commute operators}
    Given a set of stabilizer generators $\{S_1, \cdots, S_m\}$ with $[S_i, S_j] = 0$ for all $i,j\in[m]$, we have 
    \begin{align}
        \left[\prod_{i=1}^m \left(\1 + S_i\right)\right]
        \left(\1 + \prod_{i=1}^m S_i\right) 
        = 2 \prod_{i=1}^m \left(\1 + S_i\right). 
    \end{align}
\end{lemma}
\begin{proof}
    For $m=1$ case, we have 
    \begin{align}
        \left(\1 + S_1\right) \left(\1 + S_1\right) = 2 \left(\1 + S_1\right). 
    \end{align}
    Thus, this lemma holds when $m=1$. 
    Now, assume that this lemma holds when $m=k-1$, we want to prove that it is hold for $m=k$ case. Then, for $m=k$ case, we have 
    \begin{align}
        \left[\prod_{i=1}^k \left(\1 + S_i\right)\right]
        \left(\1 + \prod_{i=1}^k S_i\right) 
        &= \left[\prod_{i=1}^{k-1} \left(\1 + S_i\right)\right]
        \left(\1 + S_k\right)
        \left(\1 + \prod_{i=1}^{k-1} S_i \cdot S_k\right) \\
        &= \left[\prod_{i=1}^{k-1} \left(\1 + S_i\right)\right]
        \left(\1 + S_k\right) S_k 
        \left(S_k + \prod_{i=1}^{k-1} S_i\right) \\
        &= \left[\prod_{i=1}^{k-1} \left(\1 + S_i\right)\right]
        \left(S_k + \prod_{i=1}^{k-1} S_i\right)
        \left(\1 + S_k\right) \\
        &= \left[\prod_{i=1}^{k-1} \left(\1 + S_i\right)\right]
        \left[\left(\1 + \prod_{i=1}^{k-1} S_i\right) + \left(S_k - \1\right) \right]
        \left(\1 + S_k\right) \\
        &= 2 \prod_{i=1}^{k}\left(\1 + S_i\right), 
    \end{align}
    with the fact that $(S_k - \1)(\1 + S_k) = 0$. 
    Therefore, the lemma holds for $m=k$ case. 
\end{proof}

Similar to previous discussions, we first prove the condition for perfect completeness.
For all $\ell\in[m]$, we have 
\begin{align}
    \tr\left[\Pi_{\cV} P_\ell\right] 
    &= \tr\left[\left(\prod_{\bm{y}} \frac{\1 + S^{\bm{y}}}{2}\right)
    \left(\prod_{\bm{y}\in A_\ell}\frac{\1+S^{\bm{y}}}{2}\right)\right] \\
    &= \tr[\Pi_{\cV}] = \rank(\Pi_{\cV}), 
\end{align}
where we use Lemma~\ref{lemma:product of commute operators}. 
Then, we consider the spectral gap of $\Omega(\mathscr{A}, \mu)$ defined in Eq.~\eqref{eq:verification operator of coloring protocol}. 
We define a set of stabilizer generators 
\begin{align}
    \cK = \{\underbrace{K_1, \cdots, K_{n-k}}_{\text{generators of }\mathscr{G}, \{S^{\bm{y}}\}}, 
    K_{n-k+1}, \cdots, K_n\}. 
\end{align}
Based on these generators, we can define a set of orthogonal states $\{\ket{G_{\bm{x}}}\}$, where 
\begin{align}
    \ketbra{G_{\bm{x}}} = \prod_{i}^n \frac{\1 + (-1)^{\bm{x}_i}K_i}{2}. 
\end{align}
The state $\ketbra{G_{\bm{x}}}\in \mathscr{D}(\cV)$ if and only if ${\rm supp}(\bm{x})\cap[n-k] = \emptyset$. 
Define 
\begin{align}
    A_\ell' := \{i: S^{\bm{y}}=K_i, \bm{y}\in A_\ell\} \text{ with }
    \bigcup_{\ell} A_\ell' = [n-k] \text{ and }
    A_\ell \cap A_\iota = \emptyset \text{ for all $\ell\neq\iota$}. 
\end{align}
For a fixed $P_\ell$, we have  
\begin{align}
    \bra{G_{\bm{x}}}P_\ell\ket{G_{\bm{x}}} 
    &= \tr\left[
        \left(\prod_{i=1}^n \frac{\1 + (-1)^{\bm{x}_i}K_i}{2}\right)
        \left(\prod_{\bm{y}\in A_\ell} \frac{\1 + S^{\bm{y}}}{2}\right)
    \right] \\
    &= \tr\left[
        \left(\prod_{i=1}^n \frac{\1 + (-1)^{\bm{x}_i}K_i}{2}\right)
        \left(\prod_{i\in A_\ell'} \frac{\1 + K_i}{2}\right)
    \right] \\
    &= \begin{cases}
        1, & \text{if ${\rm supp}(\bm{x})\subseteq \bar{A}_\ell'$} \\
        0, & \text{otherwise}
    \end{cases}, 
\end{align}
where $\bar{A}_\ell' = [n] \setminus A_\ell'$. 
Therefore, for a fixed $\bm{x}$, we have 
\begin{align}
    \bra{G_{\bm{x}}}\Omega(\mathscr{A}, \mu)\ket{G_{\bm{x}}} 
    = \bra{G_{\bm{x}}}\sum_\ell \mu_\ell P_\ell\ket{G_{\bm{x}}} 
    = \sum_{\ell: {\rm supp}(\bm{x})\subseteq \bar{A}_\ell'} \mu_\ell. 
\end{align}
Define arbitrary pure state in $\cV^\bot$ as 
\begin{align}
    \ket{\Psi^\bot} = \sum_{\bm{x}:{\rm supp}(\bm{x})\cap[n-k]\neq\emptyset} 
    \alpha_{\bm{x}} \ket{G_{\bm{x}}}, \quad 
    \sum |\alpha_{\bm{x}}|^2 = 1, 
\end{align}
then the spectral gap can be represented as 
\begin{align}
    \nu(\Omega(\mathscr{A}, \mu))
    &= 1 - \max_{\ket{\Psi^\bot}} \bra{\Psi^\bot} \Omega(\mathscr{A}, \mu) \ket{\Psi^\bot} \\
    &= 1 - \max_{\{\alpha_{\bm{x}}\}} \left(\sum_{\bm{x}} 
    |\alpha_{\bm{x}}|^2 
    \bra{G_{\bm{x}}} \Omega(\mathscr{A}, \mu)\ket{G_{\bm{x}}}
    \right) \\
    &= 1 - \max_{\bm{x}} \left(
        \sum_{\ell: {\rm supp}(\bm{x})\subseteq \bar{A}_\ell} \mu_\ell\right). 
\end{align}
For the above maximization problem, it suffices to consider the case in which $|{\rm supp}(\bm{x})\cap[n-k]| = 1$. 
Therefore, 
\begin{align}
    \nu(\Omega(\mathscr{A}, \mu)) 
    = 1 - \max_{\ell} \left(1-\mu_\ell\right)
    = \min_{\ell} \mu_{\ell}. 
\end{align}
With the previous analysis, we want to maximize the spectral gap, 
\begin{align}
    \max_{\mu} \nu(\Omega(\mathscr{A}, \mu)) 
    = \max_{\mu}\min_{\ell} \mu_\ell = \frac{1}{m}, 
\end{align}
when $\mu_1 = \mu_2 = \cdots = \mu_m = \frac{1}{m}$. 
Therefore, we can define  
\begin{align}
    \Omega_{\mathscr{G}} := \frac{1}{m} \sum_{\ell} P_\ell, \quad 
    \nu(\Omega_{\mathscr{G}}) = \frac{1}{m}. 
\end{align}

\subsection{Examples of graph subspace verification}
\label{app:examples of graph subspace verification}

We present specific verification strategies for subspaces induced by various well-known graph codes: 
the $[\![4, 1, 2]\!]$ code, the five-qubit code, the $[\![7, 1, 3]\!]$ code, and the $[\![8, 3, 3]\!]$ code. 

\subsubsection{The $[\![4, 1, 2]\!]$ graph code subspace}
An $[\![4, 1, 2]\!]$ graph code is illustrated in Fig.~\ref{fig:graph example 412}. 
Obviously, we have $W = \langle \bm{w} = 1100\rangle$, and the set of generators is 
\begin{align}
    \cG_{\mathscr{G}} = \{S_1S_2, S_3, S_4\}. 
\end{align}

\begin{figure}[!htbp]
    \centering
    \includegraphics[width=0.8\linewidth]{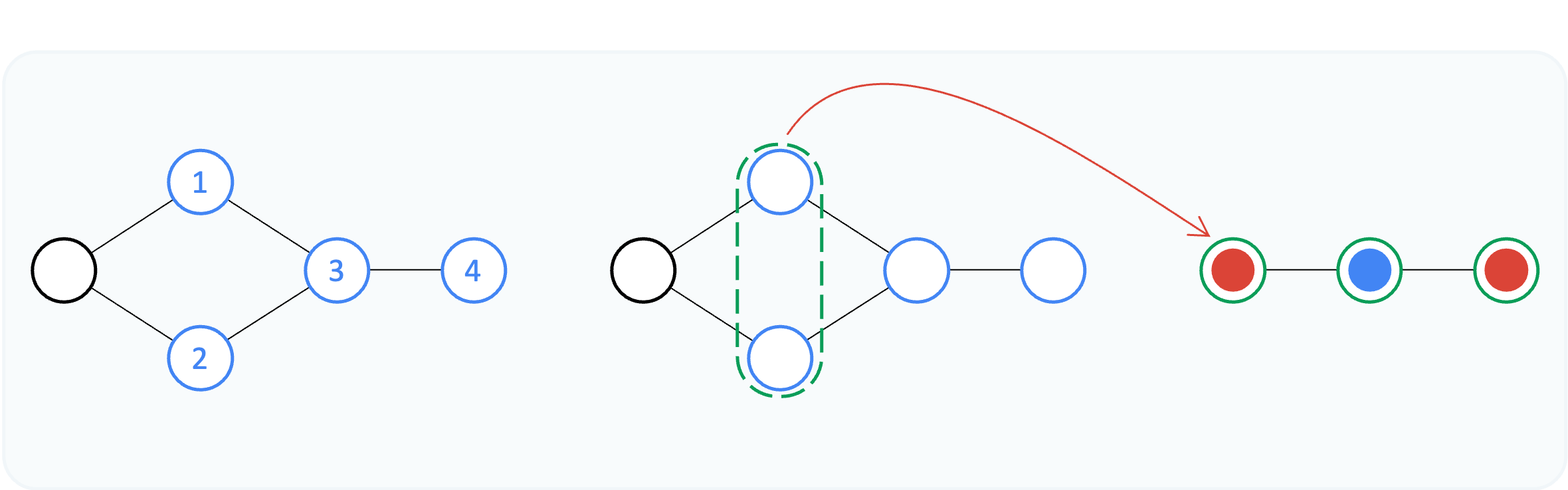}
    \caption{The coloring strategy for the $[\![4, 1, 2]\!]$ graph code. }
    \label{fig:graph example 412}
\end{figure}

As depicted in Fig.~\ref{fig:graph example 412}, a new graph can be constructed where the vertices represent the generators. 
Using a coloring strategy, we can verify the corresponding graph subspace with $2$ measurement settings: 
\red{\texttt{XXZX}} and \blue{\texttt{ZZXZ}}. 
The corresponding test operators are 
\begin{align}
    \red{P_1 = \left(\frac{\1 + S_1S_2}{2}\right)\left(\frac{\1 + S_4}{2}\right)}, \quad 
    \blue{P_2 = \left(\frac{\1 + S_3}{2}\right)}. 
\end{align}
Finally, the verification operator is 
\begin{align}
    \Omega_{\mathscr{G}} = \frac{1}{2} \sum_{\ell = 1}^2 P_\ell. 
\end{align}

\subsubsection{The five-qubit graph code subspace}

The five-qubit code is the smallest code capable of correcting single-qubit errors on a logical qubit~\cite{baccari2020deviceindependent}. 
It is also a graph code, as illustrated in Fig.~\ref{fig:five-qubit code}, 
where blue circles represent physical qubits, and the black circle represents the logical qubit. 
From the graph, we define $W = \langle \bm{w}=11111 \rangle$, and the set of generators is 
\begin{align}
    \cG_{\mathscr{G}} = \{S_1S_2, S_2S_3, S_3S_4, S_4S_5\}. 
\end{align}

\begin{figure}[!htbp]
    \centering
    \includegraphics[width=0.7\linewidth]{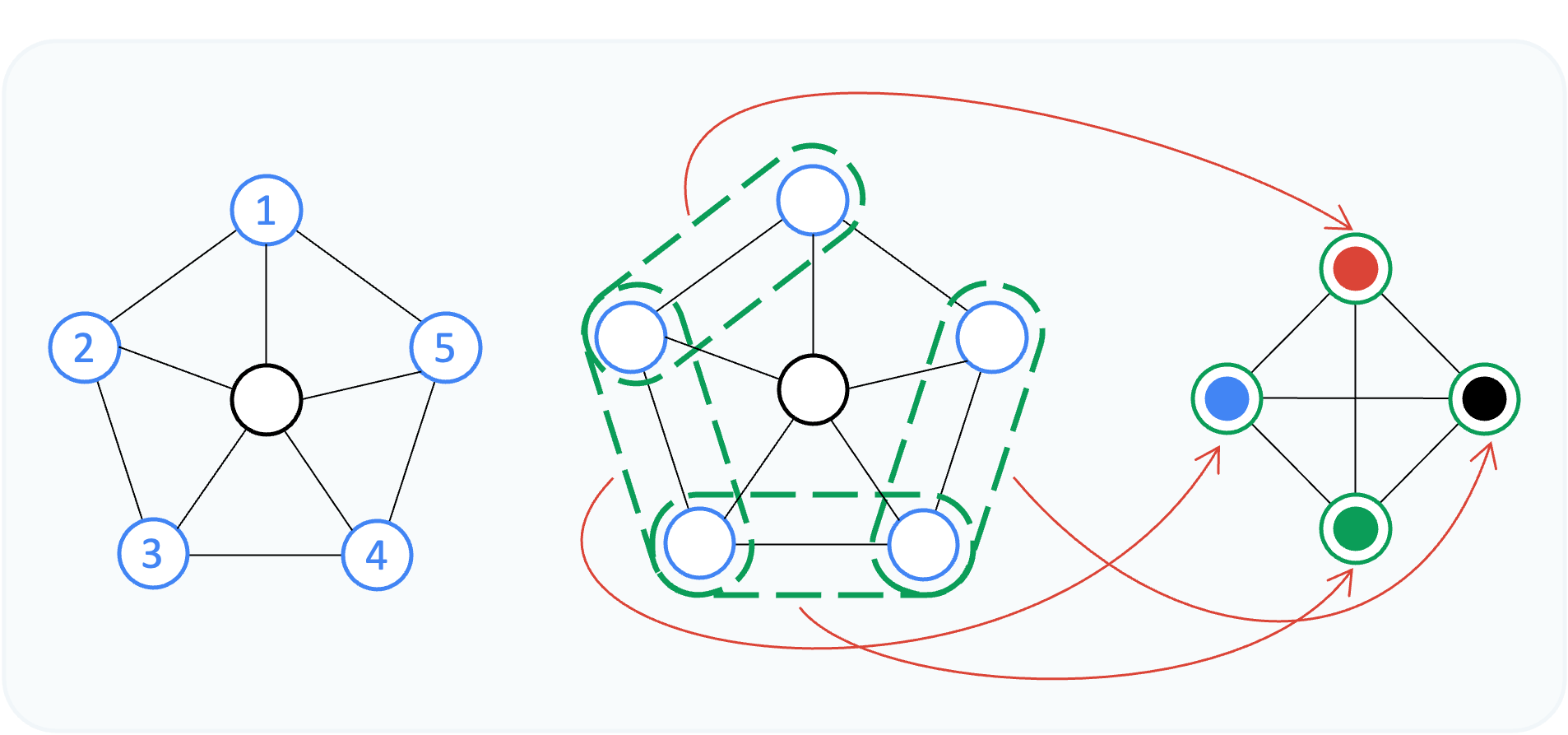}
    \caption{The coloring strategy for the five-qubit code subspace. }
    \label{fig:five-qubit code}
\end{figure}

Next, we define a new graph $G'$, where each vertex corresponds to a generator. 
Following the previous analysis, we implement a coloring strategy with $m=4$ measurement settings: 
\red{\texttt{YYZZZ}}, \blue{\texttt{ZYYZZ}}, \green{\texttt{ZZYYZ}}, and \texttt{ZZZYY}. 
The corresponding test operators are 
\begin{align}
    \red{P_1 = \frac{\1 + S_1S_2}{2}}, \quad 
    \blue{P_2 = \frac{\1 + S_2S_3}{2}}, \quad 
    \green{P_3 = \frac{\1 + S_3S_4}{2}}, \quad 
    P_4 = \frac{\1 + S_4S_5}{2}. 
\end{align}
Finally, the verification operator is given by 
\begin{align}
    \Omega_{\mathscr{G}} 
    = \frac{1}{4}\sum_{\ell=1}^4 P_\ell. 
\end{align}

\subsubsection{The $[\![7, 1, 3]\!]$ graph code subspace}

An $[\![7, 1, 3]\!]$ graph code is illustrated in Fig.~\ref{fig:graph example 713}. 
Obviously, we have $W = \langle \bm{w} = 1010100\rangle$, and the set of generators is 
\begin{align}
    \cG_{\mathscr{G}} = \{S_1S_3, S_1S_5, S_2, S_4, S_6, S_7\}. 
\end{align}

\begin{figure}[!htbp]
    \centering
    \includegraphics[width=0.8\linewidth]{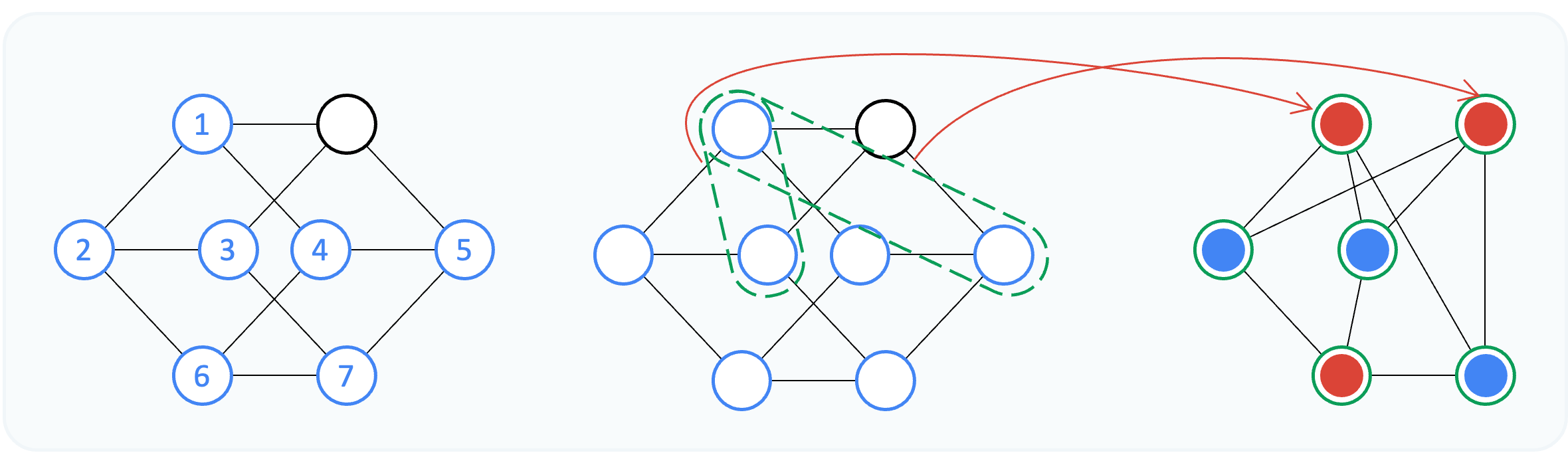}
    \caption{The coloring strategy for the $[\![7, 1, 3]\!]$ graph code. }
    \label{fig:graph example 713}
\end{figure}

As depicted in Fig.~\ref{fig:graph example 713}, a new graph can be constructed where the vertices represent the generators. 
Using a coloring strategy, we can verify the corresponding graph subspace with $2$ measurement settings: 
\red{\texttt{XZXZXXZ}} and \blue{\texttt{ZXZXZZX}}. 
The corresponding test operators are 
\begin{align}
    \red{P_1 = \left(\frac{\1 + S_1S_4}{2}\right)\left(\frac{\1 + S_1S_5}{2}\right)\left(\frac{\1 + S_6}{2}\right)}, \quad 
    \blue{P_2 = \left(\frac{\1 + S_2}{2}\right)\left(\frac{\1 + S_4}{2}\right)\left(\frac{\1 + S_7}{2}\right)}. 
\end{align}
Finally, the verification operator is 
\begin{align}
    \Omega_{\mathscr{G}} = \frac{1}{2} \sum_{\ell = 1}^2 P_\ell. 
\end{align}

\subsubsection{The $[\![8, 3, 3]\!]$ graph code subspace}

Here, we consider a subspace induced by an $[\![8, 3, 3]\!]$ graph code~\cite{cafaro2014schemea}. 
Similarly, we first obtain the generators of this code. 
Following the steps illustrated in Fig.~\ref{fig:graph example generators}(a), we derive the set of generators: 
\begin{align}
    \cG_{\mathscr{G}} = \{S_1S_8, S_1S_4S_6S_7S_8, S_2S_7, S_3S_6, S_4S_5\}.
\end{align}

\begin{figure}[!htbp]
    \centering
    \includegraphics[width=0.8\linewidth]{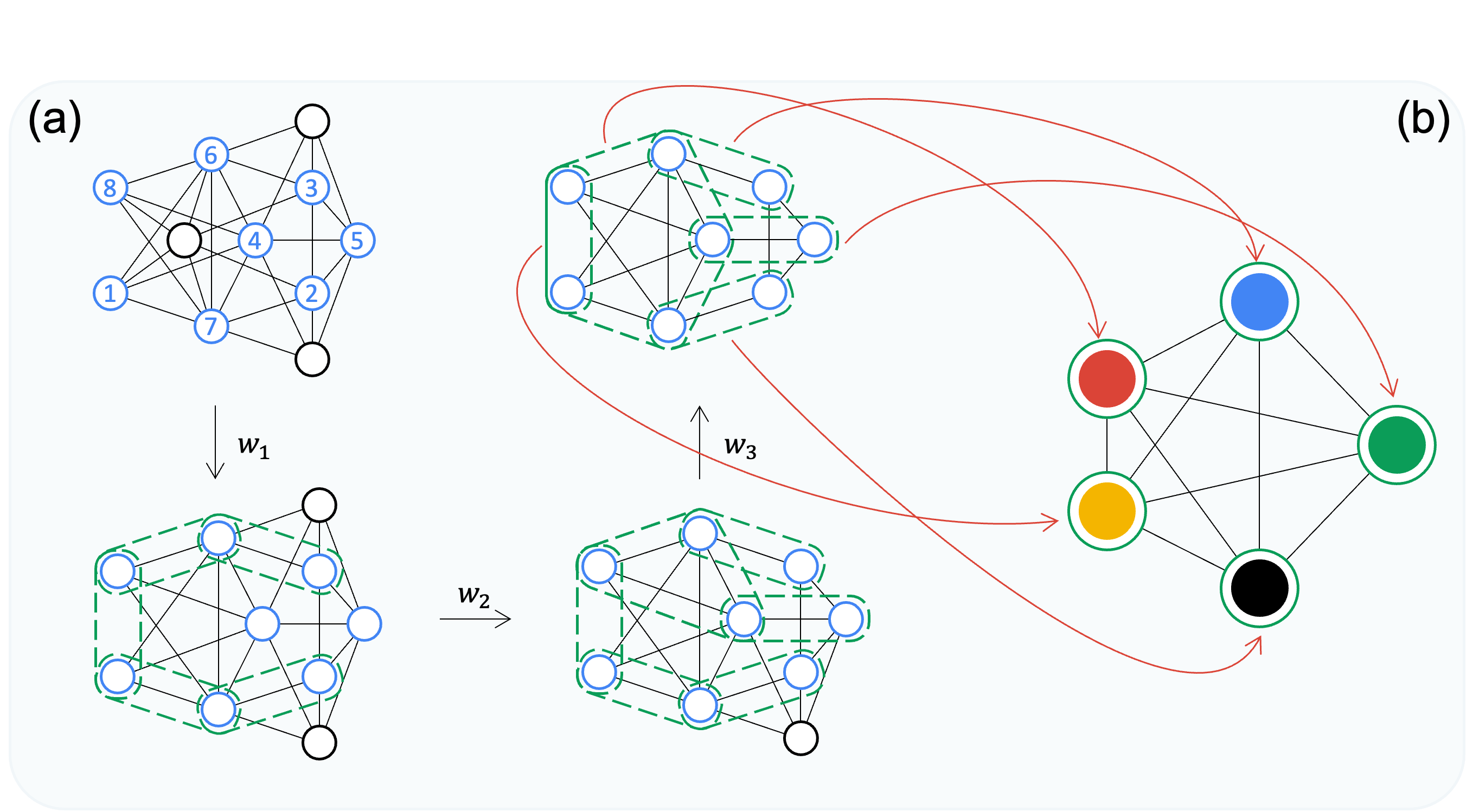}
    \caption{The coloring strategy for the $[\![8, 3, 3]\!]$ graph code. 
    (a) The generator of this graph code. 
    (b) Coloring strategy. }
    \label{fig:graph example generators}
\end{figure}

Using the new graph $G'$, shown in Fig.~\ref{fig:graph example generators}(b), 
we apply a coloring strategy with $m=5$ measurement settings: 
\red{\texttt{YZZXZXXY}}, \blue{\texttt{ZZYZZYZZ}}, \green{\texttt{ZZZYYZZZ}}, \texttt{ZYZZZZYZ}, and \yellow{\texttt{XZZZZZZX}}. 
The corresponding test operators are 
\begin{align}
    \red{P_1 = \frac{\1 + S_1S_4S_6S_7S_8}{2}}, \quad 
    \blue{P_2 = \frac{\1 + S_3S_6}{2}}, \quad 
    \green{P_3 = \frac{\1 + S_3S_4}{2}}, \quad 
    P_4 = \frac{\1 + S_2S_7}{2}, \quad 
    \yellow{P_5 = \frac{\1 + S_1S_8}{2}}. 
\end{align}
Finally, the verification operator is given by 
\begin{align}
    \Omega_{\mathscr{G}} 
    = \frac{1}{5}\sum_{\ell=1}^5 P_\ell. 
\end{align}

\section{Proof of CSS subspace verification}
\label{app:proof of the CSS subspace}

A CSS code can be described by two matrices 
$H_X\in\bZ_2^{k_X\times n}$ and $H_Z\in\bZ_2^{k_Z\times n}$, 
where $H_X H_Z^T = \bm{0}$~\cite{mackay2004sparse}. 
There are $k_X+k_Z$ stabilizer generators: $X^{\bm{c}_X}$ and $Z^{\bm{c}_Z}$, 
where $P^{\bm{c}}:=\bigotimes_{i\in[n]} P^{c_i}$, 
$\bm{c}_X$ and $\bm{c_Z}$ are the rows of $H_X$ and $H_Z$ respectively. 
Such a CSS code can determine a CSS subspace $\cV_{\rm CSS}$ with dimension $2^k$ with $k:= n-(k_X+k_Z)$, which projector is 
\begin{align}
    \Pi_{\rm CSS} := 
    \left(\prod_{\bm{c}_X} \frac{\1 + X^{\bm{c}_X}}{2}\right) 
    \left(\prod_{\bm{c}_Z} \frac{\1 + Z^{\bm{c}_Z}}{2}\right). 
\end{align}
The corresponding verification operator is defined as 
\begin{align}
    \Omega_{\rm CSS}^{\rm XZ} = \frac{1}{2} 
    \left(\prod_{\bm{c}_X} \frac{\1 + X^{\bm{c}_X}}{2} + 
    \prod_{\bm{c}_Z} \frac{\1 + Z^{\bm{c}_Z}}{2}\right). 
\end{align}
Firstly, we prove the perfect completeness condition defined in Lemma~\ref{lemma:perfect completeness condition} as follows, 
\begin{align}
    \tr[\Pi_{\rm CSS} \Omega_{\rm CSS}^{\rm XZ}] 
    &= \frac{1}{2} \tr\left[
    \left(\prod_{\bm{c}_X} \frac{\1 + X^{\bm{c}_X}}{2}\right) 
    \left(\prod_{\bm{c}_Z} \frac{\1 + Z^{\bm{c}_Z}}{2}\right) 
    \left(\prod_{\bm{c}_X} \frac{\1 + X^{\bm{c}_X}}{2} + 
    \prod_{\bm{c}_Z} \frac{\1 + Z^{\bm{c}_Z}}{2}\right)
    \right] \\
    &= \frac{1}{2} \tr\left[
    \left(\prod_{\bm{c}_X} \frac{\1 + X^{\bm{c}_X}}{2}\right) 
    \left(\prod_{\bm{c}_Z} \frac{\1 + Z^{\bm{c}_Z}}{2}\right) 
    \right]
    = \rank(\Pi_{\rm CSS}). 
\end{align}
Then, we analyze the spectral gap of $\Omega_{\rm CSS}^{\rm XZ}$. 
We define a set of stabilizer generators with size $n$, 
\begin{align}
    \cG_{n} = \{\underbrace{S_1, \cdots, S_{k_X}}_{\{X^{\bm{c}_X}\}_{\bm{c}_X}}, \underbrace{S_{k_X+1}, \cdots, S_{n-k}}_{\{Z^{\bm{c}_Z}\}_{\bm{c}_Z}}, S_{n-k+1}, \cdots, S_{n}\}. 
\end{align}
Then, we can define a set of orthogonal bases $\{\ket{C_{\bm{w}}}\}_{\bm{w}}$, where 
\begin{align}
    \ketbra{C_{\bm{w}}} = \prod_{j=1}^{n} \frac{\1 + (-1)^{w_j} S_j}{2}, \quad 
    \bm{w} \in \bZ_2^n. 
\end{align}
And there is a subset $W\subseteq \bZ_2^n$, for all $\bm{w}\in W$, $\ket{C_{\bm{w}}}\in \cV_{\rm CSS}$. 
In other word, for a fixed $\bm{w}\in W$, the first $k$ bits of it are all zeros. 
Then, we can define arbitrary state in $\cV^\bot_{\rm CSS}$ as 
\begin{align}
    \ket{\Psi^\bot} = \sum_{\bm{w}\in W^\bot} \alpha_w \ket{C_w  }, \quad 
    \sum_{\bm{w}} |\alpha_{\bm{w}}|^2 = 1, 
\end{align}
where $W^\bot = \bZ_2^{n} \setminus W$, and we have 
\begin{align}
    \bra{\Psi^\bot}\Omega_{\rm CSS}^{\rm XZ}\ket{\Psi^\bot} 
    &= \sum_{\bm{w},\bm{w}'\in W^\bot} \alpha_{\bm{w}}^* \alpha_{\bm{w}'} 
    \bra{C_{\bm{w}}} \Omega_{\rm CSS}^{\rm XZ} \ket{C_{\bm{w}'}} \\
    &= \frac{1}{2} \sum_{\bm{w},\bm{w}'\in W^\bot} \alpha_{\bm{w}}^* \alpha_{\bm{w}'} 
    \left(\bra{C_{\bm{w}}} \prod_{i=1}^{k_X} \frac{\1 + S_i}{2} \ket{C_{\bm{w}'}} 
    + \bra{C_{\bm{w}}} \prod_{i=k_X+1}^{n-k} \frac{\1 + S_i}{2} \ket{C_{\bm{w}'}}\right) \\ 
    &= \frac{1}{2} \sum_{\bm{w}\in W^\bot} |\alpha_{\bm{w}}|^2 
    \left(\epsilon_{\bm{w}}^x + \epsilon_{\bm{w}}^z \right), 
\end{align}
where 
\begin{align}
    \prod_{i=1}^{k_X} \frac{\1 + S_i}{2} \ket{C_{\bm{w}}} 
    &= \epsilon_{\bm{w}}^x \ket{C_{\bm{w}}}, \quad 
    \epsilon_{\bm{w}}^x 
    = \begin{cases}
        1 \quad \text{if $\bm{w}_i = 0$ for all $i=1, \cdots, k_X$}\\
        0 \quad \text{else}
    \end{cases},  \\
    \prod_{i=k_X+1}^{k} \frac{\1 + S_i}{2} \ket{C_{\bm{w}}} 
    &= \epsilon_{\bm{w}}^z \ket{C_{\bm{w}}}, \quad 
    \epsilon_{\bm{w}}^z 
    = \begin{cases}
        1 \quad \text{if $\bm{w}_i = 0$ for all $i=k_X+1, \cdots, n-k$}\\
        0 \quad \text{else}
    \end{cases}.   
\end{align}
Therefore, we have 
\begin{align}
    \bra{\Psi^\bot}\Omega_{\rm CSS}^{\rm XZ}\ket{\Psi^\bot} 
    &= \frac{1}{2} \sum_{\bm{w}\in W^\bot} |\alpha_{\bm{w}}|^2 \left(\epsilon_{\bm{w}}^x + \epsilon_{\bm{w}}^z \right) 
    = \frac{1}{2} \sum_{\bm{w}\in W^\bot} |\alpha_{\bm{w}}|^2
    \leq \frac{1}{2},  
\end{align}
with the fact that for all $\bm{w}\in W^\bot$, $\epsilon_{\bm{w}}^x + \epsilon_{\bm{w}}^z \leq 1$.
Additionally, note that the above equality is achievable. 
Thus, we have 
\begin{align}
    \nu(\Omega_{\rm CSS}^{\rm XZ}) = 1 - \frac{1}{2} = \frac{1}{2}. 
\end{align}

Then, we prove that the uniform distribution is the optimal one as follows. 
We define that 
\begin{align}
    \Omega_{{\rm CSS}, \mu}^{\rm XZ} 
    := \mu(X) \prod_{\bm{c}_X} \frac{\1 + X^{\bm{c}_X}}{2} + 
    \mu(Z)\prod_{\bm{c}_Z} \frac{\1 + Z^{\bm{c}_Z}}{2}. 
\end{align}
We have 
\begin{align}
    \bra{\Psi^\bot}\Omega_{{\rm CSS}, \mu}^{\rm XZ}\ket{\Psi^\bot} 
    &= \sum_{\bm{w}\in W^\bot} |\alpha_{\bm{w}}|^2 \left[\mu(X)\epsilon_{\bm{w}}^x + \mu(Z)\epsilon_{\bm{w}}^z \right] \\
    &\leq \max\{\mu(X), \mu(Z)\}. 
\end{align}
The above equality is achievable. 
With the definition of spectral gap, we have 
\begin{align}
    \nu(\Omega_{{\rm CSS}, \mu}^{\rm XZ}) 
    &= 1 - \max_{\ket{\Psi^\bot}}\; \bra{\Psi^\bot}\Omega_{{\rm CSS}, \mu}^{\rm XZ}\ket{\Psi^\bot} \\
    &= 1 - \max\{\mu(X), \mu(Z)\}. 
\end{align}
Therefore, we obtain the optimal performance of verification strategy when $\mu(X) = \mu(Z) = 1/2$, i.e., sampling uniformly. 

\subsection{Example: toric subspace verification}

A typical CSS code is the \emph{toric code}, which is described by a $L\times L$ torus and each edge represents a qubit~\cite{kitaev2003faulttoleranta}.
The corresponding stabilizer generators can be divided into two groups: 
\begin{enumerate}
    \item those associated with each lattice vertex $v$, with $X$ acting on every qubit associated with an edge attached to the given vertex; 
    \item those associated with each plaquette $p$ of the lattice, with $Z$ acting on each qubit represented by an edge surrounding the plaquette. 
\end{enumerate}
Mathematically, they can be written as 
\begin{align}
    S_v = \prod_{i\in v} X_i \quad \text{and} \quad 
    S_p = \prod_{i\in p} Z_i. 
\end{align}
The corresponding subspace is termed as \emph{toric subspace} $\cV_{\rm Toric}$. 
To verify this subspace, we can directly apply $\Omega_{\rm CSS}^{\rm XZ}$ 
using $2$ measurement settings and consuming $2/\epsilon \ln(1/\delta)$ state copies.
Concretely, the verification operator reads 
\begin{align}
    \Omega_{\rm Toric}^{\rm XZ}
    = \frac{1}{2}\left(\prod_v S_v^+ + \prod_p S_p^+\right),
\end{align}
where $S_v^+$ ($S_p^+$) is the projector onto the positive eigenspace of stabilizer generator $S_v$ ($S_p$).

\subsection{Proof of the dual-containing subspace verification}

Here, we consider the dual-containing code, 
which is a special CSS code with $H_X = H_Z$. 
We label the corresponding dual-containing subspace as $\cV_{\rm DC}$ with the projector
\begin{align}
    \Pi_{\rm DC} := 
    \prod_{\bm{c}} \left(\frac{\1 + X^{\bm{c}}}{2}\right) 
    \left(\frac{\1 + Z^{\bm{c}}}{2}\right), 
\end{align}
where $\bm{c}$ are the rows of $H_X$. 
The verification operator reads 
\begin{align}
    \Omega_{\rm DC}^{\rm XYZ} = \frac{1}{3} 
    \left(\prod_{\bm{c}} \frac{\1 + X^{\bm{c}}}{2} + 
    \prod_{\bm{c}} \frac{\1 + Z^{\bm{c}}}{2} + \prod_{\bm{c}} \frac{\1 + X^{\bm{c}}Z^{\bm{c}}}{2}\right). 
\end{align}
With Lemma~\ref{lemma:product of commute operators}, we have 
\begin{align}
    \tr\left[\prod_{\bm{c}} \frac{\1 + X^{\bm{c}}Z^{\bm{c}}}{2} 
    \prod_{\bm{c}} \left(\frac{\1 + X^{\bm{c}}}{2}\right) 
    \left(\frac{\1 + Z^{\bm{c}}}{2}\right) \right]
    &= \rank(\Pi_{\rm DC}) \\
    \Rightarrow\quad 
    \tr[\Omega_{\rm DC}^{\rm XYZ}\Pi_{\rm DC}] 
    &= \rank(\Pi_{\rm DC}). 
\end{align}
Then, we compute the spectral gap as follows. 
Similarly, we can also define define arbitrary state in $\cV^\bot_{\rm DC}$ as $\ket{\Psi^\bot}$ and have 
\begin{align}
    \bra{\Psi^\bot} \Omega_{\rm DC}^{\rm XYZ} \ket{\Psi^\bot} 
    &= \sum_{\bm{w},\bm{w}'\in W^\bot} \alpha_{\bm{w}}^* \alpha_{\bm{w}'} 
    \bra{C_{\bm{w}}} \Omega_{\rm DC}^{\rm XYZ} \ket{C_{\bm{w}'}} \\
    &= \frac{1}{3} \sum_{\bm{w}\in W^\bot} |\alpha_{\bm{w}}|^2 
    \left(\epsilon_{\bm{w}}^x + \epsilon_{\bm{w}}^z + \epsilon_{\bm{w}}^y\right), 
\end{align}
where 
\begin{align}
    \prod_{i=1}^{k_X} \frac{\1 + S_i S_{i+k_X}}{2} \ket{C_{\bm{w}}} 
    &= \epsilon_{\bm{w}}^y \ket{C_{\bm{w}}}, \quad 
    \epsilon_{\bm{w}}^y 
    = \begin{cases}
        1 \quad \text{if $\bm{w}_i = \bm{w}_{i+k_X}$ for all $i=1, \cdots, k_X$}\\
        0 \quad \text{else}
    \end{cases}. 
\end{align}
Therefore, we have 
\begin{align}
    \bra{\Psi^\bot}\Omega_{\rm DC}^{\rm XYZ}\ket{\Psi^\bot} 
    &= \frac{1}{3} \sum_{\bm{w}\in W^\bot} |\alpha_{\bm{w}}|^2 \left(\epsilon_{\bm{w}}^x + \epsilon_{\bm{w}}^z +\epsilon_{\bm{w}}^y \right) 
    = \frac{1}{3} \sum_{\bm{w}\in W^\bot} |\alpha_{\bm{w}}|^2
    \leq \frac{1}{3},  
\end{align}
with the fact that for all $\bm{w}\in W^\bot$, $\epsilon_{\bm{w}}^x + \epsilon_{\bm{w}}^z +\epsilon_{\bm{w}}^y \leq 1$.
Additionally, note that the above equality is achievable. 
Thus, we have 
\begin{align}
    \nu(\Omega_{\rm DC}^{\rm XYZ}) = 1 - \frac{1}{3} = \frac{2}{3}. 
\end{align}

Then, we can also prove that the uniform distribution is the optimal one as follows. 
We define that 
\begin{align}
    \Omega_{{\rm DC}, \mu}^{\rm XYZ} 
    := \mu(X) \prod_{\bm{c}_X} \frac{\1 + X^{\bm{c}_X}}{2} + 
    \mu(Z)\prod_{\bm{c}_Z} \frac{\1 + Z^{\bm{c}_Z}}{2} +
    \mu(Y) \prod_{\bm{c}} \frac{\1 + X^{\bm{c}}Z^{\bm{c}}}{2}. 
\end{align}
We have 
\begin{align}
    \bra{\Psi^\bot}\Omega_{{\rm DC}, \mu}^{\rm XYZ}\ket{\Psi^\bot} 
    &= \sum_{\bm{w}\in W^\bot} |\alpha_{\bm{w}}|^2 \left[\mu(X)\epsilon_{\bm{w}}^x + \mu(Z)\epsilon_{\bm{w}}^z + \mu(Y)\epsilon_{\bm{w}}^y \right] \\
    &\leq \max\{\mu(X), \mu(Z), \mu(Y)\}. 
\end{align}
The above equality is achievable. 
With the definition of spectral gap, we have 
\begin{align}
    \nu(\Omega_{{\rm DC}, \mu}^{\rm XYZ} ) 
    &= 1 - \max_{\ket{\Psi^\bot}}\; \bra{\Psi^\bot}\Omega_{{\rm DC}, \mu}^{\rm XYZ}\ket{\Psi^\bot} \\
    &= 1 - \max\{\mu(X), \mu(Z), \mu(Y)\}. 
\end{align}
Therefore, we obtain the optimal performance of verification strategy when $\mu(X) = \mu(Z) = \mu(Y) = 1/3$, i.e., sampling uniformly. 


\begin{thebibliography}{52}%
\makeatletter
\providecommand \@ifxundefined [1]{%
 \@ifx{#1\undefined}
}%
\providecommand \@ifnum [1]{%
 \ifnum #1\expandafter \@firstoftwo
 \else \expandafter \@secondoftwo
 \fi
}%
\providecommand \@ifx [1]{%
 \ifx #1\expandafter \@firstoftwo
 \else \expandafter \@secondoftwo
 \fi
}%
\providecommand \natexlab [1]{#1}%
\providecommand \enquote  [1]{``#1''}%
\providecommand \bibnamefont  [1]{#1}%
\providecommand \bibfnamefont [1]{#1}%
\providecommand \citenamefont [1]{#1}%
\providecommand \href@noop [0]{\@secondoftwo}%
\providecommand \href [0]{\begingroup \@sanitize@url \@href}%
\providecommand \@href[1]{\@@startlink{#1}\@@href}%
\providecommand \@@href[1]{\endgroup#1\@@endlink}%
\providecommand \@sanitize@url [0]{\catcode `\\12\catcode `\$12\catcode `\&12\catcode `\#12\catcode `\^12\catcode `\_12\catcode `\%12\relax}%
\providecommand \@@startlink[1]{}%
\providecommand \@@endlink[0]{}%
\providecommand \url  [0]{\begingroup\@sanitize@url \@url }%
\providecommand \@url [1]{\endgroup\@href {#1}{\urlprefix }}%
\providecommand \urlprefix  [0]{URL }%
\providecommand \Eprint [0]{\href }%
\providecommand \doibase [0]{https://doi.org/}%
\providecommand \selectlanguage [0]{\@gobble}%
\providecommand \bibinfo  [0]{\@secondoftwo}%
\providecommand \bibfield  [0]{\@secondoftwo}%
\providecommand \translation [1]{[#1]}%
\providecommand \BibitemOpen [0]{}%
\providecommand \bibitemStop [0]{}%
\providecommand \bibitemNoStop [0]{.\EOS\space}%
\providecommand \EOS [0]{\spacefactor3000\relax}%
\providecommand \BibitemShut  [1]{\csname bibitem#1\endcsname}%
\let\auto@bib@innerbib\@empty
\bibitem [{\citenamefont {Xue}\ \emph {et~al.}(2022{\natexlab{a}})\citenamefont {Xue}, \citenamefont {Wang}, \citenamefont {Zhan}, \citenamefont {Wang}, \citenamefont {Zeng}, \citenamefont {Ding}, \citenamefont {Shi}, \citenamefont {Liu}, \citenamefont {Liu}, \citenamefont {Huang}, \citenamefont {Huang}, \citenamefont {Yu}, \citenamefont {Wang}, \citenamefont {Fu}, \citenamefont {Qiang}, \citenamefont {Xu}, \citenamefont {Deng}, \citenamefont {Yang},\ and\ \citenamefont {Wu}}]{PhysRevLett.129.133601}%
  \BibitemOpen
  \bibfield  {author} {\bibinfo {author} {\bibfnamefont {S.}~\bibnamefont {Xue}}, \bibinfo {author} {\bibfnamefont {Y.}~\bibnamefont {Wang}}, \bibinfo {author} {\bibfnamefont {J.}~\bibnamefont {Zhan}}, \bibinfo {author} {\bibfnamefont {Y.}~\bibnamefont {Wang}}, \bibinfo {author} {\bibfnamefont {R.}~\bibnamefont {Zeng}}, \bibinfo {author} {\bibfnamefont {J.}~\bibnamefont {Ding}}, \bibinfo {author} {\bibfnamefont {W.}~\bibnamefont {Shi}}, \bibinfo {author} {\bibfnamefont {Y.}~\bibnamefont {Liu}}, \bibinfo {author} {\bibfnamefont {Y.}~\bibnamefont {Liu}}, \bibinfo {author} {\bibfnamefont {A.}~\bibnamefont {Huang}}, \bibinfo {author} {\bibfnamefont {G.}~\bibnamefont {Huang}}, \bibinfo {author} {\bibfnamefont {C.}~\bibnamefont {Yu}}, \bibinfo {author} {\bibfnamefont {D.}~\bibnamefont {Wang}}, \bibinfo {author} {\bibfnamefont {X.}~\bibnamefont {Fu}}, \bibinfo {author} {\bibfnamefont {X.}~\bibnamefont {Qiang}}, \bibinfo {author} {\bibfnamefont {P.}~\bibnamefont {Xu}}, \bibinfo {author} {\bibfnamefont
  {M.}~\bibnamefont {Deng}}, \bibinfo {author} {\bibfnamefont {X.}~\bibnamefont {Yang}},\ and\ \bibinfo {author} {\bibfnamefont {J.}~\bibnamefont {Wu}},\ }\href {https://doi.org/10.1103/PhysRevLett.129.133601} {\bibfield  {journal} {\bibinfo  {journal} {Phys. Rev. Lett.}\ }\textbf {\bibinfo {volume} {129}},\ \bibinfo {pages} {133601} (\bibinfo {year} {2022}{\natexlab{a}})}\BibitemShut {NoStop}%
\bibitem [{\citenamefont {Xue}\ \emph {et~al.}(2022{\natexlab{b}})\citenamefont {Xue}, \citenamefont {Liu}, \citenamefont {Wang}, \citenamefont {Zhu}, \citenamefont {Guo},\ and\ \citenamefont {Wu}}]{PhysRevA.105.032427}%
  \BibitemOpen
  \bibfield  {author} {\bibinfo {author} {\bibfnamefont {S.}~\bibnamefont {Xue}}, \bibinfo {author} {\bibfnamefont {Y.}~\bibnamefont {Liu}}, \bibinfo {author} {\bibfnamefont {Y.}~\bibnamefont {Wang}}, \bibinfo {author} {\bibfnamefont {P.}~\bibnamefont {Zhu}}, \bibinfo {author} {\bibfnamefont {C.}~\bibnamefont {Guo}},\ and\ \bibinfo {author} {\bibfnamefont {J.}~\bibnamefont {Wu}},\ }\href {https://doi.org/10.1103/PhysRevA.105.032427} {\bibfield  {journal} {\bibinfo  {journal} {Phys. Rev. A}\ }\textbf {\bibinfo {volume} {105}},\ \bibinfo {pages} {032427} (\bibinfo {year} {2022}{\natexlab{b}})}\BibitemShut {NoStop}%
\bibitem [{\citenamefont {Eisert}\ \emph {et~al.}(2020)\citenamefont {Eisert}, \citenamefont {Hangleiter}, \citenamefont {Walk}, \citenamefont {Roth}, \citenamefont {Markham}, \citenamefont {Parekh}, \citenamefont {Chabaud},\ and\ \citenamefont {Kashefi}}]{eisert2020quantum}%
  \BibitemOpen
  \bibfield  {author} {\bibinfo {author} {\bibfnamefont {J.}~\bibnamefont {Eisert}}, \bibinfo {author} {\bibfnamefont {D.}~\bibnamefont {Hangleiter}}, \bibinfo {author} {\bibfnamefont {N.}~\bibnamefont {Walk}}, \bibinfo {author} {\bibfnamefont {I.}~\bibnamefont {Roth}}, \bibinfo {author} {\bibfnamefont {D.}~\bibnamefont {Markham}}, \bibinfo {author} {\bibfnamefont {R.}~\bibnamefont {Parekh}}, \bibinfo {author} {\bibfnamefont {U.}~\bibnamefont {Chabaud}},\ and\ \bibinfo {author} {\bibfnamefont {E.}~\bibnamefont {Kashefi}},\ }\href {https://doi.org/10.1038/s42254-020-0186-4} {\bibfield  {journal} {\bibinfo  {journal} {Nature Reviews Physics}\ }\textbf {\bibinfo {volume} {2}},\ \bibinfo {pages} {382} (\bibinfo {year} {2020})}\BibitemShut {NoStop}%
\bibitem [{\citenamefont {Kliesch}\ and\ \citenamefont {Roth}(2021)}]{kliesch2021theory}%
  \BibitemOpen
  \bibfield  {author} {\bibinfo {author} {\bibfnamefont {M.}~\bibnamefont {Kliesch}}\ and\ \bibinfo {author} {\bibfnamefont {I.}~\bibnamefont {Roth}},\ }\href {https://doi.org/10.1103/PRXQuantum.2.010201} {\bibfield  {journal} {\bibinfo  {journal} {PRX Quantum}\ }\textbf {\bibinfo {volume} {2}},\ \bibinfo {pages} {010201} (\bibinfo {year} {2021})}\BibitemShut {NoStop}%
\bibitem [{\citenamefont {Flammia}\ and\ \citenamefont {Liu}(2011)}]{flammia2011direct}%
  \BibitemOpen
  \bibfield  {author} {\bibinfo {author} {\bibfnamefont {S.~T.}\ \bibnamefont {Flammia}}\ and\ \bibinfo {author} {\bibfnamefont {Y.-K.}\ \bibnamefont {Liu}},\ }\href {https://doi.org/10.1103/PhysRevLett.106.230501} {\bibfield  {journal} {\bibinfo  {journal} {Physical Review Letters}\ }\textbf {\bibinfo {volume} {106}},\ \bibinfo {pages} {230501} (\bibinfo {year} {2011})}\BibitemShut {NoStop}%
\bibitem [{\citenamefont {Huang}\ \emph {et~al.}(2020)\citenamefont {Huang}, \citenamefont {Kueng},\ and\ \citenamefont {Preskill}}]{huang2020predicting}%
  \BibitemOpen
  \bibfield  {author} {\bibinfo {author} {\bibfnamefont {H.-Y.}\ \bibnamefont {Huang}}, \bibinfo {author} {\bibfnamefont {R.}~\bibnamefont {Kueng}},\ and\ \bibinfo {author} {\bibfnamefont {J.}~\bibnamefont {Preskill}},\ }\href {https://doi.org/10.1038/s41567-020-0932-7} {\bibfield  {journal} {\bibinfo  {journal} {Nature Physics}\ }\textbf {\bibinfo {volume} {16}},\ \bibinfo {pages} {1050} (\bibinfo {year} {2020})}\BibitemShut {NoStop}%
\bibitem [{\citenamefont {Elben}\ \emph {et~al.}(2020{\natexlab{a}})\citenamefont {Elben}, \citenamefont {Vermersch}, \citenamefont {{van Bijnen}}, \citenamefont {Kokail}, \citenamefont {Brydges}, \citenamefont {Maier}, \citenamefont {Joshi}, \citenamefont {Blatt}, \citenamefont {Roos},\ and\ \citenamefont {Zoller}}]{elben2020crossplatforma}%
  \BibitemOpen
  \bibfield  {author} {\bibinfo {author} {\bibfnamefont {A.}~\bibnamefont {Elben}}, \bibinfo {author} {\bibfnamefont {B.}~\bibnamefont {Vermersch}}, \bibinfo {author} {\bibfnamefont {R.}~\bibnamefont {{van Bijnen}}}, \bibinfo {author} {\bibfnamefont {C.}~\bibnamefont {Kokail}}, \bibinfo {author} {\bibfnamefont {T.}~\bibnamefont {Brydges}}, \bibinfo {author} {\bibfnamefont {C.}~\bibnamefont {Maier}}, \bibinfo {author} {\bibfnamefont {M.}~\bibnamefont {Joshi}}, \bibinfo {author} {\bibfnamefont {R.}~\bibnamefont {Blatt}}, \bibinfo {author} {\bibfnamefont {C.~F.}\ \bibnamefont {Roos}},\ and\ \bibinfo {author} {\bibfnamefont {P.}~\bibnamefont {Zoller}},\ }\href {https://doi.org/10.1103/PhysRevLett.124.010504} {\bibfield  {journal} {\bibinfo  {journal} {Physical Review Letters}\ }\textbf {\bibinfo {volume} {124}},\ \bibinfo {pages} {010504} (\bibinfo {year} {2020}{\natexlab{a}})},\ \Eprint {https://arxiv.org/abs/1909.01282} {arxiv:1909.01282 [cond-mat, physics:quant-ph]} \BibitemShut {NoStop}%
\bibitem [{\citenamefont {Zhu}\ \emph {et~al.}(2022)\citenamefont {Zhu}, \citenamefont {Cian}, \citenamefont {Noel}, \citenamefont {Risinger}, \citenamefont {Biswas}, \citenamefont {Egan}, \citenamefont {Zhu}, \citenamefont {Green}, \citenamefont {Alderete}, \citenamefont {Nguyen}, \citenamefont {Wang}, \citenamefont {Maksymov}, \citenamefont {Nam}, \citenamefont {Cetina}, \citenamefont {Linke}, \citenamefont {Hafezi},\ and\ \citenamefont {Monroe}}]{zhu2022crossplatform}%
  \BibitemOpen
  \bibfield  {author} {\bibinfo {author} {\bibfnamefont {D.}~\bibnamefont {Zhu}}, \bibinfo {author} {\bibfnamefont {Z.~P.}\ \bibnamefont {Cian}}, \bibinfo {author} {\bibfnamefont {C.}~\bibnamefont {Noel}}, \bibinfo {author} {\bibfnamefont {A.}~\bibnamefont {Risinger}}, \bibinfo {author} {\bibfnamefont {D.}~\bibnamefont {Biswas}}, \bibinfo {author} {\bibfnamefont {L.}~\bibnamefont {Egan}}, \bibinfo {author} {\bibfnamefont {Y.}~\bibnamefont {Zhu}}, \bibinfo {author} {\bibfnamefont {A.~M.}\ \bibnamefont {Green}}, \bibinfo {author} {\bibfnamefont {C.~H.}\ \bibnamefont {Alderete}}, \bibinfo {author} {\bibfnamefont {N.~H.}\ \bibnamefont {Nguyen}}, \bibinfo {author} {\bibfnamefont {Q.}~\bibnamefont {Wang}}, \bibinfo {author} {\bibfnamefont {A.}~\bibnamefont {Maksymov}}, \bibinfo {author} {\bibfnamefont {Y.}~\bibnamefont {Nam}}, \bibinfo {author} {\bibfnamefont {M.}~\bibnamefont {Cetina}}, \bibinfo {author} {\bibfnamefont {N.~M.}\ \bibnamefont {Linke}}, \bibinfo {author} {\bibfnamefont {M.}~\bibnamefont {Hafezi}},\
  and\ \bibinfo {author} {\bibfnamefont {C.}~\bibnamefont {Monroe}},\ }\href {https://doi.org/10.1038/s41467-022-34279-5} {\bibfield  {journal} {\bibinfo  {journal} {Nature Communications}\ }\textbf {\bibinfo {volume} {13}},\ \bibinfo {pages} {6620} (\bibinfo {year} {2022})}\BibitemShut {NoStop}%
\bibitem [{\citenamefont {Zheng}\ \emph {et~al.}(2024)\citenamefont {Zheng}, \citenamefont {Yu},\ and\ \citenamefont {Wang}}]{zheng2024crossplatform}%
  \BibitemOpen
  \bibfield  {author} {\bibinfo {author} {\bibfnamefont {C.}~\bibnamefont {Zheng}}, \bibinfo {author} {\bibfnamefont {X.}~\bibnamefont {Yu}},\ and\ \bibinfo {author} {\bibfnamefont {K.}~\bibnamefont {Wang}},\ }\href {https://doi.org/10.1038/s41534-023-00797-3} {\bibfield  {journal} {\bibinfo  {journal} {npj Quantum Information}\ }\textbf {\bibinfo {volume} {10}},\ \bibinfo {pages} {1} (\bibinfo {year} {2024})}\BibitemShut {NoStop}%
\bibitem [{\citenamefont {Zhou}\ \emph {et~al.}(2020)\citenamefont {Zhou}, \citenamefont {Zeng},\ and\ \citenamefont {Liu}}]{zhou2020singlecopies}%
  \BibitemOpen
  \bibfield  {author} {\bibinfo {author} {\bibfnamefont {Y.}~\bibnamefont {Zhou}}, \bibinfo {author} {\bibfnamefont {P.}~\bibnamefont {Zeng}},\ and\ \bibinfo {author} {\bibfnamefont {Z.}~\bibnamefont {Liu}},\ }\href {https://doi.org/10.1103/PhysRevLett.125.200502} {\bibfield  {journal} {\bibinfo  {journal} {Physical Review Letters}\ }\textbf {\bibinfo {volume} {125}},\ \bibinfo {pages} {200502} (\bibinfo {year} {2020})}\BibitemShut {NoStop}%
\bibitem [{\citenamefont {Zhu}\ and\ \citenamefont {Hayashi}(2019{\natexlab{a}})}]{zhu2019optimal}%
  \BibitemOpen
  \bibfield  {author} {\bibinfo {author} {\bibfnamefont {H.}~\bibnamefont {Zhu}}\ and\ \bibinfo {author} {\bibfnamefont {M.}~\bibnamefont {Hayashi}},\ }\href {https://doi.org/10.1103/PhysRevA.99.052346} {\bibfield  {journal} {\bibinfo  {journal} {Physical Review A}\ }\textbf {\bibinfo {volume} {99}},\ \bibinfo {pages} {052346} (\bibinfo {year} {2019}{\natexlab{a}})}\BibitemShut {NoStop}%
\bibitem [{\citenamefont {Brand{\~a}o}\ and\ \citenamefont {Vianna}(2004)}]{brandao2004separable}%
  \BibitemOpen
  \bibfield  {author} {\bibinfo {author} {\bibfnamefont {F.~G. S.~L.}\ \bibnamefont {Brand{\~a}o}}\ and\ \bibinfo {author} {\bibfnamefont {R.~O.}\ \bibnamefont {Vianna}},\ }\href {https://doi.org/10.1103/PhysRevLett.93.220503} {\bibfield  {journal} {\bibinfo  {journal} {Physical Review Letters}\ }\textbf {\bibinfo {volume} {93}},\ \bibinfo {pages} {220503} (\bibinfo {year} {2004})}\BibitemShut {NoStop}%
\bibitem [{\citenamefont {Elben}\ \emph {et~al.}(2020{\natexlab{b}})\citenamefont {Elben}, \citenamefont {Kueng}, \citenamefont {Huang}, \citenamefont {{van Bijnen}}, \citenamefont {Kokail}, \citenamefont {Dalmonte}, \citenamefont {Calabrese}, \citenamefont {Kraus}, \citenamefont {Preskill}, \citenamefont {Zoller},\ and\ \citenamefont {Vermersch}}]{elben2020mixedstate}%
  \BibitemOpen
  \bibfield  {author} {\bibinfo {author} {\bibfnamefont {A.}~\bibnamefont {Elben}}, \bibinfo {author} {\bibfnamefont {R.}~\bibnamefont {Kueng}}, \bibinfo {author} {\bibfnamefont {H.-Y.~R.}\ \bibnamefont {Huang}}, \bibinfo {author} {\bibfnamefont {R.}~\bibnamefont {{van Bijnen}}}, \bibinfo {author} {\bibfnamefont {C.}~\bibnamefont {Kokail}}, \bibinfo {author} {\bibfnamefont {M.}~\bibnamefont {Dalmonte}}, \bibinfo {author} {\bibfnamefont {P.}~\bibnamefont {Calabrese}}, \bibinfo {author} {\bibfnamefont {B.}~\bibnamefont {Kraus}}, \bibinfo {author} {\bibfnamefont {J.}~\bibnamefont {Preskill}}, \bibinfo {author} {\bibfnamefont {P.}~\bibnamefont {Zoller}},\ and\ \bibinfo {author} {\bibfnamefont {B.}~\bibnamefont {Vermersch}},\ }\href {https://doi.org/10.1103/PhysRevLett.125.200501} {\bibfield  {journal} {\bibinfo  {journal} {Physical Review Letters}\ }\textbf {\bibinfo {volume} {125}},\ \bibinfo {pages} {200501} (\bibinfo {year} {2020}{\natexlab{b}})}\BibitemShut {NoStop}%
\bibitem [{\citenamefont {Pallister}\ \emph {et~al.}(2018)\citenamefont {Pallister}, \citenamefont {Linden},\ and\ \citenamefont {Montanaro}}]{pallister2018optimala}%
  \BibitemOpen
  \bibfield  {author} {\bibinfo {author} {\bibfnamefont {S.}~\bibnamefont {Pallister}}, \bibinfo {author} {\bibfnamefont {N.}~\bibnamefont {Linden}},\ and\ \bibinfo {author} {\bibfnamefont {A.}~\bibnamefont {Montanaro}},\ }\href {https://doi.org/10.1103/PhysRevLett.120.170502} {\bibfield  {journal} {\bibinfo  {journal} {Physical Review Letters}\ }\textbf {\bibinfo {volume} {120}},\ \bibinfo {pages} {170502} (\bibinfo {year} {2018})}\BibitemShut {NoStop}%
\bibitem [{\citenamefont {Wang}\ and\ \citenamefont {Hayashi}(2019)}]{wang2019optimala}%
  \BibitemOpen
  \bibfield  {author} {\bibinfo {author} {\bibfnamefont {K.}~\bibnamefont {Wang}}\ and\ \bibinfo {author} {\bibfnamefont {M.}~\bibnamefont {Hayashi}},\ }\href {https://doi.org/10.1103/PhysRevA.100.032315} {\bibfield  {journal} {\bibinfo  {journal} {Physical Review A}\ }\textbf {\bibinfo {volume} {100}},\ \bibinfo {pages} {032315} (\bibinfo {year} {2019})}\BibitemShut {NoStop}%
\bibitem [{\citenamefont {Yu}\ \emph {et~al.}(2022)\citenamefont {Yu}, \citenamefont {Shang},\ and\ \citenamefont {G{\"u}hne}}]{yu2022statisticala}%
  \BibitemOpen
  \bibfield  {author} {\bibinfo {author} {\bibfnamefont {X.-D.}\ \bibnamefont {Yu}}, \bibinfo {author} {\bibfnamefont {J.}~\bibnamefont {Shang}},\ and\ \bibinfo {author} {\bibfnamefont {O.}~\bibnamefont {G{\"u}hne}},\ }\href {https://doi.org/10.1002/qute.202100126} {\bibfield  {journal} {\bibinfo  {journal} {Advanced Quantum Technologies}\ }\textbf {\bibinfo {volume} {5}},\ \bibinfo {pages} {2100126} (\bibinfo {year} {2022})}\BibitemShut {NoStop}%
\bibitem [{\citenamefont {Li}\ \emph {et~al.}(2020)\citenamefont {Li}, \citenamefont {Han},\ and\ \citenamefont {Zhu}}]{li2020optimal}%
  \BibitemOpen
  \bibfield  {author} {\bibinfo {author} {\bibfnamefont {Z.}~\bibnamefont {Li}}, \bibinfo {author} {\bibfnamefont {Y.-G.}\ \bibnamefont {Han}},\ and\ \bibinfo {author} {\bibfnamefont {H.}~\bibnamefont {Zhu}},\ }\href {https://doi.org/10.1103/PhysRevApplied.13.054002} {\bibfield  {journal} {\bibinfo  {journal} {Physical Review Applied}\ }\textbf {\bibinfo {volume} {13}},\ \bibinfo {pages} {054002} (\bibinfo {year} {2020})}\BibitemShut {NoStop}%
\bibitem [{\citenamefont {Dangniam}\ \emph {et~al.}(2020)\citenamefont {Dangniam}, \citenamefont {Han},\ and\ \citenamefont {Zhu}}]{dangniam2020optimal}%
  \BibitemOpen
  \bibfield  {author} {\bibinfo {author} {\bibfnamefont {N.}~\bibnamefont {Dangniam}}, \bibinfo {author} {\bibfnamefont {Y.-G.}\ \bibnamefont {Han}},\ and\ \bibinfo {author} {\bibfnamefont {H.}~\bibnamefont {Zhu}},\ }\href {https://doi.org/10.1103/PhysRevResearch.2.043323} {\bibfield  {journal} {\bibinfo  {journal} {Physical Review Research}\ }\textbf {\bibinfo {volume} {2}},\ \bibinfo {pages} {043323} (\bibinfo {year} {2020})}\BibitemShut {NoStop}%
\bibitem [{\citenamefont {Chen}\ \emph {et~al.}(2023{\natexlab{a}})\citenamefont {Chen}, \citenamefont {Xie},\ and\ \citenamefont {Wang}}]{chen2023memory}%
  \BibitemOpen
  \bibfield  {author} {\bibinfo {author} {\bibfnamefont {S.}~\bibnamefont {Chen}}, \bibinfo {author} {\bibfnamefont {W.}~\bibnamefont {Xie}},\ and\ \bibinfo {author} {\bibfnamefont {K.}~\bibnamefont {Wang}},\ }\href@noop {} {\bibfield  {journal} {\bibinfo  {journal} {arXiv:2312.11066}\ } (\bibinfo {year} {2023}{\natexlab{a}})},\ \Eprint {https://arxiv.org/abs/2312.11066} {arXiv:2312.11066 [quant-ph]} \BibitemShut {NoStop}%
\bibitem [{\citenamefont {Li}\ \emph {et~al.}(2021)\citenamefont {Li}, \citenamefont {Han}, \citenamefont {Sun}, \citenamefont {Shang},\ and\ \citenamefont {Zhu}}]{li2021verification}%
  \BibitemOpen
  \bibfield  {author} {\bibinfo {author} {\bibfnamefont {Z.}~\bibnamefont {Li}}, \bibinfo {author} {\bibfnamefont {Y.-G.}\ \bibnamefont {Han}}, \bibinfo {author} {\bibfnamefont {H.-F.}\ \bibnamefont {Sun}}, \bibinfo {author} {\bibfnamefont {J.}~\bibnamefont {Shang}},\ and\ \bibinfo {author} {\bibfnamefont {H.}~\bibnamefont {Zhu}},\ }\href {https://doi.org/10.1103/PhysRevA.103.022601} {\bibfield  {journal} {\bibinfo  {journal} {Physical Review A}\ }\textbf {\bibinfo {volume} {103}},\ \bibinfo {pages} {022601} (\bibinfo {year} {2021})}\BibitemShut {NoStop}%
\bibitem [{\citenamefont {Jiang}\ \emph {et~al.}(2020)\citenamefont {Jiang}, \citenamefont {Wang}, \citenamefont {Qian}, \citenamefont {Chen}, \citenamefont {Chen}, \citenamefont {Lu}, \citenamefont {Xia}, \citenamefont {Song}, \citenamefont {Zhu},\ and\ \citenamefont {Ma}}]{jiang2020towards}%
  \BibitemOpen
  \bibfield  {author} {\bibinfo {author} {\bibfnamefont {X.}~\bibnamefont {Jiang}}, \bibinfo {author} {\bibfnamefont {K.}~\bibnamefont {Wang}}, \bibinfo {author} {\bibfnamefont {K.}~\bibnamefont {Qian}}, \bibinfo {author} {\bibfnamefont {Z.}~\bibnamefont {Chen}}, \bibinfo {author} {\bibfnamefont {Z.}~\bibnamefont {Chen}}, \bibinfo {author} {\bibfnamefont {L.}~\bibnamefont {Lu}}, \bibinfo {author} {\bibfnamefont {L.}~\bibnamefont {Xia}}, \bibinfo {author} {\bibfnamefont {F.}~\bibnamefont {Song}}, \bibinfo {author} {\bibfnamefont {S.}~\bibnamefont {Zhu}},\ and\ \bibinfo {author} {\bibfnamefont {X.}~\bibnamefont {Ma}},\ }\href {https://doi.org/10.1038/s41534-020-00317-7} {\bibfield  {journal} {\bibinfo  {journal} {npj Quantum Information}\ }\textbf {\bibinfo {volume} {6}},\ \bibinfo {pages} {90} (\bibinfo {year} {2020})}\BibitemShut {NoStop}%
\bibitem [{\citenamefont {Zhang}\ \emph {et~al.}(2020)\citenamefont {Zhang}, \citenamefont {Zhang}, \citenamefont {Chen}, \citenamefont {Peng}, \citenamefont {Xu}, \citenamefont {Yin}, \citenamefont {Yu}, \citenamefont {Ye}, \citenamefont {Han}, \citenamefont {Xu}, \citenamefont {Chen}, \citenamefont {Li},\ and\ \citenamefont {Guo}}]{Zhang_2020}%
  \BibitemOpen
  \bibfield  {author} {\bibinfo {author} {\bibfnamefont {W.-H.}\ \bibnamefont {Zhang}}, \bibinfo {author} {\bibfnamefont {C.}~\bibnamefont {Zhang}}, \bibinfo {author} {\bibfnamefont {Z.}~\bibnamefont {Chen}}, \bibinfo {author} {\bibfnamefont {X.-X.}\ \bibnamefont {Peng}}, \bibinfo {author} {\bibfnamefont {X.-Y.}\ \bibnamefont {Xu}}, \bibinfo {author} {\bibfnamefont {P.}~\bibnamefont {Yin}}, \bibinfo {author} {\bibfnamefont {S.}~\bibnamefont {Yu}}, \bibinfo {author} {\bibfnamefont {X.-J.}\ \bibnamefont {Ye}}, \bibinfo {author} {\bibfnamefont {Y.-J.}\ \bibnamefont {Han}}, \bibinfo {author} {\bibfnamefont {J.-S.}\ \bibnamefont {Xu}}, \bibinfo {author} {\bibfnamefont {G.}~\bibnamefont {Chen}}, \bibinfo {author} {\bibfnamefont {C.-F.}\ \bibnamefont {Li}},\ and\ \bibinfo {author} {\bibfnamefont {G.-C.}\ \bibnamefont {Guo}},\ }\href {https://doi.org/10.1103/PhysRevLett.125.030506} {\bibfield  {journal} {\bibinfo  {journal} {Physical Review Letters}\ }\textbf {\bibinfo {volume} {125}},\ \bibinfo {pages} {030506}
  (\bibinfo {year} {2020})}\BibitemShut {NoStop}%
\bibitem [{\citenamefont {Xia}\ \emph {et~al.}(2022)\citenamefont {Xia}, \citenamefont {Lu}, \citenamefont {Wang}, \citenamefont {Jiang}, \citenamefont {Zhu},\ and\ \citenamefont {Ma}}]{Xia_2022}%
  \BibitemOpen
  \bibfield  {author} {\bibinfo {author} {\bibfnamefont {L.}~\bibnamefont {Xia}}, \bibinfo {author} {\bibfnamefont {L.}~\bibnamefont {Lu}}, \bibinfo {author} {\bibfnamefont {K.}~\bibnamefont {Wang}}, \bibinfo {author} {\bibfnamefont {X.}~\bibnamefont {Jiang}}, \bibinfo {author} {\bibfnamefont {S.}~\bibnamefont {Zhu}},\ and\ \bibinfo {author} {\bibfnamefont {X.}~\bibnamefont {Ma}},\ }\href {https://doi.org/10.1088/1367-2630/ac8a67} {\bibfield  {journal} {\bibinfo  {journal} {New Journal of Physics}\ }\textbf {\bibinfo {volume} {24}},\ \bibinfo {pages} {095002} (\bibinfo {year} {2022})}\BibitemShut {NoStop}%
\bibitem [{\citenamefont {Gottesman}(1997)}]{gottesman1997stabilizer}%
  \BibitemOpen
  \bibfield  {author} {\bibinfo {author} {\bibfnamefont {D.}~\bibnamefont {Gottesman}},\ }\href {https://doi.org/10.48550/arXiv.quant-ph/9705052} {\bibinfo {title} {Stabilizer codes and quantum error correction}} (\bibinfo {year} {1997}),\ \Eprint {https://arxiv.org/abs/quant-ph/9705052} {arXiv:quant-ph/9705052} \BibitemShut {NoStop}%
\bibitem [{\citenamefont {Shor}(1995)}]{shor1995scheme}%
  \BibitemOpen
  \bibfield  {author} {\bibinfo {author} {\bibfnamefont {P.~W.}\ \bibnamefont {Shor}},\ }\href {https://doi.org/10.1103/PhysRevA.52.R2493} {\bibfield  {journal} {\bibinfo  {journal} {Physical Review A}\ }\textbf {\bibinfo {volume} {52}},\ \bibinfo {pages} {R2493} (\bibinfo {year} {1995})}\BibitemShut {NoStop}%
\bibitem [{\citenamefont {Steane}(1996)}]{1996multipleparticle}%
  \BibitemOpen
  \bibfield  {author} {\bibinfo {author} {\bibfnamefont {A.}~\bibnamefont {Steane}},\ }\bibfield  {journal} {\bibinfo  {journal} {Proceedings of the Royal Society of London. Series A: Mathematical, Physical and Engineering Sciences}\ }\href {https://doi.org/10.1098/rspa.1996.0136} {10.1098/rspa.1996.0136} (\bibinfo {year} {1996})\BibitemShut {NoStop}%
\bibitem [{\citenamefont {Calderbank}\ and\ \citenamefont {Shor}(1996)}]{calderbank1996good}%
  \BibitemOpen
  \bibfield  {author} {\bibinfo {author} {\bibfnamefont {A.~R.}\ \bibnamefont {Calderbank}}\ and\ \bibinfo {author} {\bibfnamefont {P.~W.}\ \bibnamefont {Shor}},\ }\href {https://doi.org/10.1103/PhysRevA.54.1098} {\bibfield  {journal} {\bibinfo  {journal} {Physical Review A}\ }\textbf {\bibinfo {volume} {54}},\ \bibinfo {pages} {1098} (\bibinfo {year} {1996})}\BibitemShut {NoStop}%
\bibitem [{\citenamefont {Kitaev}(2003)}]{kitaev2003faulttoleranta}%
  \BibitemOpen
  \bibfield  {author} {\bibinfo {author} {\bibfnamefont {A.~{\relax Yu}.}\ \bibnamefont {Kitaev}},\ }\href {https://doi.org/10.1016/S0003-4916(02)00018-0} {\bibfield  {journal} {\bibinfo  {journal} {Annals of Physics}\ }\textbf {\bibinfo {volume} {303}},\ \bibinfo {pages} {2} (\bibinfo {year} {2003})}\BibitemShut {NoStop}%
\bibitem [{\citenamefont {{Google Quantum AI}}(2021)}]{chen2021exponential}%
  \BibitemOpen
  \bibfield  {author} {\bibinfo {author} {\bibnamefont {{Google Quantum AI}}},\ }\href {https://doi.org/10.1038/s41586-021-03588-y} {\bibfield  {journal} {\bibinfo  {journal} {Nature}\ }\textbf {\bibinfo {volume} {595}},\ \bibinfo {pages} {383} (\bibinfo {year} {2021})}\BibitemShut {NoStop}%
\bibitem [{\citenamefont {Krinner}\ \emph {et~al.}(2022)\citenamefont {Krinner}, \citenamefont {Lacroix}, \citenamefont {Remm}, \citenamefont {Di~Paolo}, \citenamefont {Genois}, \citenamefont {Leroux}, \citenamefont {Hellings}, \citenamefont {Lazar}, \citenamefont {Swiadek}, \citenamefont {Herrmann}, \citenamefont {Norris}, \citenamefont {Andersen}, \citenamefont {M{\"u}ller}, \citenamefont {Blais}, \citenamefont {Eichler},\ and\ \citenamefont {Wallraff}}]{krinner2022realizing}%
  \BibitemOpen
  \bibfield  {author} {\bibinfo {author} {\bibfnamefont {S.}~\bibnamefont {Krinner}}, \bibinfo {author} {\bibfnamefont {N.}~\bibnamefont {Lacroix}}, \bibinfo {author} {\bibfnamefont {A.}~\bibnamefont {Remm}}, \bibinfo {author} {\bibfnamefont {A.}~\bibnamefont {Di~Paolo}}, \bibinfo {author} {\bibfnamefont {E.}~\bibnamefont {Genois}}, \bibinfo {author} {\bibfnamefont {C.}~\bibnamefont {Leroux}}, \bibinfo {author} {\bibfnamefont {C.}~\bibnamefont {Hellings}}, \bibinfo {author} {\bibfnamefont {S.}~\bibnamefont {Lazar}}, \bibinfo {author} {\bibfnamefont {F.}~\bibnamefont {Swiadek}}, \bibinfo {author} {\bibfnamefont {J.}~\bibnamefont {Herrmann}}, \bibinfo {author} {\bibfnamefont {G.~J.}\ \bibnamefont {Norris}}, \bibinfo {author} {\bibfnamefont {C.~K.}\ \bibnamefont {Andersen}}, \bibinfo {author} {\bibfnamefont {M.}~\bibnamefont {M{\"u}ller}}, \bibinfo {author} {\bibfnamefont {A.}~\bibnamefont {Blais}}, \bibinfo {author} {\bibfnamefont {C.}~\bibnamefont {Eichler}},\ and\ \bibinfo {author} {\bibfnamefont
  {A.}~\bibnamefont {Wallraff}},\ }\href {https://doi.org/10.1038/s41586-022-04566-8} {\bibfield  {journal} {\bibinfo  {journal} {Nature}\ }\textbf {\bibinfo {volume} {605}},\ \bibinfo {pages} {669} (\bibinfo {year} {2022})}\BibitemShut {NoStop}%
\bibitem [{\citenamefont {Takeda}\ \emph {et~al.}(2022)\citenamefont {Takeda}, \citenamefont {Noiri}, \citenamefont {Nakajima}, \citenamefont {Kobayashi},\ and\ \citenamefont {Tarucha}}]{takeda2022quantum}%
  \BibitemOpen
  \bibfield  {author} {\bibinfo {author} {\bibfnamefont {K.}~\bibnamefont {Takeda}}, \bibinfo {author} {\bibfnamefont {A.}~\bibnamefont {Noiri}}, \bibinfo {author} {\bibfnamefont {T.}~\bibnamefont {Nakajima}}, \bibinfo {author} {\bibfnamefont {T.}~\bibnamefont {Kobayashi}},\ and\ \bibinfo {author} {\bibfnamefont {S.}~\bibnamefont {Tarucha}},\ }\href {https://doi.org/10.1038/s41586-022-04986-6} {\bibfield  {journal} {\bibinfo  {journal} {Nature}\ }\textbf {\bibinfo {volume} {608}},\ \bibinfo {pages} {682} (\bibinfo {year} {2022})}\BibitemShut {NoStop}%
\bibitem [{\citenamefont {{Google Quantum AI}}(2023)}]{acharya2023suppressing}%
  \BibitemOpen
  \bibfield  {author} {\bibinfo {author} {\bibnamefont {{Google Quantum AI}}},\ }\href {https://doi.org/10.1038/s41586-022-05434-1} {\bibfield  {journal} {\bibinfo  {journal} {Nature}\ }\textbf {\bibinfo {volume} {614}},\ \bibinfo {pages} {676} (\bibinfo {year} {2023})}\BibitemShut {NoStop}%
\bibitem [{\citenamefont {Ni}\ \emph {et~al.}(2023)\citenamefont {Ni}, \citenamefont {Li}, \citenamefont {Deng}, \citenamefont {Cai}, \citenamefont {Zhang}, \citenamefont {Wang}, \citenamefont {Yang}, \citenamefont {Yu}, \citenamefont {Yan}, \citenamefont {Liu}, \citenamefont {Zou}, \citenamefont {Sun}, \citenamefont {Zheng}, \citenamefont {Xu},\ and\ \citenamefont {Yu}}]{ni2023beating}%
  \BibitemOpen
  \bibfield  {author} {\bibinfo {author} {\bibfnamefont {Z.}~\bibnamefont {Ni}}, \bibinfo {author} {\bibfnamefont {S.}~\bibnamefont {Li}}, \bibinfo {author} {\bibfnamefont {X.}~\bibnamefont {Deng}}, \bibinfo {author} {\bibfnamefont {Y.}~\bibnamefont {Cai}}, \bibinfo {author} {\bibfnamefont {L.}~\bibnamefont {Zhang}}, \bibinfo {author} {\bibfnamefont {W.}~\bibnamefont {Wang}}, \bibinfo {author} {\bibfnamefont {Z.-B.}\ \bibnamefont {Yang}}, \bibinfo {author} {\bibfnamefont {H.}~\bibnamefont {Yu}}, \bibinfo {author} {\bibfnamefont {F.}~\bibnamefont {Yan}}, \bibinfo {author} {\bibfnamefont {S.}~\bibnamefont {Liu}}, \bibinfo {author} {\bibfnamefont {C.-L.}\ \bibnamefont {Zou}}, \bibinfo {author} {\bibfnamefont {L.}~\bibnamefont {Sun}}, \bibinfo {author} {\bibfnamefont {S.-B.}\ \bibnamefont {Zheng}}, \bibinfo {author} {\bibfnamefont {Y.}~\bibnamefont {Xu}},\ and\ \bibinfo {author} {\bibfnamefont {D.}~\bibnamefont {Yu}},\ }\href {https://doi.org/10.1038/s41586-023-05784-4} {\bibfield  {journal} {\bibinfo  {journal}
  {Nature}\ }\textbf {\bibinfo {volume} {616}},\ \bibinfo {pages} {56} (\bibinfo {year} {2023})}\BibitemShut {NoStop}%
\bibitem [{\citenamefont {Bluvstein}\ \emph {et~al.}(2024)\citenamefont {Bluvstein}, \citenamefont {Evered}, \citenamefont {Geim}, \citenamefont {Li}, \citenamefont {Zhou}, \citenamefont {Manovitz}, \citenamefont {Ebadi}, \citenamefont {Cain}, \citenamefont {Kalinowski}, \citenamefont {Hangleiter}, \citenamefont {Bonilla~Ataides}, \citenamefont {Maskara}, \citenamefont {Cong}, \citenamefont {Gao}, \citenamefont {Sales~Rodriguez}, \citenamefont {Karolyshyn}, \citenamefont {Semeghini}, \citenamefont {Gullans}, \citenamefont {Greiner}, \citenamefont {Vuleti{\'c}},\ and\ \citenamefont {Lukin}}]{bluvstein2024logical}%
  \BibitemOpen
  \bibfield  {author} {\bibinfo {author} {\bibfnamefont {D.}~\bibnamefont {Bluvstein}}, \bibinfo {author} {\bibfnamefont {S.~J.}\ \bibnamefont {Evered}}, \bibinfo {author} {\bibfnamefont {A.~A.}\ \bibnamefont {Geim}}, \bibinfo {author} {\bibfnamefont {S.~H.}\ \bibnamefont {Li}}, \bibinfo {author} {\bibfnamefont {H.}~\bibnamefont {Zhou}}, \bibinfo {author} {\bibfnamefont {T.}~\bibnamefont {Manovitz}}, \bibinfo {author} {\bibfnamefont {S.}~\bibnamefont {Ebadi}}, \bibinfo {author} {\bibfnamefont {M.}~\bibnamefont {Cain}}, \bibinfo {author} {\bibfnamefont {M.}~\bibnamefont {Kalinowski}}, \bibinfo {author} {\bibfnamefont {D.}~\bibnamefont {Hangleiter}}, \bibinfo {author} {\bibfnamefont {J.~P.}\ \bibnamefont {Bonilla~Ataides}}, \bibinfo {author} {\bibfnamefont {N.}~\bibnamefont {Maskara}}, \bibinfo {author} {\bibfnamefont {I.}~\bibnamefont {Cong}}, \bibinfo {author} {\bibfnamefont {X.}~\bibnamefont {Gao}}, \bibinfo {author} {\bibfnamefont {P.}~\bibnamefont {Sales~Rodriguez}}, \bibinfo {author} {\bibfnamefont
  {T.}~\bibnamefont {Karolyshyn}}, \bibinfo {author} {\bibfnamefont {G.}~\bibnamefont {Semeghini}}, \bibinfo {author} {\bibfnamefont {M.~J.}\ \bibnamefont {Gullans}}, \bibinfo {author} {\bibfnamefont {M.}~\bibnamefont {Greiner}}, \bibinfo {author} {\bibfnamefont {V.}~\bibnamefont {Vuleti{\'c}}},\ and\ \bibinfo {author} {\bibfnamefont {M.~D.}\ \bibnamefont {Lukin}},\ }\href {https://doi.org/10.1038/s41586-023-06927-3} {\bibfield  {journal} {\bibinfo  {journal} {Nature}\ }\textbf {\bibinfo {volume} {626}},\ \bibinfo {pages} {58} (\bibinfo {year} {2024})}\BibitemShut {NoStop}%
\bibitem [{\citenamefont {{Ryan-Anderson}}\ \emph {et~al.}(2021)\citenamefont {{Ryan-Anderson}}, \citenamefont {Bohnet}, \citenamefont {Lee}, \citenamefont {Gresh}, \citenamefont {Hankin}, \citenamefont {Gaebler}, \citenamefont {Francois}, \citenamefont {Chernoguzov}, \citenamefont {Lucchetti}, \citenamefont {Brown}, \citenamefont {Gatterman}, \citenamefont {Halit}, \citenamefont {Gilmore}, \citenamefont {Gerber}, \citenamefont {Neyenhuis}, \citenamefont {Hayes},\ and\ \citenamefont {Stutz}}]{ryan-anderson2021realizationa}%
  \BibitemOpen
  \bibfield  {author} {\bibinfo {author} {\bibfnamefont {C.}~\bibnamefont {{Ryan-Anderson}}}, \bibinfo {author} {\bibfnamefont {J.~G.}\ \bibnamefont {Bohnet}}, \bibinfo {author} {\bibfnamefont {K.}~\bibnamefont {Lee}}, \bibinfo {author} {\bibfnamefont {D.}~\bibnamefont {Gresh}}, \bibinfo {author} {\bibfnamefont {A.}~\bibnamefont {Hankin}}, \bibinfo {author} {\bibfnamefont {J.~P.}\ \bibnamefont {Gaebler}}, \bibinfo {author} {\bibfnamefont {D.}~\bibnamefont {Francois}}, \bibinfo {author} {\bibfnamefont {A.}~\bibnamefont {Chernoguzov}}, \bibinfo {author} {\bibfnamefont {D.}~\bibnamefont {Lucchetti}}, \bibinfo {author} {\bibfnamefont {N.~C.}\ \bibnamefont {Brown}}, \bibinfo {author} {\bibfnamefont {T.~M.}\ \bibnamefont {Gatterman}}, \bibinfo {author} {\bibfnamefont {S.~K.}\ \bibnamefont {Halit}}, \bibinfo {author} {\bibfnamefont {K.}~\bibnamefont {Gilmore}}, \bibinfo {author} {\bibfnamefont {J.~A.}\ \bibnamefont {Gerber}}, \bibinfo {author} {\bibfnamefont {B.}~\bibnamefont {Neyenhuis}}, \bibinfo {author}
  {\bibfnamefont {D.}~\bibnamefont {Hayes}},\ and\ \bibinfo {author} {\bibfnamefont {R.~P.}\ \bibnamefont {Stutz}},\ }\href {https://doi.org/10.1103/PhysRevX.11.041058} {\bibfield  {journal} {\bibinfo  {journal} {Physical Review X}\ }\textbf {\bibinfo {volume} {11}},\ \bibinfo {pages} {041058} (\bibinfo {year} {2021})}\BibitemShut {NoStop}%
\bibitem [{\citenamefont {{Rodriguez-Blanco}}\ \emph {et~al.}(2021)\citenamefont {{Rodriguez-Blanco}}, \citenamefont {Bermudez}, \citenamefont {M{\"u}ller},\ and\ \citenamefont {Shahandeh}}]{rodriguez-blanco2021efficient}%
  \BibitemOpen
  \bibfield  {author} {\bibinfo {author} {\bibfnamefont {A.}~\bibnamefont {{Rodriguez-Blanco}}}, \bibinfo {author} {\bibfnamefont {A.}~\bibnamefont {Bermudez}}, \bibinfo {author} {\bibfnamefont {M.}~\bibnamefont {M{\"u}ller}},\ and\ \bibinfo {author} {\bibfnamefont {F.}~\bibnamefont {Shahandeh}},\ }\href {https://doi.org/10.1103/PRXQuantum.2.020304} {\bibfield  {journal} {\bibinfo  {journal} {PRX Quantum}\ }\textbf {\bibinfo {volume} {2}},\ \bibinfo {pages} {020304} (\bibinfo {year} {2021})}\BibitemShut {NoStop}%
\bibitem [{\citenamefont {Miao}\ \emph {et~al.}(2023)\citenamefont {Miao}, \citenamefont {McEwen}, \citenamefont {Atalaya}, \citenamefont {Kafri}, \citenamefont {Pryadko}, \citenamefont {Bengtsson}, \citenamefont {Opremcak}, \citenamefont {Satzinger}, \citenamefont {Chen}, \citenamefont {Klimov} \emph {et~al.}}]{miao2023overcoming}%
  \BibitemOpen
  \bibfield  {author} {\bibinfo {author} {\bibfnamefont {K.~C.}\ \bibnamefont {Miao}}, \bibinfo {author} {\bibfnamefont {M.}~\bibnamefont {McEwen}}, \bibinfo {author} {\bibfnamefont {J.}~\bibnamefont {Atalaya}}, \bibinfo {author} {\bibfnamefont {D.}~\bibnamefont {Kafri}}, \bibinfo {author} {\bibfnamefont {L.~P.}\ \bibnamefont {Pryadko}}, \bibinfo {author} {\bibfnamefont {A.}~\bibnamefont {Bengtsson}}, \bibinfo {author} {\bibfnamefont {A.}~\bibnamefont {Opremcak}}, \bibinfo {author} {\bibfnamefont {K.~J.}\ \bibnamefont {Satzinger}}, \bibinfo {author} {\bibfnamefont {Z.}~\bibnamefont {Chen}}, \bibinfo {author} {\bibfnamefont {P.~V.}\ \bibnamefont {Klimov}}, \emph {et~al.},\ }\href {https://doi.org/10.1038/s41567-023-02226-w} {\bibfield  {journal} {\bibinfo  {journal} {Nature Physics}\ }\textbf {\bibinfo {volume} {19}},\ \bibinfo {pages} {1780} (\bibinfo {year} {2023})}\BibitemShut {NoStop}%
\bibitem [{\citenamefont {Zhang}\ \emph {et~al.}(2012)\citenamefont {Zhang}, \citenamefont {Laflamme},\ and\ \citenamefont {Suter}}]{zhang2012experimental}%
  \BibitemOpen
  \bibfield  {author} {\bibinfo {author} {\bibfnamefont {J.}~\bibnamefont {Zhang}}, \bibinfo {author} {\bibfnamefont {R.}~\bibnamefont {Laflamme}},\ and\ \bibinfo {author} {\bibfnamefont {D.}~\bibnamefont {Suter}},\ }\href {https://doi.org/10.1103/PhysRevLett.109.100503} {\bibfield  {journal} {\bibinfo  {journal} {Physical Review Letters}\ }\textbf {\bibinfo {volume} {109}},\ \bibinfo {pages} {100503} (\bibinfo {year} {2012})},\ \Eprint {https://arxiv.org/abs/1208.4797} {arXiv:1208.4797 [quant-ph]} \BibitemShut {NoStop}%
\bibitem [{\citenamefont {Baccari}\ \emph {et~al.}(2020)\citenamefont {Baccari}, \citenamefont {Augusiak}, \citenamefont {{\v S}upi{\'c}},\ and\ \citenamefont {Ac{\'i}n}}]{baccari2020deviceindependent}%
  \BibitemOpen
  \bibfield  {author} {\bibinfo {author} {\bibfnamefont {F.}~\bibnamefont {Baccari}}, \bibinfo {author} {\bibfnamefont {R.}~\bibnamefont {Augusiak}}, \bibinfo {author} {\bibfnamefont {I.}~\bibnamefont {{\v S}upi{\'c}}},\ and\ \bibinfo {author} {\bibfnamefont {A.}~\bibnamefont {Ac{\'i}n}},\ }\href {https://doi.org/10.1103/PhysRevLett.125.260507} {\bibfield  {journal} {\bibinfo  {journal} {Physical Review Letters}\ }\textbf {\bibinfo {volume} {125}},\ \bibinfo {pages} {260507} (\bibinfo {year} {2020})}\BibitemShut {NoStop}%
\bibitem [{\citenamefont {Zhu}\ \emph {et~al.}(2024)\citenamefont {Zhu}, \citenamefont {Li},\ and\ \citenamefont {Chen}}]{zhu2024efficient}%
  \BibitemOpen
  \bibfield  {author} {\bibinfo {author} {\bibfnamefont {H.}~\bibnamefont {Zhu}}, \bibinfo {author} {\bibfnamefont {Y.}~\bibnamefont {Li}},\ and\ \bibinfo {author} {\bibfnamefont {T.}~\bibnamefont {Chen}},\ }\href {https://doi.org/10.22331/q-2024-01-10-1221} {\bibfield  {journal} {\bibinfo  {journal} {Quantum}\ }\textbf {\bibinfo {volume} {8}},\ \bibinfo {pages} {1221} (\bibinfo {year} {2024})}\BibitemShut {NoStop}%
\bibitem [{\citenamefont {Chen}\ \emph {et~al.}(2023{\natexlab{b}})\citenamefont {Chen}, \citenamefont {Li},\ and\ \citenamefont {Zhu}}]{chen2023efficient}%
  \BibitemOpen
  \bibfield  {author} {\bibinfo {author} {\bibfnamefont {T.}~\bibnamefont {Chen}}, \bibinfo {author} {\bibfnamefont {Y.}~\bibnamefont {Li}},\ and\ \bibinfo {author} {\bibfnamefont {H.}~\bibnamefont {Zhu}},\ }\href {https://doi.org/10.1103/PhysRevA.107.022616} {\bibfield  {journal} {\bibinfo  {journal} {Physical Review A}\ }\textbf {\bibinfo {volume} {107}},\ \bibinfo {pages} {022616} (\bibinfo {year} {2023}{\natexlab{b}})}\BibitemShut {NoStop}%
\bibitem [{\citenamefont {Breuckmann}\ and\ \citenamefont {Eberhardt}(2021)}]{breuckmann2021quantum}%
  \BibitemOpen
  \bibfield  {author} {\bibinfo {author} {\bibfnamefont {N.~P.}\ \bibnamefont {Breuckmann}}\ and\ \bibinfo {author} {\bibfnamefont {J.~N.}\ \bibnamefont {Eberhardt}},\ }\href {https://doi.org/10.1103/PRXQuantum.2.040101} {\bibfield  {journal} {\bibinfo  {journal} {PRX Quantum}\ }\textbf {\bibinfo {volume} {2}},\ \bibinfo {pages} {040101} (\bibinfo {year} {2021})}\BibitemShut {NoStop}%
\bibitem [{\citenamefont {Schlingemann}(2001)}]{schlingemann2001stabilizer}%
  \BibitemOpen
  \bibfield  {author} {\bibinfo {author} {\bibfnamefont {D.}~\bibnamefont {Schlingemann}},\ }\href@noop {} {\bibinfo {title} {Stabilizer codes can be realized as graph codes}} (\bibinfo {year} {2001}),\ \Eprint {https://arxiv.org/abs/quant-ph/0111080} {arxiv:quant-ph/0111080} \BibitemShut {NoStop}%
\bibitem [{\citenamefont {Bell}\ \emph {et~al.}(2023)\citenamefont {Bell}, \citenamefont {Pettersson},\ and\ \citenamefont {Paesani}}]{bell2023optimizing}%
  \BibitemOpen
  \bibfield  {author} {\bibinfo {author} {\bibfnamefont {T.~J.}\ \bibnamefont {Bell}}, \bibinfo {author} {\bibfnamefont {L.~A.}\ \bibnamefont {Pettersson}},\ and\ \bibinfo {author} {\bibfnamefont {S.}~\bibnamefont {Paesani}},\ }\href {https://doi.org/10.1103/PRXQuantum.4.020328} {\bibfield  {journal} {\bibinfo  {journal} {PRX Quantum}\ }\textbf {\bibinfo {volume} {4}},\ \bibinfo {pages} {020328} (\bibinfo {year} {2023})}\BibitemShut {NoStop}%
\bibitem [{\citenamefont {Zhu}\ and\ \citenamefont {Hayashi}(2019{\natexlab{b}})}]{zhu2019efficientb}%
  \BibitemOpen
  \bibfield  {author} {\bibinfo {author} {\bibfnamefont {H.}~\bibnamefont {Zhu}}\ and\ \bibinfo {author} {\bibfnamefont {M.}~\bibnamefont {Hayashi}},\ }\href {https://doi.org/10.1103/PhysRevApplied.12.054047} {\bibfield  {journal} {\bibinfo  {journal} {Physical Review Applied}\ }\textbf {\bibinfo {volume} {12}},\ \bibinfo {pages} {054047} (\bibinfo {year} {2019}{\natexlab{b}})}\BibitemShut {NoStop}%
\bibitem [{\citenamefont {Cafaro}\ \emph {et~al.}(2014)\citenamefont {Cafaro}, \citenamefont {Markham},\ and\ \citenamefont {{van Loock}}}]{cafaro2014schemea}%
  \BibitemOpen
  \bibfield  {author} {\bibinfo {author} {\bibfnamefont {C.}~\bibnamefont {Cafaro}}, \bibinfo {author} {\bibfnamefont {D.}~\bibnamefont {Markham}},\ and\ \bibinfo {author} {\bibfnamefont {P.}~\bibnamefont {{van Loock}}},\ }\href@noop {} {\bibinfo {title} {Scheme for constructing graphs associated with stabilizer quantum codes}} (\bibinfo {year} {2014}),\ \Eprint {https://arxiv.org/abs/1407.2777} {arXiv:1407.2777 [quant-ph]} \BibitemShut {NoStop}%
\bibitem [{\citenamefont {Scheinerman}\ and\ \citenamefont {Ullman}(2013)}]{scheinerman2013fractional}%
  \BibitemOpen
  \bibfield  {author} {\bibinfo {author} {\bibfnamefont {E.~R.}\ \bibnamefont {Scheinerman}}\ and\ \bibinfo {author} {\bibfnamefont {D.~H.}\ \bibnamefont {Ullman}},\ }\href@noop {} {\emph {\bibinfo {title} {Fractional graph theory: a rational approach to the theory of graphs}}}\ (\bibinfo  {publisher} {Courier Corporation},\ \bibinfo {year} {2013})\BibitemShut {NoStop}%
\bibitem [{\citenamefont {Nielsen}\ and\ \citenamefont {Chuang}(2010)}]{Nielsen_Chuang_2010}%
  \BibitemOpen
  \bibfield  {author} {\bibinfo {author} {\bibfnamefont {M.~A.}\ \bibnamefont {Nielsen}}\ and\ \bibinfo {author} {\bibfnamefont {I.~L.}\ \bibnamefont {Chuang}},\ }\href@noop {} {\emph {\bibinfo {title} {Quantum Computation and Quantum Information: 10th Anniversary Edition}}}\ (\bibinfo  {publisher} {Cambridge University Press},\ \bibinfo {year} {2010})\BibitemShut {NoStop}%
\bibitem [{\citenamefont {Haah}(2011)}]{haah2011local}%
  \BibitemOpen
  \bibfield  {author} {\bibinfo {author} {\bibfnamefont {J.}~\bibnamefont {Haah}},\ }\href {https://doi.org/10.1103/PhysRevA.83.042330} {\bibfield  {journal} {\bibinfo  {journal} {Physical Review A}\ }\textbf {\bibinfo {volume} {83}},\ \bibinfo {pages} {042330} (\bibinfo {year} {2011})}\BibitemShut {NoStop}%
\bibitem [{\citenamefont {MacKay}\ \emph {et~al.}(2004)\citenamefont {MacKay}, \citenamefont {Mitchison},\ and\ \citenamefont {McFadden}}]{mackay2004sparse}%
  \BibitemOpen
  \bibfield  {author} {\bibinfo {author} {\bibfnamefont {D.~J.~C.}\ \bibnamefont {MacKay}}, \bibinfo {author} {\bibfnamefont {G.}~\bibnamefont {Mitchison}},\ and\ \bibinfo {author} {\bibfnamefont {P.~L.}\ \bibnamefont {McFadden}},\ }\href {https://doi.org/10.1109/TIT.2004.834737} {\bibfield  {journal} {\bibinfo  {journal} {IEEE Transactions on Information Theory}\ }\textbf {\bibinfo {volume} {50}},\ \bibinfo {pages} {2315} (\bibinfo {year} {2004})}\BibitemShut {NoStop}%
\bibitem [{\citenamefont {Liu}\ \emph {et~al.}(2021)\citenamefont {Liu}, \citenamefont {Shang}, \citenamefont {Han},\ and\ \citenamefont {Zhang}}]{liu2021universallya}%
  \BibitemOpen
  \bibfield  {author} {\bibinfo {author} {\bibfnamefont {Y.-C.}\ \bibnamefont {Liu}}, \bibinfo {author} {\bibfnamefont {J.}~\bibnamefont {Shang}}, \bibinfo {author} {\bibfnamefont {R.}~\bibnamefont {Han}},\ and\ \bibinfo {author} {\bibfnamefont {X.}~\bibnamefont {Zhang}},\ }\href {https://doi.org/10.1103/PhysRevLett.126.090504} {\bibfield  {journal} {\bibinfo  {journal} {Physical Review Letters}\ }\textbf {\bibinfo {volume} {126}},\ \bibinfo {pages} {090504} (\bibinfo {year} {2021})}\BibitemShut {NoStop}%
\bibitem [{\citenamefont {Chen}\ \emph {et~al.}(2024)\citenamefont {Chen}, \citenamefont {Zeng}, \citenamefont {Zhao}, \citenamefont {Ma},\ and\ \citenamefont {Zhou}}]{chen2024quantuma}%
  \BibitemOpen
  \bibfield  {author} {\bibinfo {author} {\bibfnamefont {J.}~\bibnamefont {Chen}}, \bibinfo {author} {\bibfnamefont {P.}~\bibnamefont {Zeng}}, \bibinfo {author} {\bibfnamefont {Q.}~\bibnamefont {Zhao}}, \bibinfo {author} {\bibfnamefont {X.}~\bibnamefont {Ma}},\ and\ \bibinfo {author} {\bibfnamefont {Y.}~\bibnamefont {Zhou}},\ }\href@noop {} {\bibinfo {title} {Quantum subspace verification for error correction codes}} (\bibinfo {year} {2024}),\ \Eprint {https://arxiv.org/abs/2410.12551} {arXiv:2410.12551} \BibitemShut {NoStop}%
\end{thebibliography}
\end{document}